\documentclass[11pt]{article}

\usepackage{latexsym}
\usepackage{amsthm}
\usepackage{amsmath,amssymb,amsfonts,bm,amsbsy,bbm,mathabx}
\usepackage{epsfig}
\usepackage{graphicx,rotating,lscape}
\usepackage{verbatim}
\usepackage{multirow,color}
\usepackage{geometry}
\usepackage{setspace}
\usepackage{float}
\usepackage{physics}
\usepackage{mathtools}
\usepackage{enumitem}
\usepackage{caption, subcaption}

\usepackage{algorithm}
\usepackage{algorithmic}
\usepackage[algo2e]{algorithm2e}

\usepackage{natbib}

\newcommand{\E}{\mathrm{E}}

\newcommand{\Ep}{\mathrm{E}_p}
\newcommand{\Eq}{\mathrm{E}_q}

\newtheorem{lemma}{\underline{\bf Lemma}}
\newtheorem{example}{\underline{\bf Example}}
\newtheorem{pro}{\underline{\bf Proposition}}
\newtheorem{Th}{\underline{\bf Theorem}}
\newtheorem{Rem}{\underline{\bf Remark}}
\newtheorem{Cor}{\underline{\bf Corollary}}

\def\bse{\begin{eqnarray*}}
\def\ese{\end{eqnarray*}}
\def\be{\begin{eqnarray}}
\def\ee{\end{eqnarray}}
\def\bsq{\begin{equation*}}
\def\esq{\end{equation*}}
\def\bq{\begin{equation}}
\def\eq{\end{equation}}
\def\bi{\begin{itemize}}
\def\ei{\end{itemize}}

\def\sumi{\sum_{i=1}^n}
\def\sumj{\sum_{j=1}^n}

\def\wh{\widehat}
\def\wt{\widetilde}

\def\pr{\hbox{pr}}
\def\calF{{\cal F}}
\def\calL{{\cal L}}

\def\calT{{\cal T}}

\def\trans{^{\rm T}}
\def\eff{_{\rm eff}}

\def\n{\nonumber}

\def\bb{{\boldsymbol\beta}}
\def\bt{{\boldsymbol\theta}}
\def\bT{{\boldsymbol\Theta}}
\def\bg{{\boldsymbol\gamma}}

\def\bSigma{{\boldsymbol\Sigma}}

\def\ba{{\boldsymbol\alpha}}
\def\bphi{{\boldsymbol\phi}}
\def\bPhi{{\boldsymbol\Phi}}

\def\bd{{\boldsymbol\delta}}

\def\bzeta{{\boldsymbol\zeta}}

\def\0{{\bf 0}}
\def\X{{\bf X}}
\def\x{{\bf x}}

\def\c{{\bf c}}

\def\w{{\bf w}}
\def\U{{\bf U}}

\def\V{{\bf V}}
\def\v{{\bf v}}
\def\S{{\bf S}}

\def\B{{\bf B}}
\def\b{{\bf b}}
\def\A{{\bf A}}
\def\a{{\bf a}}

\def\g{{\bf g}}

\def\G{{\bf G}}

\def\I{{\bf I}}

\def\t{{\bf t}}
\def\w{{\bf w}}
\def\y{{\bf y}}
\def\a{{\bf a}}

\def\boxit#1{\vbox{\hrule\hbox{\vrule\kern6pt\vbox{\kern6pt#1\kern6pt}\kern6pt\vrule}\hrule}}
\def\macomment#1{\vskip 2mm\boxit{\vskip 2mm{\color{red}\bf#1} {\color{blue}\bf -- MA\vskip 2mm}}\vskip 2mm}

\def\px{p_{\X}}
\def\py{p_{Y}}
\def\pyx{p_{Y\mid\X}}
\def\pxy{p_{\X\mid Y}}

\def\qx{q_{\X}}
\def\qy{q_{Y}}
\def\qyx{q_{Y\mid\X}}
\def\qxy{q_{\X\mid Y}}

\def\mp{{\mathcal{P}}}
\def\mq{{\mathcal{Q}}}

\def\var{\hbox{var}}

\makeatletter
\def\widebreve{\mathpalette\wide@breve}
\def\wide@breve#1#2{\sbox\z@{$#1#2$}%
     \mathop{\vbox{\m@th\ialign{##\crcr
\kern0.08em\brevefill#1{0.8\wd\z@}\crcr\noalign{\nointerlineskip}%
                    $\hss#1#2\hss$\crcr}}}\limits}
\def\brevefill#1#2{$\m@th\sbox\tw@{$#1($}%
  \hss\resizebox{#2}{\wd\tw@}{\rotatebox[origin=c]{90}{\upshape(}}\hss$}
\makeatletter

\begin{document}

\baselineskip 15 pt

\title{Doubly Flexible Estimation under Label Shift}

\author{
Seong-ho Lee\thanks{Penn State University}~~~~~
Yanyuan Ma\thanks{Penn State University}~~~~~
Jiwei Zhao\thanks{University of Wisconsin-Madison; e-mail: \texttt{jiwei.zhao@wisc.edu}}
}

\maketitle
\thispagestyle{empty}

\newpage
\setcounter{page}{1} 
\thispagestyle{empty}

\begin{center}
{\Large Doubly Flexible Estimation under Label Shift}
\end{center}

\begin{abstract}
\noindent
In studies ranging from clinical medicine to policy research, complete
data  are usually available from a population $\mp$,
but the quantity of interest is often sought for  a related but
different population $\mq$ which only has partial data.
In this paper, we consider the setting that both outcome $Y$ and
covariate $\X$ are available from $\mp$ whereas only $\X$ is available
from $\mq$, under the so-called label shift assumption, i.e., the
conditional distribution of $\X$ given $Y$ remains the same across the two populations.
To estimate the parameter of interest in population $\mq$ via
leveraging the information from population $\mp$, the following three
ingredients are essential:
(a) the common conditional distribution of $\X$ given $Y$, (b) the
regression model of $Y$ given $\X$ in population $\mp$, and (c) the
density ratio of the outcome $Y$ between the two populations.
We propose an estimation procedure that only needs some standard
nonparametric regression technique to approximate the conditional
expectations with respect to (a),
while by no means needs an estimate or model for
(b) or (c); i.e., doubly flexible to the possible model
misspecifications of both (b) and (c).
This is conceptually different from the well-known doubly robust
estimation in that, double robustness allows at most one model to be
misspecified whereas our proposal here can allow both (b)
and (c) to be misspecified.
This is of particular interest in our setting because estimating (c)
is difficult, if not impossible, by virtue of the absence of the
$Y$-data in population $\mq$.
Furthermore, even though the estimation of (b) is sometimes
off-the-shelf, it can face curse of dimensionality or
  computational challenges.
We develop the large sample theory for the proposed estimator, and
examine its finite-sample performance through simulation studies as
well as an application to the MIMIC-III database.
\end{abstract}


\newpage

\setcounter{equation}{0}

\section{Introduction}

In studies ranging from clinical medicine to policy research, there
often exist data and information from a population $\mp$, while the
quantity of interest is defined on a particular target
$\mq$, relevant but different from $\mp$.
For instance, in a clinical trial setting, physicians may be left
interpreting evidence from a randomized controlled trial consisting of patients
who have demographics and comorbidities that are quite different from
those of their own patients (population $\mq$).
As another example, to build a predictive model on pneumonia outbreak
for the flu season (population $\mq$), researchers might find a
similar model during the non-flu season relevant and useful.
In these scenarios, there is a discrepancy between the distributions
of $\mp$ and $\mq$, termed distribution shift throughout.
Distribution shift
can also refer to the fact that the distribution of the training
sample is different from that of the testing sample, in the evaluation of a learning algorithm.

In all of these situations, it is of vital interest to propose methods
that can appropriately leverage the information from $\mp$ into the
statistical tasks for $\mq$.
Our methodology will use the information from both outcome (output,
response, label) $Y$ and covariate (input, predictor, feature) $\X$ in
population $\mp$ as well as covariate $\X$ in population
$\mq$.
This setting is also named unsupervised domain adaptation \citep{quinonero2008dataset, moreno2012unifying, kouw2019review}.

Without any assumptions on the nature of shift, it is certainly impossible to
leverage information between two heterogeneous populations.
Two major types of distribution shifts have been defined in the literature.
The first is called  covariate shift where the shift
happens between the marginal distributions of $\X$ while the
conditional distribution of $Y$ given $\X$ does not change; i.e., $\px(\x)\neq \qx(\x)$ and
$\pyx(y,\x)=\qyx(y, \x)$.
The difference between $\mp$ and $\mq$ can be summarized as a
density ratio $\qx(\x)/\px(\x)$, which is, fortunately, estimable
since covariate $\X$ is available from both populations.
Covariate shift aligns
with the causal learning setting \citep{scholkopf2012causal} where $\X$ is the cause and $Y$ is
the effect.
Covariate shift has attracted a great deal of attention and has been
investigated in many literatures, such as
\cite{shimodaira2000improving, huang2007correcting,
  sugiyama2008direct, gretton2009covariate, sugiyama2012machine,
  kpotufe2021marginal} and the references therein.

The second type, which is the focus of this paper, is named label
shift, because it assumes that the shift is induced by the marginal
distributions of $Y$ while the
process generating $\X$ given $Y$ is identical in both populations.
Formally, it assumes
\bse
\py(y)\neq \qy(y), \mbox{ and } \pxy(\x,y)=\qxy(\x,y)\equiv g(\x,y).
\ese
Label shift is also called prior probability shift
\citep{storkey2009training, tasche2017fisher}, target shift
\citep{zhang2013domain, nguyen2016continuous}, or class prior change
\citep{du2014semi, iyer2014maximum}.
Label shift aligns with the anticausal learning setting in which the
outcome/label $Y$ causes the covariate/feature $\X$; for example, diseases
cause symptoms or objects cause sensory observations.
Consider the situation that one fits a model to predict whether a
patient has pneumonia based on observed symptoms, and that this model
predicts reliably when deployed in the clinic during the
non-flu season.
When the flu season starts, there is a sudden surge of
pneumonia cases and the probability of patients developing pneumonia given
that they show symptoms rises, while the mechanism of showing
symptoms of pneumonia is rather stable.
Label shift also exists in many computer vision applications, such as
predicting object locations and directions, and human poses; see
\cite{martinez2017simple, yang2018position, guo2020ltf}.

In the label shift framework, one fundamental problem
\citep{garg2020unified} is determining whether the shift
has occurred and estimating the label distribution $\qy(y)$, or
equivalently, assessing the density ratio
$\qy(y)/\py(y)\equiv \rho(y)$.
In contrast to estimating the density ratio $\qx(\x)/\px(\x)$ under
covariate shift, estimating
$\rho(y)$ is a daunting task due to the absence of the $Y$ observations in population $\mq$.
Works in the label shift framework are mainly limited to the
classification problems in the machine learning literature.
\cite{saerens2002adjusting} proposed a simple Expectation-Maximization
(EM) \citep{dempster1977maximum} procedure, named maximum likelihood
label shift (MLLS), to estimate $\qy(y)$ assuming access to a
classifier that outputs the true conditional probabilities of the
population $\mp$, $\pyx(y,\x)$.
Later on, \cite{chan2005word} proposed an EM algorithm that requires
the estimation of $g(\x,y)$, which is unfortunately difficult for
high-dimensional $\X$ and moreover, it does not apply to regression
problems.
Alternatively, \cite{lipton2018detecting} and
\cite{azizzadenesheli2019regularized} proposed moment-matching based
estimators, named black box shift learning (BBSL) and regularized
learning under label shift (RLLS), that make use of the invertible
confusion matrix of a classifier learned from population $\mp$.
The connection and comparison of these two lines of research, either
empirical or theoretical, remain unclear.
To our best knowledge, neither BBSL nor RLLS has been benchmarked against EM.
\cite{alexandari2020} showed that, in combination with a calibration
named bias-corrected temperature scaling, MLLS outperforms BBSL and
RLLS empirically; whereas MLLS underperforms BBSL when applied naively.
Under label shift, \cite{maity2020minimax} also studied the minimax
rate of convergence for nonparametric classification.

For continuous $Y$ in regression problems, estimating $\qy(y)$ becomes
the problem of estimating a function instead of a finite number of
parameters. Not surprisingly, its literature is quite scarce.
\cite{zhang2013domain} proposed a nonparametric method to estimate the
density ratio by kernel mean matching of distributions. However, this approach
does not scale to large data as the computational cost is quadratic in the sample size.
\cite{nguyen2016continuous} considered continuous label shift
adaptation and studied an importance weight estimator, but their
approach relies on a parametric model for $\pyx(y,\x)$ hence can only be applied in supervised learning.

In this paper, we take a completely distinct approach from all of the existing literature.
Different from the current majority, our methodology is devised to accommodate both classification and regression.
Whether the outcome $Y$ is discrete or continuous is not essential in our proposal.
We directly estimate a characteristic of the population $\mq$.
Specifically, we estimate the parameter $\bt$ such that $\Eq\{\U(\X,Y,\bt)\}=\0$ where $\U(\cdot)$ is a user specified function and $\Eq(\cdot)$ stands for the expectation
with respect to $\qy(y)g(\x,y)$ or equivalently to $\qyx(y,\x)\qx(\x)$.
This is a general framework, including estimating the mean of $Y$ or the $t$-th quantile of $Y$ as special cases.
According to how the nuisance components are estimated, detailed in the next three paragraphs, we propose various estimators for $\bt$, and develop large sample theory for these estimators to quantify the estimation uncertainties and to conduct statistical inference.

To estimate $\bt$, three nuisance components are involved.
First and foremost is the density ratio $\rho(y)$, which is almost
infeasible to estimate based on the observed data due to the lack of
$Y$-observations in population $\mq$.
Our intention is to bypass the challenging task of estimating $\rho(y)$.
This turns out achievable through careful manipulation of other
components of the influence function.  In fact, a unique feature of
our work is that, we do not need
to estimate $\rho(y)$ throughout the estimation procedure. Instead, only a working model, denoted
as $\rho^*(y)$, is needed.

The second one is $\pyx(y,\x)$, or  some dependent
quantities such as $\Ep(\cdot\mid \x)$. In contrast to
$\rho(y)$, estimating $\Ep(\cdot\mid \x)$ is blessed with the observed data in
population $\cal P$.
Indeed, we can use off-the-shelf machine learning methods or
nonparametric regression methods to obtain the corresponding estimator
$\wh \E_p(\cdot\mid \x)$.
Nonetheless, we can also choose to give up estimating
$\Ep(\cdot\mid\x)$ even though we can do it. This means that we can
misspecify the conditional distribution
$\pyx(y,\x)$ while we also misspecify the density ratio $\rho(y)$.
We call such an estimator $\wh\bt$ doubly flexible---the working
density ratio model $\rho^*(y)$ is flexible, so is the working
conditional distribution model $\pyx^\star(y,\x)$.
Note that our superscripts here are different: superscript $^*$ stands for the working model of the density ratio whereas $^\star$ is for the conditional distribution model.
This doubly flexible property is much more favorable than the classic ``doubly
robust'' in the literature.
The standard double robustness means that one can misspecify either
one of two models but not both, while here, we can misspecify both
models.
As an alternative, if one chooses to estimate
$\Ep(\cdot\mid\x)$, say, $\wh\E_p(\cdot\mid\x)$, we name the corresponding estimator $\wt\bt$
singly flexible---only flexible in working model $\rho^*(y)$.

The third nuisance is the conditional density function $g(\x,y)$, whose estimation might be subject to the curse of dimensionality.
Fortunately, in our estimation procedure, $g(\x,y)$ only affects
quantities of the form $\E(\cdot\mid y)$, which are  one dimensional
regression problems hence can be easily solved via the most basic
nonparametric regression procedure such as the Nadaraya-Watson
estimation.

The remaining of the paper is structured as follows.
In Section~\ref{sec:prelim}, we first outline our strategy of how to
incorporate samples from two heterogeneous populations.
The proposed doubly flexible estimator is presented in
Section~\ref{sec:theta}, and the alternative singly flexible estimator
is contained in Section~\ref{sec:singleflex}.
For easier understanding and improved readability,
we present both the methodology and the theory for a special
parameter $\theta=\Eq(Y)$ in the main text, while defer the general
results for $\bt$ such that $\Eq\{\U(\X,Y,\bt)\}=\0$ to the Supplement.
Section~\ref{sec:sim} contains empirical results for extensive
simulation studies. We  present an application to the MIMIC-III database in Section  \ref{sec:data}.
The paper is concluded with discussions in Section~\ref{sec:disc}.
All the technical details are also included in the Supplement.

\section{Model Structure}\label{sec:prelim}

We consider independent and identically distributed (iid) observations
$\{Y_i,\X_i\}, i=1,\ldots,n_1$ from population $\mp$, and iid
observations $\X_j, j=n_1+1, \ldots, n_1+n_0=n$ from population
$\mq$.
To use the information in population $\mp$ under label shift, we
stack the two random samples together and assemble a new data set of
size $n$, which represents a random sample for an imaginary population
consisting of $100\pi\%$ population $\mp$ members and $100(1-\pi)\%$
population $\mq$ members.
Here we define $\pi\equiv n_1/n$.
Throughout our derivation, other than $\Ep(\cdot)$ and $\Eq(\cdot)$,
we also compute $\E(\cdot)$ that is with respect to this imaginary
population; however, this imaginary population is only used as an
intermediate tool to leverage information from two heterogeneous
populations under label shift.
Our final conclusion will only be made for the target population
$\mq$.

For convenience, we introduce a binary indicator $R$ in this stacked
random sample, where $R=1$ means the subject is from population $\mp$
and $R=0$ population $\mq$.
Thus, the likelihood of one observation from the stacked random sample is
\be
&& \{g(\x, y)\py(y)\}^r \left\{\int g(\x, y)\qy(y)\dd
  y\right\}^{1-r} \pi^r(1-\pi)^{1-r}\label{eq:likelihood}\\
&=& \{g(\x, y)\py(y)\}^r \left\{\int g(\x, y)\rho(y)\py(y)\dd
  y\right\}^{1-r} \pi^r(1-\pi)^{1-r}.\label{eq:likelihoodwithrho}
\ee
Although $g(\x,y)$ and $\py(y)$ can be identified from
(\ref{eq:likelihood}), unfortunately $\qy(y)$ may not.
Below is a simple example illustrating the possible nonidentifiability of $\qy(y)$.
\begin{example}
  Consider a discrete $Y$ with three supporting values $0,1,2$ and a
  discrete $X$ with two supporting values $0,1$. In both
  populations $\mp$ and $\mq$, $g(x,y)$ is given as
  \bse
  \pr(X=0\mid Y=0)=1/5, \pr(X=0\mid Y=1)=1/8, \mbox{ and } \pr(X=0\mid Y=2)=2/3.
  \ese
  In population $\mp$, the marginal distribution of $Y$, $\py(y)$,  is given as
  \bse
  \pr(Y=0)=5/16, \pr(Y=1)=1/2, \mbox{ and } \pr(Y=2)=3/16.
  \ese
  In population $\mq$, the marginal distribution of $Y$, $\qy(y)$, is given as
  \bse
 \pr(Y=0)=\frac{5(25-416t)}{336}, \pr(Y=1)=\frac{32t-1}{6}, \mbox{ and } \pr(Y=2)=\frac{89+96t}{112},
  \ese
  where $t\in(1/32, 25/336)$.
  Clearly this satisfies the label shift assumption, and the marginal distribution of $X$ in population $\mq$ is identifiable since $\pr(X=0)=7/12$ is free of $t$; however, the marginal distribution of $Y$ in population $\mq$, $\qy(y)$, is not identifiable.
\end{example}
The following result demonstrates that the completeness condition on
$\pyx(y,\x)$ would ensure the identifiability.
Its proof is contained in Section~\ref{sec:slemmaproof} of the Supplement.

\begin{lemma}\label{lem:iden}
If the conditional pdf/pmf $\pyx(y, \x)$ of population $\mp$ satisfies
the completeness condition in the sense that, for any function $h(Y)$
with finite mean, $\Ep\{h(Y)\mid\X\}=\int h(y)\pyx(y, \x)\dd y=0$
implies $h(Y)=0$ almost surely, then all the unknown components in
(\ref{eq:likelihoodwithrho}), i.e., $g(\x, y)$, $\py(y)$ and $\rho(y)$, are
 identifiable. Subsequently, $\qy(y)$ is also identifiable.
\end{lemma}
The completeness condition in Lemma~\ref{lem:iden} is mild and has
been widely assumed in instrumental variables, measurement error
models, and econometrics; see, e.g, \cite{newey2003instrumental,
  d2011completeness, hu2018nonparametric}.
Because the condition is imposed on population $\mp$ and we have random
observations $(\X_i, Y_i)$'s from population $\mp$, it
can be examined and verified in empirical studies.
One can easily check that many commonly-used distributions such as exponential
families satisfy the completeness condition.
In particular, if the outcome $Y$ is discrete with finitely many supporting
values, \cite{newey2003instrumental} pointed out that the completeness
condition only means that the covariate $\X$ has a
support whose cardinality is no smaller than that of $Y$.

In this paper, we focus on estimating a characteristic of
population $\mq$. For better clarity,
we will present the results for
$\theta=\Eq(Y)=\int y g(\x,y)\qy(y)\dd\x\dd y$ in the main text, then generalize
the results to $\bt$ that satisfies $\Eq\{\U(Y,\X,\bt)\}=\0$ in the Supplement.
The main challenge in estimating $\theta$ is caused by the lack of
knowledge and data on $\qy(y)$,
or equivalently, $\rho(y)$.
Nevertheless, we will construct an
estimator that bypasses the difficulty of assessing $\rho(y)$. We will
show that we only need a working model of $\rho(y)$, denoted as
$\rho^*(y)$, that can be flexible.
Furthermore, we find that our procedure can also simultaneously
avoid estimating $\pyx(y, \x)$, in that we can insert a possibly
misspecified working model $\pyx^\star(y, \x)$. Thus, our procedure
is flexible with respect to both $\rho(y)$ and $\pyx(y, \x)$---doubly flexible.
This is a property different from the classic ``double
robustness'' which means that one
can only misspecify one of two models but not both. In contrast, here, we can misspecify
both.

\section{Proposed Doubly Flexible Estimation for
  $\theta=\Eq(Y)$}\label{sec:theta}

If the density ratio function $\rho(y)$ were known, an intuitive
estimator of $\theta=\Eq(Y)$ can be created by noticing the relation
$\theta=\Eq(Y)=\Ep\{\rho(Y)Y\}=\E\{R\rho(Y)Y\}/\pi$; that is,
\be\label{eq:shift-dependent}
\widecheck\theta = \frac1n\sumi \frac{r_i}{\pi}\rho(y_i)y_i.
\ee
We call this estimator shift-dependent since it requires the correct specification of $\rho(y)$.
Clearly, if a working model $\rho^*(y)$ is adopted, the corresponding
estimator $\widecheck\theta^*$ is likely biased.

\subsection{General Approach to Estimating $\theta$}\label{sec:eif}

The creation of an estimator that is not solely shift-dependent is possible.
To motivate our proposed estimator, we first make some simple observations via balancing the samples from populations $\mp$ and $\mq$.
Recognizing the relation between $\Ep(\cdot)$, $\Eq(\cdot)$ and $\E(\cdot)$, the balancing of $Y$ is
\bse
\E\left\{\frac{R}{\pi}\rho(Y)Y\right\} = \Ep\{\rho(Y)Y\} = \Eq(Y)=
\E\left(\frac{1-R}{1-\pi}\theta\right).
\ese

Further, replacing the variable $Y$ above by an arbitrary function of
$\X$, we obtain another balancing function
\bse
\E\left\{\frac{R}{\pi}\rho(Y)b(\X)\right\}= \E\left\{\frac{1-R}{1-\pi}b(\X)\right\}.
\ese
Certainly, we also have
\bse
\E\left(\frac{R}{\pi}c\right) = \E\left(\frac{1-R}{1-\pi}c\right)
\ese
for any constant $c$.
Combining the above three, we can obtain a family of mean zero
functions
\be\label{eq:meanzero}
\frac{r}{\pi}\{\rho(y)y-b(\x)\rho(y)+c\}+\frac{1-r}{1-\pi}\{b(\x)-\theta-c\}:
  \forall b(\x), \forall c.
\ee

Note that the model in (\ref{eq:likelihoodwithrho}) contains three unknown functions $\py(y)$, $g(\x,y)$ and $\rho(y)$.
For this model, in Section
\ref{sec:influence} of the Supplement, we establish that
\bse
\calF
\equiv\left[\frac{r}{\pi}\{\rho(y)y-b(\x)\rho(y)+c\}+\frac{1-r}{1-\pi}\{b(\x)-\theta-c\}:
  \E\{b(\X)\mid y\}=y, \forall c \right]
\ese
is the family of all influence functions \citep{bickel1993efficient,
  tsiatis2006semiparametric} for estimating $\theta$.
According to the definition of the influence function, $\calF$ is
sufficiently comprehensive since it can generate any regular
asymptotically linear estimator of
$\theta$.
The requirement $\E\{b(\X)\mid y\}=y$ in the definition of $\calF$ is pivotal.
Different from the mean zero function in (\ref{eq:meanzero}), which
critically relies on the correct specification of $\rho(y)$, the
element in $\calF$ preserves its zero mean even if $\rho(y)$ is
misspecified as long as an appropriate $b(\x)$ is chosen so that
$\E\{b(\X)\mid y\}=y$.
To further discover a wise choice of such a $b(\x)$, we first derive a
special element in $\calF$, the efficient influence function
$\phi\eff(\x,r,ry)$, that corresponds to the semiparametric
efficiency bound and that provides guidance on constructing flexible
estimators for $\theta$.
\begin{pro}\label{pro:eff}
The efficient influence function $\phi\eff(\x,r,ry)$ for $\theta$ is
\bse
\phi\eff(\x,r,ry)
&=&\frac{r}{\pi}\rho(y)\left[y
-\frac{\Ep\{a(Y)\rho(Y)\mid\x\}}{\Ep\{\rho^2(Y)\mid\x\}+\pi/(1-\pi)\Ep\{\rho(Y)\mid\x\}}\right]\\
&&+\frac{1-r}{1-\pi}
\left[\frac{\Ep\{a(Y)\rho(Y)\mid\x\}}{\Ep\{\rho^2(Y)\mid\x\}+\pi/(1-\pi)\Ep\{\rho(Y)\mid\x\}}-\theta\right],
\ese
where $a(y)$ satisfies
\be\label{eq:a}
\E\left[\frac{\Ep\{a(Y)\rho(Y)\mid\X\}}{\Ep\{\rho^2(Y)\mid\X\}+\pi/(1-\pi)\Ep\{\rho(Y)\mid\X\}}\mid y\right]
=y.
\ee
\end{pro}
The detailed derivation of the efficient influence function in
Proposition \ref{pro:eff} is provided in Section \ref{sec:eff} of
the Supplement.
Clearly, the unique $b(\x)$ that leads to
the efficient influence function is
\bse
b(\x)\equiv
\frac{\Ep\{a(Y)\rho(Y)\mid\x\}}{\Ep\{\rho^2(Y)\mid\x\}
+\pi/(1-\pi)\Ep\{\rho(Y)\mid\x\}}
= \frac{\Eq\{a(Y)\mid\X\}}{\Eq\{\rho(Y)\mid\X\}+\pi/(1-\pi)}.
\ese
In principle, if both $\rho(y)$ and $b(\x)$ were known, we can estimate
$\theta$ by solving the estimating equation\\
$\sumi\phi\eff(\x_i,r_i,r_iy_i)=0$, which leads to
\be\label{eq:thetahat}
\widebreve\theta=
\frac{1}{n}\sumi\left[\frac{r_i}{\pi}\rho(y_i)\{y_i-b(\x_i)\}
+\frac{1-r_i}{1-\pi}b(\x_i)\right].
\ee
However, the estimator $\widebreve\theta$ is impractical because
of the following three obstacles.
First, as we pointed out, $\rho(y)$ is almost infeasible to estimate based on the observed data.
Second, $\Ep(\cdot\mid\x)$ is unknown and needs to be estimated.
Though various off-the-shelf machine learning or nonparametric
regression methods are available, when
the dimension of $\x$ is high, their performances are not always
satisfactory and their computation can be expensive.
The third obstacle lies in solving $a(y)$ from the integral equation
(\ref{eq:a}), which requires $g(\x,y)$ to evaluate its left hand side.
Estimating conditional density $g(\x,y)$ could be even more
difficult than estimating $\Ep(\cdot\mid\x)$, due to the curse of dimensionality.

Our proposed estimator will bypass the challenging task of estimating $\rho(y)$.
Throughout the estimation procedure, only a working model $\rho^*(y)$
is needed, which can be arbitrarily misspecified hence is flexible.
This turns out achievable through careful manipulation of other
components of the efficient influence function.
Our proposed estimator can also avoid estimating $\Ep(\cdot\mid\x)$
even though we can do it if we decide to.
This means that we can misspecify the conditional density model
$\pyx(y,\x)$, encoded as $\pyx^\star(y,\x)$, while we also misspecify
the density ratio $\rho(y)$. We call such an estimation procedure doubly flexible.
To overcome the third obstacle, we recognize that $g(\x,y)$ only
affects quantities of the form
$\E(\cdot\mid y)$, which are one dimensional regression problems hence
can be easily solved via the most basic nonparametric regression
procedure such as the Nadaraya-Watson estimator.

In a nutshell, a unique feature of our work is the tolerance of
both $\rho^*(y)$ and
$\pyx^\star(y,\x)$, which can be simultaneously misspecified.
We thus name the procedure doubly flexible.

\subsection{Proposed Estimator $\wh\theta$: Doubly Flexible in $\rho^*(y)$ and $\pyx^\star(y,\x)$}\label{sec:doubleflex}

Interestingly and critically, we discover that, even when both
$\rho^*(y)$ and $\pyx^\star(y,\x)$ are misspecified, the corresponding
estimator following the implementation of $\widebreve\theta$ in
(\ref{eq:thetahat}) is still consistent for $\theta$.
We summarize this result in Proposition~\ref{pro:consistency} and give
its proof in Section~\ref{sec:consistency} of the Supplement.
Below, we use superscripts $^*$ and $^\star$ to indicate that the
corresponding quantities are calculated based on the working models
$\rho^*(y)$ and $\pyx^\star(y,\x)$ respectively.
\begin{pro}\label{pro:consistency}
Define
\bse
\wh\theta_t=
\frac{1}{n}\sumi\left[\frac{r_i}{\pi}\rho^*(y_i)\{y_i-b^{*\star}(\x_i)\}
+\frac{1-r_i}{1-\pi}b^{*\star}(\x_i)\right],
\ese
where
\bse
b^{*\star}(\x)\equiv
\frac{\Ep^\star\{a^{*\star}(Y)\rho^*(Y)\mid\x\}}{\Ep^\star\{\rho^{*2}(Y)\mid\x\}
+\pi/(1-\pi)\Ep^\star\{\rho^*(Y)\mid\x\}},
\ese
and $a^{*\star}(y)$ is a solution to
\be\label{eq:a*star}
\E\left[\frac{\Ep^\star\{a^{*\star}(Y)\rho^*(Y)\mid\X\}}
{\Ep^\star\{\rho^{*2}(Y)\mid\X\}
+\pi/(1-\pi)\Ep^\star\{\rho^*(Y)\mid\X\}}\mid y\right]
=y.
\ee
Then $\wh\theta_t$  is a consistent estimator of $\theta$.
\end{pro}

In Proposition~\ref{pro:consistency}, the subscript $_t$ in $\wh\theta_t$ indicates the conditional density $g(\x,y)$ in (\ref{eq:a*star}) is the truth.
In reality, note that $g(\x,y)$ only involves in the evaluation of  the conditional expectation $\E(\cdot\mid y)$ on the left hand side of (\ref{eq:a*star}).
This is a one dimensional regression problem and  can be easily estimated
by many  basic nonparametric regression procedures such as the
Nadaraya-Watson estimator. Specifically,
we approximate the integral equation (\ref{eq:a*star}) by
\be\label{eq:a*hat}
y&=&\wh\E\left[\frac{\Ep^\star\{a^{*\star}(Y)\rho^*(Y)\mid\X\}}
{\Ep^\star\{\rho^{*2}(Y)\mid\X\}+\pi/(1-\pi)\Ep^\star\{\rho^*(Y)\mid\X\}}\mid y\right]\n\\
&=&\sumi\frac{\Ep^\star\{a^{*\star}(Y)\rho^*(Y)\mid\x_i\}}
{\Ep^\star\{\rho^{*2}(Y)\mid\x_i\}+\pi/(1-\pi)\Ep^\star\{\rho^*(Y)\mid\x_i\}}
\frac{r_iK_h(y-y_i)}{\sumj r_jK_h(y-y_j)}\n\\
&=&\int a^{*\star}(t)\rho^*(t)
\sumi\frac{\pyx^\star(t,\x_i)}{\Ep^\star\{\rho^{*2}(Y)\mid\x_i\}+\pi/(1-\pi)\Ep^\star\{\rho^*(Y)\mid\x_i\}}
\frac{r_iK_h(y-y_i)}{\sumj r_jK_h(y-y_j)}\dd t,
\ee
where $K_h(\cdot)\equiv K(\cdot/h)/h$, $K(\cdot)$ is a kernel
function and $h$ is a bandwidth, with conditions imposed later in our theoretical investigation.
(\ref{eq:a*hat}) is a Fredholm integral equation of the first type, which is ill-posed.
Numerical methods to provide stable and reliable solutions have been well studied in the literature \citep{hansen1992numerical}.
In our numerical implementations in Sections \ref{sec:sim} and \ref{sec:data}, we
use Landweber's iterative method \citep{landweber1951iteration} that is well-known to produce a
convergent solution.
We provide those technical details in Section~\ref{sec:fredholm} of
the Supplement.

We summarize the complete estimation procedure in Algorithm \ref{algdouble}.
\begin{algorithm}[H]
    \caption{Proposed Estimator $\wh\theta$: Doubly Flexible in $\rho^*(y)$ and $\pyx^\star(y,\x)$}\label{algdouble}
 	\begin{algorithmic}
        \STATE \textbf{Input}: data from population $\mp$: $(y_i, \x_i, r_i=1)$, $i=1,\ldots,n_1$, data from population $\mq$: $(\x_j, r_j=0)$, $j=n_1+1,\ldots,n$, and value $\pi=n_1/n$.
            \STATE \textbf{do}
            \STATE (a) adopt a working model for $\rho(y)$, denoted as $\rho^*(y)$;
            \STATE (b) adopt a working model for $\pyx(y,\x)$, denoted as $\pyx^\star(y,\x)$ or $\pyx^\star(y,\x,\wh\bzeta)$;
            \STATE (c) compute
    $ w_i=[\E_p^\star\{\rho^{*2}(Y)\mid\x_i\}+\pi/(1-\pi)\E_p^\star\{\rho^*(Y)\mid\x_i\}]^{-1}$ for $i=1,\dots,n$;
            \STATE (d) obtain $\wh a^{*\star}(\cdot)$ by solving the integral equation (\ref{eq:a*hat});
            \STATE (e) compute
    $\wh b^{*\star}(\x_i)=w_i \E_p^\star\{\wh a^{*\star}(Y)\rho^*(Y)\mid\x_i\}$ for $i=1,\dots,n$;
            \STATE (f) obtain $\wh\theta$ as
    \be\label{eq:thetawh}
    \wh\theta=
    \frac{1}{n}\sumi\left[\frac{r_i}{\pi}\rho^*(y_i)\{y_i-\wh b^{*\star}(\x_i)\}
    +\frac{1-r_i}{1-\pi}\wh b^{*\star}(\x_i)\right].
    \ee
 		\STATE \textbf{Output}: $\wh\theta$.
    \end{algorithmic}
\end{algorithm}

\begin{Rem}
In step (b) of Algorithm~\ref{algdouble}, one may adopt a completely
specified $\pyx^\star(y,\x)$ or a partially specified model
$\pyx^\star(y,\x,\bzeta)$ with an unknown parameter $\bzeta$.
If the latter case, a natural strategy is
to estimate $\bzeta$ first based on the observed samples from $\cal
P$ via, say MLE, to obtain $\wh\bzeta$, then use
$\pyx^\star(y,\x,\wh\bzeta)$ to replace the completely fixed
$\pyx^\star(y,\x)$.
In fact, we will show
that the action of estimating $\bzeta$ has no consequence in terms
of estimating $\theta$. This is an important discovery, because this
means one can always include a reasonably flexible model
$\pyx^\star(y,\x,\bzeta)$ so that it has a good chance of
approximating the true
$\pyx(y,\x)$. If $\pyx(y,\x)=\pyx^\star(y,\x,\bzeta_0)$ for certain
$\bzeta_0$, then even though the
additional parameter $\bzeta$ causes extra work, the reward is that
$\theta$ can be estimated as efficiently as if we knew $\pyx(y,\x)$ completely.
In all the subsequent steps, we replace
$\pyx^\star(y,\x)$ by $\pyx^\star(y,\x,\wh\bzeta)$ for its generality,
bearing in mind that $\pyx^\star(y,\x,\wh\bzeta)$
degenerates to $\pyx^\star(y,\x)$ when the parameter $\bzeta$ vanishes.
\end{Rem}
We now study the theoretical properties of $\wh\theta$ defined in (\ref{eq:thetawh}).
The main technical challenge is quantifying
the gap between the solutions for the integral equations \eqref{eq:a*star}
and \eqref{eq:a*hat}, encoded as $a^{*\star}(y)$ and $\wh a^{*\star}(y)$ respectively.
To facilitate the derivation, we define the linear operator
\bse
\calL^{*\star}(a)(y) &\equiv& \py(y)\E\left[\frac{\Ep^\star\{a(Y)\rho^*(Y)\mid\X\}}
{\Ep^\star\{\rho^{*2}(Y)\mid\X\}+\pi/(1-\pi)\Ep^\star\{\rho^*(Y)\mid\X\}}\mid y\right]
=\int a(t)u^{*\star}(t,y)\dd t, \mbox{ where}\\
u^{*\star}(t,y)&\equiv&\py(y)\int\frac{\rho^*(t)\pyx^\star(t,\x,\bzeta)}
{\Ep^\star\{\rho^{*2}(Y)\mid\x\}+\pi/(1-\pi)\Ep^\star\{\rho^*(Y)\mid\x\}}g(\x,y)\dd\x.
\ese
Apparently, $a^{*\star}(y)$ satisfies
\bse
\calL^{*\star}(a^{*\star})(y)=v(y), \mbox{ where } v(y)\equiv\py(y)y.
\ese
Similarly,
$\wh a^{*\star}(y)$ satisfies $\wh\calL^{*\star}(\wh a^{*\star})(y)=\wh v(y)$, where
\bse
\wh\calL^{*\star}(a)(y) &\equiv&
n_1^{-1}\sumi r_i K_h(y-y_i)\frac{\Ep^\star\{a(Y)\rho^*(Y)\mid\x_i,\wh\bzeta\}}
{\Ep^\star\{\rho^{*2}(Y)\mid\x_i,\wh\bzeta\}+\pi/(1-\pi)\Ep^\star\{\rho^*(Y)\mid\x_i,\wh\bzeta\}},\mbox{ and}\\
\wh v(y) &\equiv& n_1^{-1}\sumi v_{i,h}(y)  \equiv n_1^{-1}\sumi r_i K_h(y-y_i)y.
\ese
We first establish in Lemma~\ref{lem:l} that given regularity conditions
\ref{con:1complete}-\ref{con:4compact}, the linear operator
$\calL^{*\star}$, as well as its inverse, is well behaved.
\begin{enumerate}[label=(A\arabic*),ref=(A\arabic*),start=1]
    \item\label{con:1complete}
    The working model $\pyx^\star(y,\x)$ or
    $\pyx^\star(y,\x,\wh\bzeta)$ satisfies the completeness condition
    stated in Lemma~\ref{lem:iden}.
    \item\label{con:2rho}
    $\rho^*(y)>\delta$ for all $y$ on the support of $\py(y)$
    where $\delta$ is a positive constant, and
    $\rho^*(y)$ is twice differentiable and its derivative is bounded.
    \item\label{con:3boundedstar}
    The function $u^{*\star}(t,y)$ is bounded
    and has bounded derivatives with respect to $t$ and $y$ on its support.
    The function $a^{*\star}(y)$ in \eqref{eq:a*star} is bounded.
    \item\label{con:4compact}
    The support sets of $g(\x,y),\py(y),\rho^*(y)$ are compact.
\end{enumerate}
\begin{lemma}\label{lem:l}
Let $\|a\|_\infty\equiv\sup_y|a(y)|$.
Under Conditions~\ref{con:1complete}-\ref{con:4compact},
 the linear operator $\calL^{*\star}: L^\infty(R)\to L^\infty(R)$ is invertible.
In addition,
there exist positive finite constants $c_1, c_2$ such that
for all $a(y)\in L^\infty(R)$,
\begin{enumerate}[label=(\roman*)]
    \item $c_1\|a\|_\infty\leq\|\calL^{*\star}(a)\|_\infty\leq c_2\|a\|_\infty$,
    \item $\|{\cal L}^{*\star-1}(a)\|_\infty\leq c_1^{-1}\|a\|_\infty$.
\end{enumerate}
\end{lemma}
The proof of Lemma \ref{lem:l} is in Section~\ref{sec:ldoubleproof} of the Supplement.
To analyze the asymptotic normality of the estimator $\wh\theta$, we
add two more regularity conditions on the kernel function and the
bandwidth $h$.
\begin{enumerate}[label=(A\arabic*),ref=(A\arabic*),start=5]
    \item\label{con:5kernel}
    The kernel function $K(\cdot)\ge0$ is
    symmetric, bounded, and twice differentiable with bounded first derivative.
    It has support on $(-1,1)$ and satisfies $\int_{-1}^1K(t)\dd t=1$.
    \item\label{con:6bandwidth}
    The bandwidth $h$ satisfies $n_1(\log n_1)^{-4}h^2\to\infty$ and
    $n^2n_1^{-1}h^4\to 0$.
\end{enumerate}
Condition \ref{con:5kernel} is standard for kernel functions. Note
that we only need a one-dimensional kernel function $K(\cdot)$ in our
estimation procedure.
Condition \ref{con:6bandwidth} specifies the requirement of the
bandwidth $h$ associated with kernel function $K(\cdot)$. In general
we need both $h^{-1}=o\{n_1^{1/2}(\log n_1)^{-2}\}$ and
$h=o(n^{-1/2}n_1^{1/4})$. If $\pi$ is further assumed to be bounded
away from zero, the second requirement becomes $h=o(n_1^{-1/4})$ and
one can simply choose $h=n_1^{-1/3}$ to meet both requirements.
We are now ready to present the asymptotic normality of the estimator $\wh\theta$ below.
Its proof is contained in Section~\ref{sec:doubleproof} of the Supplement.
\begin{Th}\label{th:theta}
Assume $\wh\bzeta$ satisfies $\|\wh\bzeta-\bzeta\|_2=O_p(n_1^{-1/2})$
and $\Ep^\star\{\|\S_\bzeta^\star(Y,\x,\bzeta)\|_2\mid\x\}$ is bounded,
where $\S_\bzeta^\star(y,\x,\bzeta)\equiv\partial\log\pyx^\star(y,\x,\bzeta)/\partial\bzeta$.
For any choice of $\pyx^\star(y,\x,\bzeta)$ and $\rho^*(y)$,
under Conditions \ref{con:1complete}-\ref{con:6bandwidth},
\bse
\sqrt{n_1}(\wh\theta-\theta)\to N(0,\sigma_\theta^2)
\ese
in distribution as $n_1\to\infty$, where
\bse
\sigma_\theta^2
=\var\left(\sqrt{\pi}\phi\eff^{*\star}(\X,R,RY)
+\frac{R}{\sqrt{\pi}}
\left[\frac{\Ep^\star\{a^{*\star}(Y)\rho^*(Y)\mid\X\}}
{\Ep^\star\{\rho^{*2}(Y)\mid\X\}+\pi/(1-\pi)\Ep^\star\{\rho^*(Y)\mid\X\}}-Y\right]
\{\rho^*(Y)-\rho(Y)\}\right).
\ese
\end{Th}
In Theorem \ref{th:theta},  the only requirement on $\wh\bzeta$
is $\|\wh\bzeta-\bzeta\|_2=O_p(n_1^{-1/2})$.
Thus, the asymptotic variance of $\sqrt{n_1}(\wh\bzeta-\bzeta)$
does not affect the result in Theorem \ref{th:theta}
as long as $\wh\bzeta$ is $\sqrt{n_1}$-consistent for $\bzeta$.
This is easily achievable by constructing a standard MLE or moment
based  estimator for $\bzeta$
in the regression model of $Y$ given $\X$, $\pyx^\star(y,\x,\bzeta)$,
based on the $n_1$ observations from population $\cal P$.

In addition, it is clear that the estimator $\wh\theta$ is
$\sqrt{n_1}$-consistent, even if $n$ goes to infinity much faster than
$n_1$ does.
The intuition is that when we only have $n_1$ complete observations in
this problem, although a much larger $n_0$ can help us better
understand the label shift mechanism, it cannot improve the
convergence rate of $\wh\theta$.

Finally, Theorem~\ref{th:theta} also indicates that, when $\wh\bzeta$
is a $\sqrt{n_1}$-consistent estimator for $\bzeta$
such that $\pyx^\star(y,\x,\bzeta)=\pyx(y,\x)$,
and $\rho^*(y)$ is correctly specified as $\rho(y)$, the corresponding
estimator $\wh\theta$ achieves the semiparametric efficiency bound and
is the efficient estimator.
We state this result formally as Corollary \ref{th:thetaeff}.
Since it is a special case of Theorem~\ref{th:theta}, its proof is omitted.
\begin{Cor}\label{th:thetaeff}
Assume $\wh\bzeta$ satisfies $\|\wh\bzeta-\bzeta\|_2=O_p(n_1^{-1/2})$
and $\Ep^\star\{\|\S_\bzeta^\star(Y,\x,\bzeta)\|_2\mid\x\}$ is bounded.
If $\pyx^\star(y,\x,\bzeta)=\pyx(y,\x)$ and $\rho^*(y)=\rho(y)$,
under Conditions \ref{con:1complete}-\ref{con:6bandwidth},
\bse
\sqrt{n_1}(\wh\theta\eff-\theta)\to N[0,\var\{\sqrt{\pi}\phi\eff(\X,R,RY)\}]
\ese
in distribution as $n_1\to\infty$.
\end{Cor}

\section{Alternative Estimator $\wt\theta$: Singly Flexible in $\rho^*(y)$
  }\label{sec:singleflex}

Because the assessment of $\E_p(\cdot\mid\x)$ only relies on the
observed data, instead of adopting an arbitrary known model
$\E_p^\star(\cdot\mid\x)$ or parametric model $\E_p^\star(\cdot\mid\x,
\bzeta)$, one might be willing to estimate
$\E_p(\cdot\mid\x)$ in a model free fashion and replace $\E_p^\star(\cdot\mid\x)$ in the
estimation procedure presented in Section~\ref{sec:doubleflex} by a
well-behaved estimator $\wh \E_p(\cdot\mid\x)$. Here
we consider a general estimator $\wh \E_p(\cdot\mid\x)$ which has
convergence rate  faster than $n_1^{-1/4}$.
This rate is achievable for many nonparametric regression or machine
learning algorithms \citep{chernozhukov2018double}, see for example,
\cite{chen1999improved} for a class of neural network models,
\cite{wager2015adaptive} for a class of regression trees and random
forests, and \cite{bickel2009simultaneous, buhlmann2011statistics,
  belloni2011l1, belloni2013least} for a variety of sparse models.
Meanwhile, we still do not aim to estimate $\rho(y)$ since we
do not have the $Y$-data in population $\mq$.
We denote the corresponding estimator $\wt\theta$ and call it singly flexible because of its flexibility in using a working model $\rho^*(y)$.

The idea behind the estimator $\wt\theta$ is similar to $\wh\theta$,
therefore we only emphasize the difference from Section~\ref{sec:doubleflex}.
Similar to (\ref{eq:a*star}), we define $a^{*}(y)$ as the solution of
\be\label{eq:a*}
\E\left[\frac{\Ep\{a^{*}(Y)\rho^*(Y)\mid\X\}}
{\Ep\{\rho^{*2}(Y)\mid\X\}
+\pi/(1-\pi)\Ep\{\rho^*(Y)\mid\X\}}\mid y\right]
=y.
\ee
Equivalently, $a^{*}(y)$ satisfies $\calL^*(a^*)(y)=v(y)$, where
\bse
\calL^*(a)(y)&\equiv&\py(y)\E\left[\frac{\Ep\{a(Y)\rho^*(Y)\mid\X\}}
{\Ep\{\rho^{*2}(Y)\mid\X\}+\pi/(1-\pi)\Ep\{\rho^*(Y)\mid\X\}}\mid y\right]
=\int a(t)u^*(t,y)\dd t, \mbox{ and}\\
u^*(t,y)&\equiv&\py(y)\int\frac{\rho^*(t)\pyx(t,\x)}
{\Ep\{\rho^{*2}(Y)\mid\x\}+\pi/(1-\pi)\Ep\{\rho^*(Y)\mid\x\}}g(\x,y)\dd\x.
\ese
Using the estimator $\wh \E_p(\cdot\mid \x)$, we approximate the
integral equation (\ref{eq:a*}) as
\be
y&=&\wh\E\left[\frac{\wh\E_p\{a^{*}(Y)\rho^*(Y)\mid\X\}}
{\wh\E_p\{\rho^{*2}(Y)\mid\X\}+\pi/(1-\pi)\wh\E_p\{\rho^*(Y)\mid\X\}}\mid y\right]\nonumber\\
&=&\sumi\frac{\wh\E_p\{a^{*}(Y)\rho^*(Y)\mid\x_i\}}
{\wh\E_p\{\rho^{*2}(Y)\mid\x_i\}+\pi/(1-\pi)\wh\E_p\{\rho^*(Y)\mid\x_i\}}
\frac{r_iK_h(y-y_i)}{\sumj r_jK_h(y-y_j)},\label{eq:a*approx}
\ee
and we write $\wh a^*(y)$ as the solution to $\wh\calL^*(\wh a^*)(y)=\wh v(y)$, where
\bse
\wh\calL^*(a)(y)\equiv
n_1^{-1}\sumi r_i K_h(y-y_i)\frac{\wh\E_p\{a(Y)\rho^*(Y)\mid\x_i\}}
{\wh\E_p\{\rho^{*2}(Y)\mid\x_i\}+\pi/(1-\pi)\wh\E_p\{\rho^*(Y)\mid\x_i\}}.
\ese
We summarize the algorithm for computing the estimator $\wt\theta$ below.
\begin{algorithm}[H]
    \caption{Alternative Estimator $\wt\theta$: Single Flexible in $\rho^*(y)$}\label{algsingle}
 	\begin{algorithmic}
        \STATE \textbf{Input}: data from population $\mp$: $(y_i, \x_i, r_i=1)$, $i=1,\ldots,n_1$, data from population $\mq$: $(\x_j, r_j=0)$, $j=n_1+1,\ldots,n$, and value $\pi=n_1/n$.
            \STATE \textbf{do}
            \STATE (a) adopt a working model for $\rho(y)$, denoted as $\rho^*(y)$;
            \STATE (b) adopt a nonparametric or machine learning
            algorithm for estimating $\E_p(\cdot\mid\x)$, denoted as $\wh \E_p(\cdot\mid \x)$;
            \STATE (c) compute
    $\wh w_i=[\wh\E_p\{\rho^{*2}(Y)\mid\x_i\}+\pi/(1-\pi)\wh\E_p\{\rho^*(Y)\mid\x_i\}]^{-1}$ for $i=1,\dots,n$;
            \STATE (d) obtain $\wh a^{*}(\cdot)$ by solving the integral equation (\ref{eq:a*approx});
            \STATE (e) compute
    $\wh b^{*}(\x_i)=\wh w_i\wh\E_p\{\wh a^{*}(Y)\rho^*(Y)\mid\x_i\}$ for $i=1,\dots,n$;
            \STATE (f) obtain $\wt\theta$ as
    \be\label{eq:thetatilde}
    \wt\theta=
    \frac{1}{n}\sumi\left[\frac{r_i}{\pi}\rho^*(y_i)\{y_i-\wh b^{*}(\x_i)\}
    +\frac{1-r_i}{1-\pi}\wh b^{*}(\x_i)\right].
    \ee
 		\STATE \textbf{Output}: $\wt\theta$.
    \end{algorithmic}
\end{algorithm}
To develop the asymptotic normality of the estimator $\wt\theta$, instead of Condition~\ref{con:3boundedstar}, we need
\begin{enumerate}[label=(A\arabic*),ref=(A\arabic*),start=7]
    \item\label{con:3bounded}
    The function $u^{*}(t,y)$ is bounded
    and has bounded derivatives with respect to $t$ and $y$ on its support.
    The function $a^{*}(y)$ in \eqref{eq:a*} is bounded.
\end{enumerate}
We present Theorem \ref{th:theta2}, with its proof contained in
Section~\ref{sec:singleproof} of the Supplement.
\begin{Th}\label{th:theta2}
Assume $\wh\E_p$ satisfies $|\wh\E_p\{a(Y)\mid\x\}-\Ep\{a(Y)\mid\x\}|=o_p(n_1^{-1/4})$
for any bounded function $a(y)$.
For any choice of $\rho^*(y)$,
under Conditions~\ref{con:2rho}, \ref{con:4compact}-\ref{con:3bounded},
\bse
\sqrt{n_1}(\wt\theta-\theta)\to N(0,\sigma_\theta^2)
\ese
in distribution as $n_1\to\infty$, where
\bse
\sigma_\theta^2
=\var\left(\sqrt{\pi}\phi\eff^*(\X,R,RY)
+\frac{R}{\sqrt{\pi}}
\left[\frac{\Ep\{a^*(Y)\rho^*(Y)\mid\X\}}
{\Ep\{\rho^{*2}(Y)\mid\X\}+\pi/(1-\pi)\Ep\{\rho^*(Y)\mid\X\}}-Y\right]
\{\rho^*(Y)-\rho(Y)\}\right).
\ese
\end{Th}
It is direct from Theorem~\ref{th:theta2} that when the posited model $\rho^*(y)$ is correctly specified, the estimator $\wt\theta$ becomes the efficient estimator for $\theta$.
We point out this consequence as Corollary~\ref{th:thetaeff2} below.
\begin{Cor}\label{th:thetaeff2}
Assume $\wh\E_p$ satisfies $|\wh\E_p\{a(Y)\mid\x\}-\Ep\{a(Y)\mid\x\}|=o_p(n_1^{-1/4})$
for any bounded function $a(y)$.
If $\rho^*(y)=\rho(y)$,
under Conditions~\ref{con:2rho}, \ref{con:4compact}-\ref{con:3bounded},
\bse
\sqrt{n_1}(\wt\theta\eff-\theta)\to N[0,\var\{\sqrt{\pi}\phi\eff(\X,R,RY)\}]
\ese
in distribution as $n_1\to\infty$.
\end{Cor}
Last but not least, Sections~\ref{sec:theta} and \ref{sec:singleflex} here only present the results for estimating $\theta=\Eq(Y)$.
The whole story can be extended to a general parameter $\bt$ such that $\Eq\{\U(\X,Y,\bt)\}=\0$, and the results are stated in
Sections~\ref{sec:generalU} and \ref{sec:proofgeneralU} of the Supplement.
In our numerical studies in Sections~\ref{sec:sim} and \ref{sec:data} below, we analyze both $\Eq(Y)$ and the $t$-th quantile of population $\mq$, defined as
\bse
\tau_{q,t}(Y) = \inf\left[ y: \E_q \{I(Y\leq y)\}\geq t\right],
\ese
where $0<t<1$.
This corresponds to $\E_q\left[\eta_t \left\{Y-\tau_{q,t}(Y)\right\}\right]=0$ where $\eta_t(r) = t-I(r<0)$.

\section{Simulation Studies}\label{sec:sim}

We conduct simulation studies to assess the finite sample performance
of our proposed methods.
We report the results for the mean $\E_q(Y)$ and the median $\tau_{q,0.5}(Y)$ of the outcome $Y$ in population $\mq$.

We first generate a binary indicator $R_i,i=1,\dots,n$
from the Bernoulli distribution with probability $0.5$, and record
$n_1=\sumi r_i$, $\pi=n_1/n$.
Then we generate $Y_i$ from $N(0,1)$ if $R_i=1$ and from $N(1,1)$ if $R_i=0$.
The $\X\mid Y$ distribution is generated from a 3-dimensional normal with mean
$(-0.5,0.5,1)\trans Y_i$ and covariance $\I$, the identity matrix.
This implies, $\E_q(Y)=1$, $\tau_{q,0.5}=1$ and the true density ratio model
\bse
\rho(y)=\exp(-0.5+y).
\ese
One can derive that $\pyx(y,\x,\bzeta)$ follows normal with mean
$(1,\x\trans)\bb$ where $\bb=(0,-0.2,0.2,0.4)\trans$ and variance $\sigma^2=0.4$.
Here we denote $\bzeta\equiv(\bb\trans,\sigma^2)\trans$.

We use the following misspecified working models. We define
\bse
\rho^*(y)\equiv c^*\exp(-0.7+1.2y),
\ese
where $c^*\equiv\pi/\{n^{-1}\sumi r_i\exp(-0.7+1.2y_i)\}$ in order to satisify $\E\{R\rho^*(Y)\}=\pi$.
For the working model $\pyx^\star$, we define $\x^{\star}\equiv[x_1,\exp(x_2/2),x_3/\{1+\exp(x_2)\}+10]\trans$ and define $\pyx^\star(y,\x^\star,\bzeta^\star)$ as the normal distribution with mean
$(1,\x^{\star\rm T})\bb^\star$ where $\bb^{\star}=(-7.000,-0.223,0.363,0.664)\trans$ and variance $\sigma^{\star2}=0.449$.
The parameter $\bzeta^\star\equiv(\bb^{\star\rm T},\sigma^{\star2})\trans$ is obtained by minimizing the Kullback-Leibler distance $D_{kl}(\pyx\| \pyx^\star)$.

We implement the following seven estimators:
\begin{itemize}
  \item[(1)] \verb"shift-dependent"$^*$: $\widecheck\theta$ in (\ref{eq:shift-dependent}) with $\rho^*(y)$;
  \item[(2)] \verb"doubly-flexible"$^{*\star}$: $\wh\theta$ in (\ref{eq:thetawh}) with $\rho^*(y)$ and $\pyx^\star(y,\x,\bzeta^\star)$, theoretically analyzed in Theorem~\ref{th:theta};
  \item[(3)] \verb"singly-flexible"$^*$: $\wt\theta$ in (\ref{eq:thetatilde}) with $\rho^*(y)$, theoretically analyzed in Theorem~\ref{th:theta2};
  \item[(4)] \verb"shift-dependent"$^0$: $\widecheck\theta$ in (\ref{eq:shift-dependent}) with correct $\rho(y)$;
  \item[(5)] \verb"doubly-flexible"$^0$: $\wh\theta\eff$ with correct $\rho(y)$ and $\pyx(y,\x,\bzeta)$, theoretically analyzed in Corollary~\ref{th:thetaeff};
  \item[(6)] \verb"singly-flexible"$^0$: $\wt\theta\eff$ with correct $\rho(y)$, theoretically analyzed in Corollary~\ref{th:thetaeff2};
  \item[(7)] \verb"oracle": the $\sqrt{n_0}$-consistent estimator $\frac1n\sumi \frac{1-r_i}{1-\pi}y_i$.
\end{itemize}
Note that the last four estimators (shown as ``gray'' in
Figures~\ref{fig:boxplotmean} and \ref{fig:boxplotmedian}) are
unrealistic since they either use the unknown models $\rho(y)$ and
$\pyx(y,\x,\bzeta)$ or the $Y$-data in population $\mq$.

In implementing estimators \verb"doubly-flexible"$^{*\star}$,
\verb"singly-flexible"$^*$, \verb"doubly-flexible"$^0$ and
\verb"singly-flexible"$^0$,
we solve the integral equations \eqref{eq:a*hat} and
\eqref{eq:a*approx} using the Nadaraya-Watson estimator for
$\E(\cdot\mid y)$ with Gaussian kernel and bandwidth $h=n_1^{-1/3}$
that is discussed in Condition~\ref{con:6bandwidth}.
Numerically, the integrations are approximated by the Gauss-Legendre
quadrature with 50 points on the interval $[-5,5]$ and the integral
equations are evaluated at $y_i,i=1,\dots,n_1$.
In addition, for estimators \verb"singly-flexible"$^*$ and
\verb"singly-flexible"$^0$, we estimate $\Ep(\cdot\mid \x)$ using the
Nadaraya-Watson estimator based on the product Gaussian kernel with
bandwidth $2.5n_1^{-1/7}$, where the order comes from the optimal bandwidth
$n_1^{-1/(4+d)}$ with $d$ the dimensionality of covariate $\X$. See
Section \ref{sec:fredholm} of the Supplement for technical details on
the numerical implementation.

Based on 1000 simulation replicates, Figures~\ref{fig:boxplotmean} and \ref{fig:boxplotmedian} illustrate the boxplots of the estimates for the mean and the median, respectively.
With the misspecified working model $\rho^*(y)$, the estimator \verb"shift-dependent"$^*$ is biased;
in contrast, the proposed estimators \verb"doubly-flexible"$^{*\star}$ and \verb"singly-flexible"$^*$ are both unbiased.
When the correct model $\rho(y)$ is used, not surprisingly, all of the estimators \verb"shift-dependent"$^0$, \verb"doubly-flexible"$^{0}$ and \verb"singly-flexible"$^0$ are unbiased.
It is also clear that the two proposed flexible estimators are always more efficient than the shift-dependent estimator no matter the correct model $\rho(y)$ is used or not.

To further demonstrate the efficiency comparison and the inference results, in Tables~\ref{tab:mean1} and \ref{tab:median1}, we report the mean squared error (MSE), the empirical bias (Bias), the empirical standard error (SE), the average of estimated standard error ($\wh{\rm SE}$), and the empirical coverage at $95\%$ confidence level (CI), for each of the estimators.
The estimator \verb"shift-dependent"$^*$ has an incorrect coverage
(over-coverage in Table~\ref{tab:mean1} and under-coverage in
Table~\ref{tab:median1}) because of its severe bias.
This issue is not mitigated at all in Table~\ref{tab:mean1} or becomes
even worse in Table~\ref{tab:median1} when we increase the size of the
stacked random sample from 500 to 1000.
On the contrary, the estimators \verb"doubly-flexible"$^{*\star}$ and \verb"singly-flexible"$^*$ are correctly covered.
Though there is no theoretical justification, \verb"doubly-flexible"$^{*\star}$ is slightly less efficient than \verb"singly-flexible"$^*$ in this setting.
This indicates that the effort of correctly estimating $\Ep(\cdot\mid \x)$ pays off in the sense of improving the estimation efficiency.
With the correct $\rho(y)$ model used, each of the estimators \verb"doubly-flexible"$^0$ and \verb"singly-flexible"$^0$ is more efficient than its counterpart \verb"doubly-flexible"$^{*\star}$ and \verb"singly-flexible"$^*$.
The $\sqrt{n_0}$-consistent estimator \verb"oracle" is the most efficient one in this simulation setting.

\begin{figure}[hp]
    \centering
    \caption{Boxplots of estimates for mean in the simulation study.
    Dashed line: the true estimand.}
    \label{fig:boxplotmean}
    \begin{subfigure}{0.45\textwidth}
        \centering
        \caption{mean, $n=500$}
        \includegraphics[width=\textwidth]{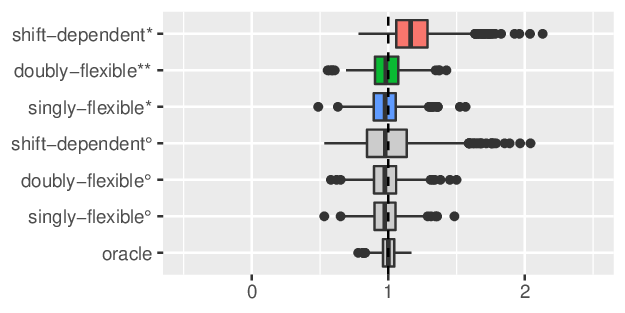}
        \label{fig:mean500}
    \end{subfigure}
    \begin{subfigure}{0.45\textwidth}
        \centering
        \caption{mean, $n=1000$}
        \includegraphics[width=\textwidth]{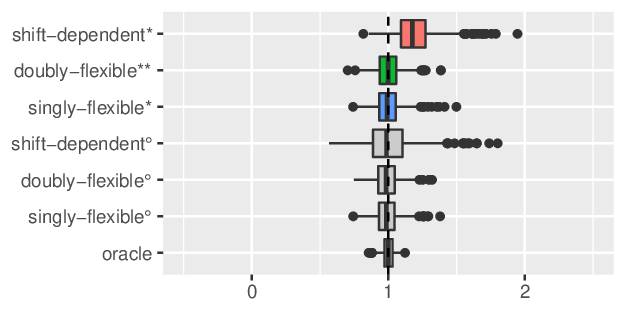}
        \label{fig:mean1000}
    \end{subfigure}
\end{figure}

\begin{table}[hp]
    \centering
    \caption{Summary of mean estimation results in the simulation study.}
    \label{tab:mean1}
    \begin{tabular}{c|lcc|rrrrr}
    $n$ & Estimator & $\rho(y)$ & $\pyx(y,\x)$ & MSE & Bias & SE & $\wh{\rm SE}$ & CI \\
    \hline
    \multirow{7}{*}{500}
    & shift-dependent$^*$        &
    $\rho^*(y)$ & -                           & .0699 & .1840  & .1899 & .2791 & 1.000 \\
    & doubly-flexible$^{*\star}$ &
    $\rho^*(y)$ & $\pyx^\star(y,\x,\wh\bzeta)$ & .0173 & -.0120 & .1311 & .1287 & .943 \\
    & singly-flexible$^*$        &
    $\rho^*(y)$ & $\wh\E_p(\cdot\mid \x)$      & .0162 & -.0201 & .1258 & .1230 & .941 \\
    \cline{2-9}
    & shift-dependent$^0$        &
    $\rho(y)$ & -                     & .0533 & .0049  & .2309 & .2119 & .899 \\
    & doubly-flexible$^0$        &
    $\rho(y)$ & $\pyx(y,\x,\wh\bzeta)$ & .0153 & -.0231 & .1214 & .1212 & .941 \\
    & singly-flexible$^0$        &
    $\rho(y)$ & $\wh\E_p(\cdot\mid\x)$ & .0138 & -.0221 & .1155 & .1169 & .939 \\
    & oracle                     &
    -        & -                     & .0040 & .0006  & .0636 & .0633 & .943 \\
    \hline
    \multirow{7}{*}{1000}
    & shift-dependent$^*$        &
    $\rho^*(y)$ & -                           & .0561 & .1906  & .1406 & .2077 & .999 \\
    & doubly-flexible$^{*\star}$ &
    $\rho^*(y)$ & $\pyx^\star(y,\x,\wh\bzeta)$ & .0085 & .0013  & .0922 & .0912 & .955 \\
    & singly-flexible$^*$        &
    $\rho^*(y)$ & $\wh\E_p(\cdot\mid \x)$      & .0081 & -.0031 & .0899 & .0850 & .952 \\
    \cline{2-9}
    & shift-dependent$^0$        &
    $\rho(y)$ & -                     & .0275 & .0024  & .1660 & .1533 & .927 \\
    & doubly-flexible$^0$        &
    $\rho(y)$ & $\pyx(y,\x,\wh\bzeta)$ & .0075 & -.0125 & .0856 & .0861 & .958 \\
    & singly-flexible$^0$        &
    $\rho(y)$ & $\wh\E_p(\cdot\mid\x)$ & .0069 & -.0094 & .0824 & .0827 & .955 \\
    & oracle                     &
    -        & -                     & .0020 & .0008  & .0451 & .0447 & .945
    \end{tabular}
\end{table}

\begin{figure}[hp]
    \centering
    \caption{Boxplots of estimates for median in the simulation study.
    Dashed line: the true estimand.}
    \label{fig:boxplotmedian}
    \begin{subfigure}{0.45\textwidth}
        \centering
        \caption{median, $n=500$}
        \includegraphics[width=\textwidth]{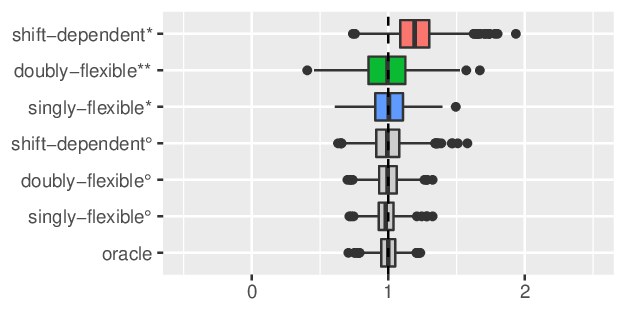}
        \label{fig:median500}
    \end{subfigure}
    \begin{subfigure}{0.45\textwidth}
        \centering
        \caption{median, $n=1000$}
        \includegraphics[width=\textwidth]{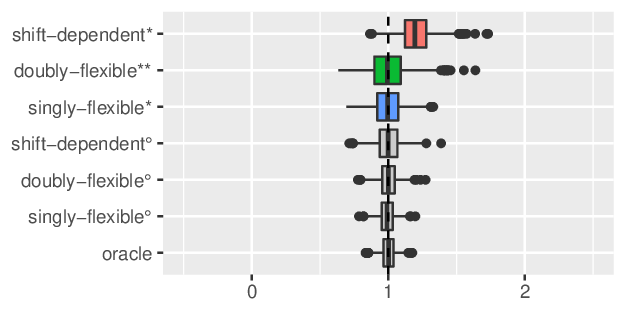}
        \label{fig:median1000}
    \end{subfigure}
\end{figure}

\begin{table}[hp]
    \centering
    \caption{Summary of median estimation results in the simulation study.}
    \label{tab:median1}
    \begin{tabular}{c|lcc|rrrrr}
    $n$ & Estimator &  $\rho(y)$ & $\pyx(y,\x)$ & MSE & Bias & SE & $\wh{\rm SE}$ & CI \\
    \hline
    \multirow{7}{*}{500}
    & shift-dependent$^*$        &
    $\rho^*(y)$ & -                           & .0695 &  .2025 & .1687 & .1547 & .802 \\
    & doubly-flexible$^{*\star}$ &
    $\rho^*(y)$ & $\pyx^\star(y,\x,\wh\bzeta)$ & .0368 & -.0072 & .1918 & .1764 & .940 \\
    & singly-flexible$^*$        &
    $\rho^*(y)$ & $\wh\E_p(\cdot\mid \x)$      & .0211 &  .0024 & .1453 & .1390 & .941 \\
    \cline{2-9}
    & shift-dependent$^0$        &
    $\rho(y)$ & -                     & .0178 &  .0011 & .1336 & .1286 & .947 \\
    & doubly-flexible$^0$        &
    $\rho(y)$ & $\pyx(y,\x,\wh\bzeta)$ & .0093 & -.0033 & .0964 & .0950 & .951 \\
    & singly-flexible$^0$        &
    $\rho(y)$ & $\wh\E_p(\cdot\mid\x)$ & .0071 & -.0162 & .0827 & .0881 & .956 \\
    & oracle                     &
    -        & -                     & .0064 & -.0020 & .0799 & .0821 & .946 \\
    \hline

    \multirow{7}{*}{1000}
    & shift-dependent$^*$        &
    $\rho^*(y)$ & -                           & .0554 &  .2018 & .1210 & .1104 & .560 \\
    & doubly-flexible$^{*\star}$ &
    $\rho^*(y)$ & $\pyx^\star(y,\x,\wh\bzeta)$ & .0229 & -.0013 & .1515 & .1444 & .943 \\
    & singly-flexible$^*$        &
    $\rho^*(y)$ & $\wh\E_p(\cdot\mid \x)$      & .0128 & -.0023 & .1130 & .1124 & .954 \\
    \cline{2-9}
    & shift-dependent$^0$        &
    $\rho(y)$ & -                     & .0091 &  .0005 & .0955 & .0915 & .934 \\
    & doubly-flexible$^0$        &
    $\rho(y)$ & $\pyx(y,\x,\wh\bzeta)$ & .0047 &  .0000 & .0688 & .0669 & .954 \\
    & singly-flexible$^0$        &
    $\rho(y)$ & $\wh\E_p(\cdot\mid\x)$ & .0036 & -.0098 & .0594 & .0625 & .948 \\
    & oracle                     &
    -        & -                     & .0031 &  .0004 & .0558 & .0578 & .959
    \end{tabular}
\end{table}

\section{Data Application}\label{sec:data}

We now illustrate the numerical performance of our proposed method through analyzing the Medical Information Mart for Intensive Care III (MIMIC-III), an openly available electronic health records database, developed by the MIT Lab for Computational Physiology \citep{johnson2016mimic}.
It comprises deidentified health-related records including demographics, vital signs, laboratory test and medications, for 46,520 patients who admitted to the intensive care unit
of the Beth Israel Deaconess Medical Center between 2001 and 2012.

The outcome of interest $Y$ in our analysis is the sequential organ failure assessment (SOFA) score \citep{singer2016third}, used to track a patient's status during the stay in an intensive care unit to determine the extent of a patient's organ function or rate of failure.
The score is based on six different sub-scores, with one of each for the respiratory, cardiovascular, hepatic, coagulation, renal and neurological systems.
The SOFA score ranges from 0 (best) to 24 (worst).
We include 16 covariates from either chart events (6 variables, diastolic blood pressure, systolic blood pressure, blood glucose, respiratory rate per minute, and two measures from body temperature) or laboratory tests (10 variables, peripheral caillary oxygen saturation, two measures from each of hematocrit level, platelets count and red blood cell count, and three measures from blood urea nitrogen).
We choose these covariates through assessing whether the absolute correlation with the outcome $Y$ is greater than 0.2 and whether the missing rate is less than 1\%.
In our analysis, we only include the first admission if the patient was admitted to the intensive care unit more than once.
We also exclude patients whose outcome $Y$ is greater than or equal to 20, whose age is greater than or equal to 65, and who has missing values in any of the covariates.
This results in a total of $n=16,691$ records.

In our analysis, we define the population $\mp$ as patients with private, goverment, and self-pay insurances ($R=1, n_1=11,695$), and population $\mq$ as patients whose insurance type is either Medicaid or Medicare ($R=0, n_0=4,996$).
The label shift assumption that the conditional distribution of $\X$ given $Y$ remains the same can be tested via the conditional independence of $\X$ and $R$ given $Y$.
In our analysis, we test the conditional independence between $R$ and each of covariates by the
invariant environment prediction test \citep{heinze2018invariant} in R package \verb"CondIndTests", and the p-values range from 0.460 to 0.628.
This indicates the label shift assumption is indeed sensible in our analysis.
We first compute the sample mean (3.7409) and sample $t$-th quantiles (1,1,3,5,8 for $t=(10, 25, 50, 75, 90)\%$) of SOFA scores among patients whose insurance type is either Medicaid or Medicare. We regard these estimates as \verb"oracle" in order to compare with our proposed methods.

To identify a reasonable working model $\rho^*(y)$, we model the data $Y+0.001$ from population $\mp$ as a parametric gamma distribution $f(y,\alpha,\beta)=\Gamma(\alpha)^{-1}\beta^{-\alpha}y^{\alpha-1}\exp(-y/\beta)$ with $\Gamma(\cdot)$ the $\Gamma$-function, the shape parameter $\alpha>0$ and the scale parameter $\beta>0$.
We estimate the unknown parameters $\alpha$ and $\beta$ as $\wh\alpha$ and $\wh\beta$ using the MLE.
For the $Y$-data in population $\mq$, we assume $Y+0.001$ follows a similar gamma distribution with shape parameter $\wh\alpha+1$ and scale parameter $\wh\beta$.
Hence, we use the working model
\bse
\rho^*(y) = \frac{f(y+0.001,\wh\alpha+1,\wh\beta)}{f(y+0.001,\wh\alpha,\wh\beta)}.
\ese
To implement the estimator \verb"doubly-flexible"$^{*\star}$, we
impose a parametric model $\pyx^\star(y,\x,\bzeta)$ by regressing
$Y+0.001$ on $\X$ as a generalized linear model with gamma
distribution, and estimate $\bzeta$ using the MLE.
To implement the estimator \verb"singly-flexible"$^*$, we identify the
first three principal components from the 16 covariates, and then
estimate $\E_p(\cdot\mid\x)$ as a function of those three principal
components using the Nadaraya-Watson estimator based on the product
Gaussian kernel with bandwidth $0.5n_1^{-1/7}$.
To solve the corresponding integral equations, similar to
Section~\ref{sec:sim}, we approximate $\E(\cdot\mid y)$ by its
Nadaraya-Watson estimator with the Gaussian kernel and bandwidth
$h=n_1^{-1/3}$.
In addition, for numerical implementation, the integration is
approximated at 50 equally-spaced points on the interval $[0, 19]$,
and the integral equations are evaluated at each supporting point of
$\{y_i:i=1,\dots,n_1\}$. See Section \ref{sec:fredholm} of the
Supplement for technical details on the implementation.

The results are summarized in Table~\ref{tab:realmean1} for the mean
and in Table~\ref{tab:realquantile1} for the quantiles.
The estimator \verb"shift-dependent"$^*$ that relies on a misspecified
model $\rho^*(y)$ severely over-estimates the quantities compared to the \verb"oracle" estimate, in most of the scenarios including estimating the mean, 25\%, 75\% and 90\% quantiles.
As a consequence, the \verb"oracle" estimate cannot be covered by the confidence intervals.
In contrast, the proposed estimators \verb"doubly-flexible"$^{*\star}$
and \verb"singly-flexible"$^*$, although also rely on $\rho^*(y)$,
provide almost identical estimates as \verb"oracle".
Accordingly, the confidence intervals from the proposed methods all
cover the \verb"oracle" estimate.
In these scenarios, \verb"singly-flexible"$^*$ is more efficient than
\verb"doubly-flexible"$^{*\star}$, which echoes our findings in
Section~\ref{sec:sim}.

When estimating 10\% and 50\% quantiles, we find that \verb"shift-dependent"$^*$ gives almost the same estimate as \verb"oracle".
It might be plausible that the difference between $\rho^*(y)$ and the true $\rho(y)$ is minor for estimating these two quantities.
Nevertheless, the proposed estimators \verb"doubly-flexible"$^{*\star}$ and \verb"singly-flexible"$^*$ are still more efficient than the estimator \verb"shift-dependent"$^*$.

Finally it is interesting to observe that, the estimator \verb"oracle" is even less efficient than the proposed estimators \verb"doubly-flexible"$^{*\star}$ and \verb"singly-flexible"$^*$, in estimating all of the quantiles.
This is because \verb"oracle" is $\sqrt{n_0}$-consistent whereas the two proposed estimators are both $\sqrt{n_1}$-consistent.
In this application, $n_1$ is 2.34 times greater than $n_0$, which might result in the situation that the \verb"oracle" estimate being less efficient.
In Section~\ref{sec:sim}, we also considered situations that $n_1$ is much larger than $n_0$ and similar phenomenon was observed as well, with detailed results omitted.

\section{Conclusion}\label{sec:disc}

In this paper, we estimate a characteristic of a target population
$\mq$, via exploiting the data and information from a different but
relevant population $\mp$, under the label shift assumption.
Different from most existing literatures, our proposal is devised to
accommodate  both classification and regression problems.
We primarily propose the doubly flexible estimate, whose unique
feature is to simultaneously allow both models to be misspecified thus
is flexible: the density ratio model $\rho(y)$ that governs the label
shift mechanism, and the conditional distribution model $\pyx(y,\x)$
of population $\mp$.
While the estimation of the latter can be done via off-the-shelf
procedures sometimes, it often faces curse of dimensionality or
computational challenges. Further,
estimating $\rho(y)$ is even more difficult
because the $Y$-data in population $\mq$ is not accessible in our
procedure.

\begin{table}[hp]
    \centering
    \caption{Mean estimation results in the data application.}
    \label{tab:realmean1}
    \begin{tabular}{lcc|rrrr}
    Estimator & $\rho(y)$ & $\pyx(y,\x)$ & Estimate & diff. with oracle & $\wh{\rm SE}$ & CI \\
    \hline
    shift-dependent$^*$ & $\rho^*(y)$ & - &
    4.0529 & .3120 & .0593 & [3.9367, 4.1691] \\
    doubly-flexible$^{*\star}$ & $\rho^*(y)$ & $\pyx^\star(y,\x,\wh\bzeta)$ &
    3.7579 & .0170 & .0803 & [3.6005, 3.9153] \\
    singly-flexible$^*$ &  $\rho^*(y)$ & $\wh\E_p(\cdot\mid \x)$ &
    3.7542 & .0133 & .0496 & [3.6570, 3.8514] \\
    oracle & - & - &
    3.7409 & -    & .0405 & [3.6616, 3.8202]
    \end{tabular}
\end{table}

\begin{table}[hp]
    \centering
    \small
    \caption{Quantile estimation results in the data application.}
    \label{tab:realquantile1}
    \begin{tabular}{c|lcc|rrrr}
    $\tau$ & Estimator & $\rho(y)$ & $\pyx(y,\x)$ & Estimate & diff. with oracle & $\wh{\rm SE}$ & CI \\
    \hline
    \multirow{4}{*}{10\%}
    & shift-dependent$^*$ & $\rho^*(y)$ & - &
    1.0000 &  .0000 & .0102 & [0.9801, 1.0199] \\
    & doubly-flexible$^{*\star}$ &  $\rho^*(y)$ & $\pyx^\star(y,\x,\wh\bzeta)$ &
    0.9998 & -.0002 & .0095 & [0.9812, 1.0185] \\
    & singly-flexible$^*$ &  $\rho^*(y)$ & $\wh\E_p(\cdot\mid \x)$ &
    0.9998 & -.0002 & .0099 & [0.9803, 1.0192] \\
    & oracle & - & - &
    1.0000 &     - & .0218 & [0.9572, 1.0428] \\
    \hline
    \multirow{4}{*}{25\%}
    & shift-dependent$^*$ & $\rho^*(y)$ & - &
    1.9995 &  .9995 & .0249 & [1.9508, 2.0483] \\
    & doubly-flexible$^{*\star}$ &  $\rho^*(y)$ & $\pyx^\star(y,\x,\wh\bzeta)$ &
    1.0002 &  .0002 & .0276 & [0.9460, 1.0543] \\
    & singly-flexible$^*$ &  $\rho^*(y)$ & $\wh\E_p(\cdot\mid \x)$ &
    1.0004 &  .0004 & .0196 & [0.9620, 1.0389] \\
    & oracle & - & - &
    1.0000 &     - & .0315 & [0.9383, 1.0617] \\
    \hline
    \multirow{4}{*}{50\%}
    & shift-dependent$^*$ & $\rho^*(y)$ & - &
    3.0002 &  .0002 & .0408 & [2.9203, 3.0801] \\
    & doubly-flexible$^{*\star}$ &  $\rho^*(y)$ & $\pyx^\star(y,\x,\wh\bzeta)$ &
    3.0001 &  .0001 & .0333 & [2.9348, 3.0654] \\
    & singly-flexible$^*$  & $\rho^*(y)$ & $\wh\E_p(\cdot\mid \x)$ &
    3.0005 &  .0005 & .0334 & [2.9350, 3.0659] \\
    & oracle & - & - &
    3.0000 &     - & .0550 & [2.8923, 3.1077] \\
    \hline
    \multirow{4}{*}{75\%}
    & shift-dependent$^*$ & $\rho^*(y)$ & - &
    5.9999 &  .9999 & .0742 & [5.8544, 6.1454] \\
    & doubly-flexible$^{*\star}$  & $\rho^*(y)$ & $\pyx^\star(y,\x,\wh\bzeta)$ &
    5.0001 &  .0001 & .0759 & [4.8512, 5.1489] \\
    & singly-flexible$^*$ &  $\rho^*(y)$ & $\wh\E_p(\cdot\mid \x)$ &
    5.0001 &  .0001 & .0381 & [4.9254, 5.0748] \\
    & oracle & - & - &
    5.0000 &     - & .0591 & [4.8842, 5.1158] \\
    \hline
    \multirow{4}{*}{90\%}
    & shift-dependent$^*$ & $\rho^*(y)$ & - &
    8.9999 &  .9999 & .1380 & [8.7295, 9.2703] \\
    & doubly-flexible$^{*\star}$ &  $\rho^*(y)$ & $\pyx^\star(y,\x,\wh\bzeta)$ &
    8.0004 &  .0004 & .1003 & [7.8039, 8.1969] \\
    & singly-flexible$^*$ &  $\rho^*(y)$ & $\wh\E_p(\cdot\mid \x)$ &
    8.0004 &  .0004 & .0638 & [7.8754, 8.1253] \\
    & oracle & - & - &
    8.0000 &     - & .1093 & [7.7858, 8.2142]
    \end{tabular}
\end{table}

\newpage
\bibliographystyle{agsm}
\bibliography{refLabelShift}

\clearpage
\pagenumbering{arabic}
{\centering
\section*{SUPPLEMENTARY MATERIAL}}
\setcounter{equation}{0}\renewcommand{\theequation}{S.\arabic{equation}}
\setcounter{subsection}{0}\renewcommand{\thesubsection}{S.\arabic{subsection}}
\setcounter{pro}{0}\renewcommand{\thepro}{S.\arabic{pro}}
\setcounter{Th}{0}\renewcommand{\theTh}{S.\arabic{Th}}
\setcounter{lemma}{0}\renewcommand{\thelemma}{S.\arabic{lemma}}
\setcounter{Rem}{0}\renewcommand{\theRem}{S.\arabic{Rem}}
\setcounter{Cor}{0}\renewcommand{\theCor}{S.\arabic{Cor}}
\setcounter{algorithm}{0}\renewcommand{\thealgorithm}{S.\arabic{algorithm}}

\subsection{Proof of Lemma~\ref{lem:iden}}\label{sec:slemmaproof}
Suppose we have two different sets of models: $g(\x, y)$, $\py(y)$,
$\rho(y)$ and $\wt g(\x, y)$, $\wt p_Y(y)$, $\wt \rho(y)$.
From the
likelihood (\ref{eq:likelihood}), it is obvious that
\bse
g(\x, y)\py(y)&=&\wt g(\x, y) \wt p_Y(y),\\
\int g(\x, y)\py(y)\rho(y)\dd y &=& \int \wt g(\x, y)\wt p_Y(y)\wt \rho(y)\dd y.
\ese
Integrate with respect to $\x$ on both sides of the first equation
above, it is apparent that $\py(y)=\wt p_Y(y)$ and $g(\x, y)=\wt
g(\x, y)$. Therefore, the second equation above becomes
\bse
&&\int g(\x, y)\py(y)\rho(y)\dd y = \int g(\x, y)\py(y)\wt \rho(y)\dd y\\
&=& \int \pyx(y\mid \x)\px(\x)\rho(y)\dd y = \int \pyx(y\mid
\x)\px(\x)\wt \rho(y)\dd y.
\ese
Hence, $\int \pyx(y\mid \x)\rho(y)\dd y = \int \pyx(y\mid \x)\wt
\rho(y)\dd y$. Using the completeness condition, clearly $\rho(y)=\wt
\rho(y)$.
\qed

\subsection{Derivation of $\calF$}\label{sec:influence}

We first state and prove a result regarding the tangant space.
\begin{pro}\label{pro:tangentspace}
The tangent space of (\ref{eq:likelihood}) is
$\calT\equiv\calT_\ba\oplus(\calT_\bb+\calT_\bg)$, where
\bse
\calT_\ba&=&\left[r\a_1(y):\Ep\{\a_1(Y)\}=\0\right],\\
\calT_\bb&=&\left[r\a_2(\x,y)+(1-r)\Eq\{\a_2(\x,Y)\mid\x\}:\E\{\a_2(\X,y)\mid y\}=\0\right],\\
\calT_\bg&=&\left[(1-r)\Eq\{\a_3(Y)\mid\x\}:\Eq\{\a_3(Y)\}=\0\right].
\ese
\end{pro}
\noindent Proof.
Consider a parametric submodel of \eqref{eq:likelihood},
\be\label{eq:submodel}
f_{\X,R,RY}(\x,r,ry,\bd)
&=&\pi^r(1-\pi)^{1-r}\{g(\x,y,\bb)\py(y,\ba)\}^{r}
\left\{\int g(\x,y,\bb)\qy(y,\bg)dy\right\}^{1-r},
\ee
where $\bd=(\ba\trans,\bb\trans,\bg\trans)\trans$.
We can derive that the score function associated with an arbitrary $\bd$ is
$\S_\bd\equiv(\S_\ba\trans,\S_\bb\trans,\S_\bg\trans)\trans$,
where
\bse
\S_\ba(\x,r,ry)
&\equiv&r\a_\ba(y),\\
\S_\bb(\x,r,ry)
&\equiv&r\a_\bb(\x,y)+(1-r)\Eq\{\a_\bb(\x,Y)\mid\x\},\\
\S_\bg(\x,r,ry)
&\equiv&(1-r)\Eq\{\a_\bg(Y)\mid\x\},
\ese
$\Ep\{\a_\ba(Y)\}=\0$,
$\Ep\{\a_\bb(\X,y)|y\}=\Eq\{\a_\bb(\X,y)|y\}=\E\{\a_\bb(\X,y)|y\}=\0$,
and $\Eq\{\a_\bg(Y)\}=\0$.
The above derivation directly leads to Proposition
\ref{pro:tangentspace}.\qed

We are now ready to establish the result regarding $\calF$, which we
explicitly write out as Proposition \ref{pro:influ}.
\begin{pro}\label{pro:influ}
The set of the influence functions for $\theta$
is
\be\label{eq:influence}
\calF
\equiv\left[\frac{r}{\pi}\{y\rho(y)-b(\x)\rho(y)+c\}+\frac{1-r}{1-\pi}\{b(\x)-\theta-c\}:
\E\{b(\X)\mid y\}=y,\forall c\right].
\ee
\end{pro}
\noindent Proof.
Note that $\theta=\Eq(Y)$, hence
\bse
\frac{\partial\theta}{\partial\ba\trans}
&=&\0\trans,\\
\frac{\partial\theta}{\partial\bb\trans}
&=&\0\trans,\\
\frac{\partial\theta}{\partial\bg\trans}
&=&\Eq\{Y\a_\bg\trans(Y)\}.
\ese

Let $\phi(\x,r,ry)$ be
\bse
\phi(\x,r,ry)\equiv\frac{r}{\pi}\phi_1(\x,y)+\frac{1-r}{1-\pi}\phi_2(\x).
\ese
For $\phi(\x,r,ry)$ to be an influence function,
it must satisfy
\be\label{eq:hilbert}
\E(\phi)=\Ep\{\phi_1(\X,Y)\}+\Eq\{\phi_2(\X)\}=0
\ee
and $\E(\phi\S_\bd\trans)=\partial\theta/\partial\bd\trans$.
$\E(\phi\S_\ba\trans)=\partial\theta/\partial\ba\trans=\0\trans$ is equivalent to
\be\label{eq:alpha}
\E\{\phi_1(\X,y)\mid y\}=c
\ee
for some constant $c$.
In addition, since
\bse
\E(\phi\S_\bb\trans)
&=&\Ep\{\phi_1(\X,Y)\a_\bb\trans(\X,Y)\}+\Eq\{\phi_2(\X)\a_\bb\trans(\X,Y)\}\\
&=&\Ep[\{\phi_1(\X,Y)+\phi_2(\X)\rho(Y)\}\a_\bb\trans(\X,Y)]
\ese
and $\a_\bb(\x,y)$ is an arbitrary function which satisfies $\E\{\a_\bb(\X,y)|y\}=\0$,
$\E(\phi\S_\bb\trans)=\partial\theta/\partial\bb\trans=\0\trans$ implies
$\phi_1(\x,y)+\phi_2(\x)\rho(y)=a(y)$ for some function $a(y)$.
Then \eqref{eq:alpha} yields
\bse
\phi_1(\x,y)=[\E\{\phi_2(\X)\mid y\}-\phi_2(\x)]\rho(y)+c.
\ese
Also, noting that
$\E(\phi\S_\bg\trans)=\partial\theta/\partial\bg\trans=\Eq\{Y\a_\bg\trans(Y)\}$ is equivalent to
$\E\{\phi_2(\X)\mid y\}=y+c^*$ for some constant $c^*$,
we have from \eqref{eq:hilbert} and \eqref{eq:alpha} that
\bse
\E\{\phi_2(\X)\mid y\}=y-\Eq(Y)-c.
\ese
Therefore, defining $b(\x)\equiv \phi_2(\x)+\Eq(Y)+c$,
the summary description of the influence function is
\bse
\phi(\x,r,ry)=\frac{r}{\pi}\{y\rho(y)-b(\x)\rho(y)+c\}+\frac{1-r}{1-\pi}\{b(\x)-\Eq(Y)-c\},
\ese
where $b(\x)$ satisfies $\E\{b(\X)\mid y\}=y$ and $c$ is a constant.
Hence, we get the result \eqref{eq:influence}.
\qed

\subsection{Proof of Proposition \ref{pro:eff}}\label{sec:eff}

Note that
\bse
\frac{\Ep\{a(\x,Y)\rho(Y)\mid\x\}}{\Ep\{\rho(Y)\mid\x\}}
=\frac{\int a(\x,y)\rho(y)g(\x,y)\py(y)dy}{\int \rho(y)g(\x,y)\py(y)dy}
=\frac{\int a(\x,y)g(\x,y)\qy(y)dy}{\int g(\x,y)\qy(y)dy}
=\Eq\{a(\x,Y)\mid\x\},
\ese
then $\phi\eff(\x,r,ry)$ can be alternatively written as
\bse
&&\phi\eff(\x,r,ry)\\
&=&\frac{r}{\pi}\rho(y)
\left[y-\frac{\Eq\{a(Y)\mid\x\}}{\Eq\{\rho(Y)\mid\x\}+\pi/(1-\pi)}\right]
+\frac{1-r}{1-\pi}
\left[\frac{\Eq\{a(Y)\mid\x\}}{\Eq\{\rho(Y)\mid\x\}+\pi/(1-\pi)}-\theta\right],
\ese
where $a(y)$ satisfies
\bse
\E\left[\frac{\Eq\{a(Y)\mid\X\}}{\Eq\{\rho(Y)\mid\X\}+\pi/(1-\pi)}\mid y\right]
=y.
\ese
First, it is immediate that $\phi\eff(\x,r,ry)$ is an influence function for $\theta$,
i.e., belongs to $\cal F$ given in \eqref{eq:influence}
from letting
\bse
b(\x)&\equiv&\frac{\Eq\{a(Y)\mid\x\}}{\Eq\{\rho(Y)\mid\x\}+\pi/(1-\pi)},\\
c&\equiv&0.
\ese
Next, we show that $\phi\eff(\x,r,ry)$ is in the tangent space $\calT$ of \eqref{eq:submodel}.
We decompose $\phi\eff(\x,r,ry)$ into
\bse
\phi\eff(\x,r,ry)
&=&\frac{r}{\pi}\left\{y\rho(y)-b(\x)\rho(y)\right\}
+\frac{1-r}{1-\pi}\left\{b(\x)-\theta\right\}\\
&=&r\{a_1(y)+a_2(\x,y)\}+(1-r)[\Eq\{a_2(\x,Y)\mid\x\}+\Eq\{a_3(Y)\mid\x\}],
\ese
where
\bse
a_1(y)&\equiv&0,\\
a_2(\x,y)&\equiv&\frac{1}{\pi}\{y\rho(y)-b(\x)\rho(y)\},\\
a_3(y)&\equiv&\frac{1}{\pi}\{a(y)-y\rho(y)\}-\frac{1}{1-\pi}\theta.
\ese
Then it is easy to show $ra_1(y)\in\calT_\ba$, and
$ra_2(\x,y)+(1-r)\Eq\{a_2(\x,Y)\mid\x\}\in\calT_\bb$
because $\E\{b(\X)\mid y\}=y$.
Further, $(1-r)\Eq\{a_3(Y)\mid\x\}\in\calT_\bg$ since
\bse
\Eq\{a_3(Y)\}
&=&\frac{1}{\pi}\Eq[\Eq\{a(Y)\mid\X\}]-\frac{1}{\pi}\Eq\{Y\rho(Y)\}
-\frac{1}{1-\pi}\theta\\
&=&\frac{1}{\pi}\Eq[b(\X)\Eq\{\rho(Y)\mid\X\}]+\frac{1}{1-\pi}\Eq\{b(\X)\}
-\frac{1}{\pi}\Eq\{Y\rho(Y)\}
-\frac{1}{1-\pi}\theta\\
&=&0.
\ese
Hence $\phi\eff(\x,r,ry)$ belongs to $\calT$, which proves the result.
\qed

\subsection{Proof of Proposition \ref{pro:consistency}}\label{sec:consistency}
The definition of $b^{*\star}(\x)$ immediately leads to $\E\{b^{*\star}(\X)\mid y\}=y$.
Therefore,
\bse
\wh\theta_t
&\xrightarrow{P}&\E\left[\frac{R}{\pi}\rho^*(Y)\{Y-b^{*\star}(\X)\}+\frac{1-R}{1-\pi}b^{*\star}(\X)\right]\\
&=&\Ep[\rho^*(Y)\{Y-b^{*\star}(\X)\}]+\Eq\{b^{*\star}(\X)\}\\
&=&\Ep[\rho^*(Y)[Y-\E\{b^{*\star}(\X)\mid Y\}])+\Eq[\E\{b^{*\star}(\X)\mid Y\}]\\
&=&\Eq(Y).
\ese
\qed

\subsection{Algorithm to Solve Equations (\ref{eq:a*hat}) and (\ref{eq:a*approx}) in Algorithms~\ref{algdouble} and \ref{algsingle}}\label{sec:fredholm}
We first illustrate how to solve \eqref{eq:a*hat}. Let
\be\label{eq:fredholmkernel}
\phi^{*\star}(y,t)\equiv\rho^*(t)\sumi\frac{\pyx^\star(t,\x_i)}{\Ep^\star\{\rho^{*2}(Y)\mid\x_i\}+\pi/(1-\pi)\Ep^\star\{\rho^*(Y)\mid\x_i\}}\frac{r_iK_h(y-y_i)}{\sumj
  r_jK_h(y-y_j)},
\ee
then \eqref{eq:a*hat} is equivalently written as
 $y=\int \phi^{*\star}(y,t)a^{*\star}(t)\dd t$, which is a Fredholm
 integral equation of the first type. We can find the solution $\wh
 a^{*\star}$ through Landweber's iterative method
 \citep{landweber1951iteration}. Choose a starting point $a_{0}$,
 then iterate the formula
\bse
a_{k+1}(y)\leftarrow a_{k}(y)+\int\phi^{*\star}(t,y)t\dd t -
\int\left\{\int\phi^{*\star}(z,y)\phi^{*\star}(z,t)\dd z\right\}
a_{k}(t)\dd t
\ese
while $\int\{a_{k+1}(t)-a_{k}(t)\}^2\dd t/\int\{a_{k}(t)\}^2\dd t$ is
greater than a chosen tolerance.

For practical implementation, we can approximate the integration in
\eqref{eq:a*hat} using  quadrature method. For a given quadrature
rule, let the quadrature nodes be $\t\equiv(t_1,\dots,t_m)\trans$ and
weights be $\w\equiv(w_1,\dots,w_m)\trans$. Also, let $\wt\y\equiv(\wt
y_1,\dots,\wt y_l)\trans,
\a^{*\star}\equiv\{a^{*\star}(t_1),\dots,a^{*\star}(t_m)\}\trans$, and
$\bPhi^{*\star}$ be an $l\times m$ matrix whose $(i,j)$ component
is $\phi^{*\star}(\wt y_i,t_j)$. Then \eqref{eq:a*hat} can be
discretized as $\wt\y=\bPhi^{*\star}(\a^{*\star}\cdot\w)$ where
$\a\cdot\b\equiv(a_1b_1,\dots,a_mb_m)\trans$ is the Hadamard product. Also, the above
iterative formula can be approximated as
\bse
\a_{k+1}\leftarrow \a_{k}+l^{-1}\bPhi^{*\star\rm T}\wt\y-l^{-1}\bPhi^{*\star\rm T}\bPhi^{*\star}(\a_{k}\cdot\w).
\ese
We iterately use this formula with the starting point $\a_{0}$ as long
as$\|\a_{k+1}-\a_{k}\|_2^2/\|\a_{k}\|_2^2$ is greater than a chosen tolerance.
We summarize the algorithm in Algorithm \ref{algfredholm}.
\begin{algorithm}[H]
    \caption{Solving the integral equation \eqref{eq:a*hat}}\label{algfredholm}
 	\begin{algorithmic}
        \STATE \textbf{Input}: a function $\phi^{*\star}(y,t)$ in \eqref{eq:fredholmkernel} and a tolerance $\Delta$.
 	    \STATE \textbf{do}
        \STATE (a) adopt a set of evaluation points $\wt\y=(\wt y_1,\dots,\wt y_l)\trans$;\\
        \STATE (b) adopt a quadrature rule $(\t,\w)$, where nodes $\t=(t_1,\dots,t_m)\trans$ and     weights $\w=(w_1,\dots,w_m)\trans$;\\
        \STATE (c) compute a matrix $\bPhi^{*\star}$ where $\bPhi^{*\star}_{(i,j)}=\phi^{*\star}(\wt     y_i,t_j),i=1,\dots,l,j=1,\dots,m$;\\
        \STATE (d) Declare a starting point $\a_0=(a_{01},\dots,a_{0m})\trans$;\\
        \While{$\|\a_{k+1}-\a_{k}\|_2^2/\|\a_{k}\|_2^2>\Delta$}{
            $\a_{k+1}\leftarrow \a_{k}+l^{-1}\bPhi^{*\star\rm T}\wt\y-l^{-1}\bPhi^{*\star\rm     T}\bPhi^{*\star}(\a_{k}\cdot\w)$;}
 	    \STATE \textbf{Output}: $\wh\a^{*\star}\leftarrow\a_{k+1}$.
    \end{algorithmic}
\end{algorithm}

We now describe how to solve \eqref{eq:a*approx}. If
$\wh\E_p\{a(Y)\mid\x\}=\int a(t)\wh p_{Y\mid\X}(t,\x_i)\dd t$ for some
$\wh p_{Y\mid\X}$, then we can follow the above procedure to obtain
the solution while replacing $\pyx^\star(t,\x_i)$ in
\eqref{eq:fredholmkernel} by $\wh p_{Y\mid\X}(t,\x_i)$. Now, suppose
$\wh\E_p\{a(Y)\mid\x\}=\sum_{i=1}^{n_1} a(y_i)w_i(\x)$ for some
$w_i,i=1,\dots,n_1$. This form includes a general class of
nonparametric regression estimators, for instance, the Nadaraya-Watson
estimator is of this form with $w_i(\x)=r_iK_h(\x-\x_i)/\sumj
r_jK_h(\x-\x_j)$. Then \eqref{eq:a*approx} is equivalent to
\bse
y&=&\sumi\frac{\sum_{k=1}^{n_1} a^{*}(y_k)\rho^*(y_k)w_k(\x_i)}
{\wh\E_p\{\rho^{*2}(Y)\mid\x_i\}+\pi/(1-\pi)\wh\E_p\{\rho^*(Y)\mid\x_i\}}
\frac{r_iK_h(y-y_i)}{\sumj r_jK_h(y-y_j)}\n\\
&=&\sum_{k=1}^{n_1} \phi^*(y,y_k,w_k)a^{*}(y_k),
\ese
where
\be\label{eq:fredholmkernel2}
\phi^*(y,y_k,w_k)\equiv\sumi\frac{\rho^*(y_k)w_k(\x_i)}
{\wh\E_p\{\rho^{*2}(Y)\mid\x_i\}+\pi/(1-\pi)\wh\E_p\{\rho^*(Y)\mid\x_i\}}
\frac{r_iK_h(y-y_i)}{\sumj r_jK_h(y-y_j)}.
\ee
Now, let $\wt\y\equiv(\wt y_1,\dots,\wt y_l)\trans,
\a^*\equiv\{a^*(y_1),\dots,a^*(y_{n_1})\}\trans$, and $\bPhi^*$ be a
$l\times n_1$ matrix whose $(i,j)$th component is $\phi^*(\wt
y_i,y_j,w_j)$. Then we can find the approximate solution $\wh\a^*$ via
iteratively using the formula
\bse
\a_{k+1}\leftarrow \a_{k}+l^{-1}\bPhi^{*\rm T}\wt\y-l^{-1}\bPhi^{*\rm T}\bPhi^*\a_{k}
\ese
with a starting point $\a_{0}$,
as long as $\|\a_{k+1}-\a_{k}\|_2^2/\|\a_{k}\|_2^2$ is greater than a chosen tolerance.
We summarize the algorithm in Algorithm \ref{algfredholm2}.
\begin{algorithm}[H]
    \caption{Solving the integral equation \eqref{eq:a*approx}}\label{algfredholm2}
 	\begin{algorithmic}
        \STATE \textbf{Input}: a function $\phi^*(y,y_k,w_k)$ in \eqref{eq:fredholmkernel2} and a tolerance $\Delta$.
 	    \STATE \textbf{do}
        \STATE (a) adopt a set of evaluation points $\wt\y=(\wt y_1,\dots,\wt y_l)\trans$;\\
        \STATE (b) compute a matrix $\bPhi^{*}$ where $\bPhi^{*}_{(i,j)}=\phi^{*}(\wt     y_i,y_j,w_j),i=1,\dots,l,j=1,\dots,n_1$;\\
        \STATE (c) Declare a starting point $\a_0=(a_{01},\dots,a_{0m})\trans$;\\
        \While{$\|\a_{k+1}-\a_{k}\|_2^2/\|\a_{k}\|_2^2>\Delta$}{
            $\a_{k+1}\leftarrow \a_{k}+l^{-1}\bPhi^{*\rm T}\wt\y-l^{-1}\bPhi^{*\rm     T}\bPhi^{*}(\a_{k}\cdot\w)$;}
 	    \STATE \textbf{Output}: $\wh\a^{*}\leftarrow\a_{k+1}$.
    \end{algorithmic}
\end{algorithm}

\subsection{Proof of Lemma \ref{lem:l}}\label{sec:ldoubleproof}
From the definition of $\calL^{*\star}(a)$, Conditions \ref{con:3boundedstar}, and
\ref{con:4compact},
it is immediate that there exists a constant $0<c_2<\infty$
such that $\|\calL^{*\star}(a)\|_\infty\leq c_2\|a\|_\infty$.
Now we show there is a constant $0<c_1<\infty$
such that $c_1\|a\|_\infty\leq\|\calL^{*\star}(a)\|_\infty$.
Note that if $\calL^{*\star}$ is invertible,
by the bounded inverse theorem
we have $\|\calL^{*\star-1}(v)\|_\infty\leq c_1^{-1}\|v\|_\infty$
for some constant $0<c_1<\infty$, i.e.,
$c_1\|a\|_\infty\leq\|\calL^{*\star}(a)\|_\infty$.
Hence it suffices to show that $\calL^{*\star}$ is invertible.
We prove this by contradiction.
Suppose there are $a_1(y)$ and $a_2(y)$ such that
$\calL^{*\star}(a_1)(y)=\calL^{*\star}(a_2)(y)=v(y)$ and $a_1\neq a_2$.
Then by Conditions \ref{con:1complete} and \ref{con:2rho}, we have
\bse
\Ep^\star\{a_1(Y)\rho^*(Y)\mid\x\}\neq\Ep^\star\{a_2(Y)\rho^*(Y)\mid\x\}.
\ese
Now, the efficient score  calculated under the posited models is
\bse
\phi\eff^{*\star}(\x,r,ry)
&=&\frac{r}{\pi}\rho^*(y)\left[y
-\frac{\Ep^\star\{a(Y)\rho^*(Y)\mid\x\}}
{\Ep^\star\{\rho^{*2}(Y)\mid\x\}+\pi/(1-\pi)\Ep^\star\{\rho^*(Y)\mid\x\}}\right]\\
&&+\frac{1-r}{1-\pi}
\left[\frac{\Ep^\star\{a(Y)\rho^*(Y)\mid\x\}}
{\Ep^\star\{\rho^{*2}(Y)\mid\x\}+\pi/(1-\pi)\Ep^\star\{\rho^*(Y)\mid\x\}}-\theta\right],
\ese
where $a(y)$ satisfies $\calL^{*\star}(a)(y)=v(y)$.
Then letting $a=a_1$ and $a=a_2$ gives two distinct efficient scores,
which contradicts the uniqueness of the efficient score.
Therefore, there is a unique solution $a^{*\star}(y)$ for $\calL^{*\star}(a^{*\star})(y)=v(y)$,
hence $\calL^{*\star}$ is invertible.
\qed

\subsection{Proof of Theorem \ref{th:theta}}\label{sec:doubleproof}
We define
\bse
b^{*\star}(\x,a,\bzeta)\equiv\frac{\Ep^\star\{a(Y)\rho^*(Y)\mid\x,\bzeta\}}
{\Ep^\star\{\rho^{*2}(Y)\mid\x,\bzeta\}+\pi/(1-\pi)\Ep^\star\{\rho^*(Y)\mid\x,\bzeta\}},
\ese
and analyze $b^{*\star}(\x,\wh a^{*\star},\wh\bzeta)$.
First, under Conditions
\ref{con:3boundedstar}-\ref{con:6bandwidth}
and $\|\wh\bzeta-\bzeta\|_2=O_p(n_1^{-1/2})$,
we have
\be\label{eq:lhatastar}
\wh\calL^{*\star}(a^{*\star})(y)
&=&n_1^{-1}\sumi r_iK_h(y-y_i)b^{*\star}(\x_i,a^{*\star},\wh\bzeta)\n\\
&=&n_1^{-1}\sumi r_iK_h(y-y_i)\left\{b^{*\star}(\x_i,a^{*\star},\bzeta)
+\frac{\partial b^{*\star}(\x_i,a^{*\star},\bzeta)}{\partial\bzeta\trans}
(\wh\bzeta-\bzeta)+o_p(n_1^{-1/2})\right\}\n\\
&=&n_1^{-1}\sumi\calL_{i,h}^{*\star}(a^{*\star})(y)
+\left[\py(y)\frac{\partial E\{b^{*\star}(\X,a^{*\star},\bzeta)\mid y\}}{\partial\bzeta\trans}+o_p(1)\right]
(\wh\bzeta-\bzeta)+o_p(n_1^{-1/2})\n\\
&=&n_1^{-1}\sumi\calL_{i,h}^{*\star}(a^{*\star})(y)
+o_p(n_1^{-1/2})
\ee
uniformly in $y$ since $E\{b^{*\star}(\X,a^{*\star},\bzeta)\mid y\}=y$.
Here, we define
\bse
\calL_{i,h}^{*\star}(a)(y) \equiv r_i K_h(y-y_i)\frac{\Ep^\star\{a(Y)\rho^*(Y)\mid\x_i\}}
{\Ep^\star\{\rho^{*2}(Y)\mid\x_i\}+\pi/(1-\pi)\Ep^\star\{\rho^*(Y)\mid\x_i\}}.
\ese
We also define $g_{i,h}(y)\equiv\calL^{*\star-1}\{v_{i,h}-\calL_{i,h}^{*\star}(a^{*\star})\}(y)$.
Note that
$\|E(V_{i,h})-v\|_\infty=O(h^2)$ by Conditions \ref{con:4compact},\ref{con:5kernel},
then since $\calL^{*\star-1}$ is a bounded linear operator by Lemma \ref{lem:l},
\bse
\|E\{\calL^{*\star-1}(V_{i,h})\}-\calL^{*\star-1}(v)\|_\infty
&=&\|\calL^{*\star-1}\{E(V_{i,h})\}-\calL^{*\star-1}(v)\|_\infty\\
&=&\|\calL^{*\star-1}\{E(V_{i,h})-v\}\|_\infty\\
&=&O(h^2).
\ese
Similarly, $\|E\{(\calL^{*\star-1}\calL_{i,h}^{*\star})(a^{*\star})\}-a^{*\star}\|_\infty=O(h^2)$.
Then using $\calL^{*\star-1}(v)(y)=a^{*\star}(y)$, we get
\be\label{eq:g}
\|E(G_{i,h})\|_\infty
&=&\|E\{\calL^{*\star-1}(V_{i,h})\}-E\{(\calL^{*\star-1}\calL_{i,h}^{*\star})(a^{*\star})\}-\{\calL^{*\star-1}(v)-a^{*\star}\}\|_\infty\n\\
&\leq&\|E\{\calL^{*\star-1}(V_{i,h})\}-\calL^{*\star-1}(v)\|_\infty+\|E\{(\calL^{*\star-1}\calL_{i,h})(a^{*\star})\}-a^{*\star}\|_\infty\n\\
&=&O(h^2).
\ee
Now, for any bounded function $a(y)$,
$\|(\wh\calL^{*\star}-\calL^{*\star})(a)\|_\infty=O_p\left\{(n_1h)^{-1/2}\log
  n_1+h^2\right\}=o_p(n_1^{-1/4})$
under Conditions
\ref{con:3boundedstar}-\ref{con:6bandwidth}
and the assumption $\|\wh\bzeta-\bzeta\|_2=O_p(n_1^{-1/2})$.
Similarly, $\|\wh v-v\|_\infty=o_p(n_1^{-1/4})$.
Then using that $\calL^{*\star-1}$ is a bounded linear operator by Lemma \ref{lem:l},
$\wh a^{*\star}(y)$ can be expressed as
\bse
\wh a^{*\star}(y)&=&\wh\calL^{*\star-1}(\wh v)(y)\\
&=&\{\calL^{*\star}+(\wh\calL^{*\star}-\calL^{*\star})\}^{-1}(\wh v)(y)\\
&=&\{\calL^{*\star-1}-\calL^{*\star-1}(\wh\calL^{*\star}-\calL^{*\star})\calL^{*\star-1}\}\{v+(\wh v-v)\}(y)+o_p(n_1^{-1/2})\\
&=&a^{*\star}(y)+\calL^{*\star-1}(\wh v-v)(y)-\{\calL^{*\star-1}(\wh\calL^{*\star}-\calL^{*\star})\}(a^{*\star})(y)+o_p(n_1^{-1/2})\\
&=&a^{*\star}(y)+\calL^{*\star-1}\{\wh v-\wh\calL^{*\star}(a^{*\star})\}(y)+o_p(n_1^{-1/2})\\
&=&a^{*\star}(y)+n_1^{-1}\sumi g_{i,h}(y)+o_p(n_1^{-1/2})
\ese
uniformly in $y$, where the last step is by \eqref{eq:lhatastar} and the definition of
$g_{i,h}(y)$.
Then we further get
\bse
&&\left|\frac{\partial b^{*\star}(\x,\wh a^{*\star},\bzeta)}
{\partial\bzeta\trans}(\wh\bzeta-\bzeta)
-\frac{\partial b^{*\star}(\x,a^{*\star},\bzeta)}
{\partial\bzeta\trans}(\wh\bzeta-\bzeta)\right|\\
&=&\left|\frac{\partial}{\partial\bzeta\trans}
\left[\frac{\Ep^\star\{(\wh a^{*\star}-a^{*\star})(Y)\rho^*(Y)\mid\x,\bzeta\}}
{\Ep^\star\{\rho^{*2}(Y)+\pi/(1-\pi)\rho^*(Y)\mid\x,\bzeta\}}\right](\wh\bzeta-\bzeta)\right|\\
&=&\left|\frac{\Ep^\star\{(\wh a^{*\star}-a^{*\star})(Y)\rho^*(Y)
\S_\bzeta^{\star\rm T}(Y,\x,\bzeta)\mid\x\}}
{\Ep^\star\{\rho^{*2}(Y)+\pi/(1-\pi)\rho^*(Y)\mid\x\}}(\wh\bzeta-\bzeta)\right.\\
&&\left.-\Ep^\star\{(\wh a^{*\star}-a^{*\star})(Y)\rho^*(Y)\mid\x\}
\frac{\Ep[\{\rho^{*2}(Y)+\pi/(1-\pi)\rho^*(Y)\}\S_\bzeta^{\star\rm T}(Y,\x,\bzeta)\mid\x]}
{[\Ep^\star\{\rho^{*2}(Y)+\pi/(1-\pi)\rho^*(Y)\mid\x\}]^2}(\wh\bzeta-\bzeta)\right|\\
&\leq&\left\|\wh a^{*\star}-a^{*\star}\right\|_\infty
\frac{\Ep^\star\{\rho^*(Y)\|\S_\bzeta^\star(Y,\x,\bzeta)\|_2\mid\x\}\|\wh\bzeta-\bzeta\|_2}
{\Ep^\star\{\rho^{*2}(Y)+\pi/(1-\pi)\rho^*(Y)\mid\x\}}\\
&&+\left\|\wh a^{*\star}-a^{*\star}\right\|_\infty\Ep^\star\{\rho^*(Y)\mid\x\}
\frac{\Ep^\star[\{\rho^{*2}(Y)+\pi/(1-\pi)\rho^*(Y)\}
\|\S_\bzeta^\star(Y,\x,\bzeta)\|_2\mid\x]\|\wh\bzeta-\bzeta\|_2}
{[\Ep^\star\{\rho^{*2}(Y)+\pi/(1-\pi)\rho^*(Y)\mid\x\}]^2}\\
&=&o_p(n_1^{-1/4})O_p(n_1^{-1/2})
=o_p(n_1^{-1/2}),
\ese
because $\|\wh\bzeta-\bzeta\|_2=O_p(n_1^{-1/2})$,
$\Ep^\star\{\|\S_\bzeta^\star(Y,\x,\bzeta)\|_2\mid\x\}$ is bounded,
and $\|\wh a^{*\star}-a^{*\star}\|_\infty
=O_p\left\{(n_1h)^{-1/2}\log n_1+h^2\right\}=o_p(n^{-1/4})$
by \eqref{eq:g} and Condition \ref{con:6bandwidth}.
Hence, noting that $b^{*\star}(\x,a,\bzeta)$ is linear with respect to $a$, we get
\be\label{eq:bpar}
b^{*\star}(\x,\wh a^{*\star},\wh\bzeta)
&=&b^{*\star}(\x,\wh a^{*\star},\bzeta)
+\frac{\partial b^{*\star}(\x,\wh a^{*\star},\bzeta)}{\partial\bzeta\trans}(\wh\bzeta-\bzeta)+o_p(n_1^{-1/2})\n\\
&=&b^{*\star}(\x,a^{*\star},\bzeta)+b^{*\star}(\x,\wh a^{*\star}-a^{*\star},\bzeta)
+\frac{\partial b^{*\star}(\x,a^{*\star},\bzeta)}{\partial\bzeta\trans}(\wh\bzeta-\bzeta)+o_p(n_1^{-1/2})\n\\
&=&b^{*\star}(\x,a^{*\star},\bzeta)+n_1^{-1}\sumi b^{*\star}(\x,g_{i,h},\bzeta)
+\frac{\partial b^{*\star}(\x,a^{*\star},\bzeta)}{\partial\bzeta\trans}(\wh\bzeta-\bzeta)+o_p(n_1^{-1/2})
\ee
uniformly in $\x$ by Condition \ref{con:4compact}.
Now we analyze $\wh\theta$. By the definition of $\wh\theta$,
\be\label{eq:theta}
\sqrt{n_1}(\wh\theta-\theta)
&=&\sqrt{n_1}n^{-1}\sumi\left[\frac{r_i}{\pi}\rho^*(y_i)\{y_i-b^{*\star}(\x_i,\wh a^{*\star},\wh\bzeta)\}
+\frac{1-r_i}{1-\pi}\{b^{*\star}(\x_i,\wh a^{*\star},\wh\bzeta)-\theta\}\right]\n\\
&=&\sqrt{n_1}n^{-1}\sumi\phi\eff^{*\star}(\x_i,r_i,r_iy_i)
+\sqrt{n_1}n^{-1}\sumi\left\{\frac{r_i}{\pi}\rho^*(y_i)-\frac{1-r_i}{1-\pi}\right\}
\{b^{*\star}(\x_i,a^{*\star},\bzeta)-b^{*\star}(\x_i,\wh a^{*\star},\wh\bzeta)\}\n\\
&=&\sqrt{n_1}n^{-1}\sumi\phi\eff^{*\star}(\x_i,r_i,r_iy_i)-T_1-T_2+o_p(1),
\ee
where
\bse
T_1&\equiv&n_1^{-1/2}n^{-1}\sumi\sumj
\left\{\frac{r_i}{\pi}\rho^*(y_i)-\frac{1-r_i}{1-\pi}\right\}b^{*\star}(\x_i,g_{j,h},\bzeta),\\
T_2&\equiv&n^{-1}\sumi\left\{\frac{r_i}{\pi}\rho^*(y_i)-\frac{1-r_i}{1-\pi}\right\}
\frac{\partial b^{*\star}(\x_i,a^{*\star},\bzeta)}{\partial\bzeta\trans}\sqrt{n_1}(\wh\bzeta-\bzeta).
\ese
Note that since $\|E(G_{j,h})\|_\infty=O(h^2)$ by \eqref{eq:g}
and $b^{*\star}(\x,a,\bzeta)$ is linear with respect to $a$,
\be\label{eq:eg}
\|E\{b^{*\star}(\x,G_{j,h},\bzeta)\}\|_\infty=\|b^{*\star}\{\x,E(G_{j,h}),\bzeta\}\|_\infty=O(h^2).
\ee
Hence using the property of the U-statistic and Condition \ref{con:6bandwidth},
we can rewrite $T_1$ as
\bse
T_1
&=&n_1^{-1/2}\sumi
\left\{\frac{r_i}{\pi}\rho^*(y_i)-\frac{1-r_i}{1-\pi}\right\}
E\{b^{*\star}(\x_i,G_{j,h},\bzeta)\mid\x_i,r_i,r_iy_i\}\n\\
&&+n_1^{-1/2}\sumj
\E\left[\left\{\frac{R_i}{\pi}\rho^*(Y_i)-\frac{1-R_i}{1-\pi}\right\}
b^{*\star}(\X_i,g_{j,h},\bzeta)\mid\x_j,r_j,r_jy_j\right]\n\\
&&-n^{1/2}n_1^{-1/2}
\E\left[\left\{\frac{R_i}{\pi}\rho^*(Y_i)-\frac{1-R_i}{1-\pi}\right\}
b^{*\star}(\X_i,G_{j,h},\bzeta)\right]+O_p(n_1^{-1/2})\n\\
&=&n_1^{-1/2}\sumj
\E\left[\left\{\frac{R_i}{\pi}\rho^*(Y_i)-\frac{1-R_i}{1-\pi}\right\}
b^{*\star}(\X_i,g_{j,h},\bzeta)\mid\x_j,r_j,r_jy_j\right]+O_p(n_1^{-1/2}+n n_1^{-1/2}h^2)\n\\
&=&n_1^{-1/2}\sumj
\int\left[\rho^*(y)\E\{b^{*\star}(\X,g_{j,h},\bzeta)\mid y\}\py(y)
-\E\{b^{*\star}(\X,g_{j,h},\bzeta)\mid y\}\qy(y)\right]dy+o_p(1)\n\\
&=&n_1^{-1/2}\sumj\int\calL^{*\star}(g_{j,h})(y)
\left\{\rho^*(y)-\rho(y)\right\}dy+o_p(1)\n\\
&=&n_1^{-1/2}\sumj\int\{v_{j,h}(y)-\calL_{j,h}^{*\star}(a^{*\star})(y)\}
\left\{\rho^*(y)-\rho(y)\right\}dy+o_p(1)\n\\
&=&n_1^{-1/2}\sumj r_j\int K_h(y-y_j)\{y-b^{*\star}(\x_j,a^{*\star},\bzeta)\}
\left\{\rho^*(y)-\rho(y)\right\}dy+o_p(1).
\ese
On the other hand, using $\E\{b^{*\star}(\X,a^{*\star},\bzeta)\mid y\}=y$, $T_2$ can be written as
\bse
T_2
&=&\E\left[\left\{\frac{R}{\pi}\rho^*(Y)-\frac{1-R}{1-\pi}\right\}
\frac{\partial b^{*\star}(\X,a^{*\star},\bzeta)}{\partial\bzeta\trans}\right]
\sqrt{n_1}(\wh\bzeta-\bzeta)+o_p(1)\\
&=&\frac{\partial}{\partial\bzeta\trans}
\E\left[\left\{\frac{R}{\pi}\rho^*(Y)-\frac{1-R}{1-\pi}\right\}Y\right]
\sqrt{n_1}(\wh\bzeta-\bzeta)+o_p(1)\\
&=&o_p(1).
\ese
Therefore, \eqref{eq:theta} leads to
\bse
\sqrt{n_1}(\wh\theta-\theta)
&=&\sqrt{n_1}n^{-1}\sumi\phi\eff^{*\star}(\x_i,r_i,r_iy_i)+n_1^{-1/2}\sumi r_ih(\x_i,y_i)+o_p(1)\\
&=&n^{-1/2}\sumi\left\{\sqrt{\pi}\phi\eff^{*\star}(\x_i,r_i,r_iy_i)+\frac{r_i}{\sqrt{\pi}}h(\x_i,y_i)\right\}+o_p(1),
\ese
where
\bse
h(\x_i,y_i)
&\equiv&\int K_h(y-y_i)\{b^{*\star}(\x_i,a^{*\star},\bzeta)-y\}\{\rho^*(y)-\rho(y)\}dy\n\\
&=&\int K(t)\{b^{*\star}(\x_i,a^{*\star},\bzeta)-y_i-ht\}\{\rho^*(y_i+ht)-\rho(y_i+ht)\}dt\n\\
&=&\{b^{*\star}(\x_i,a^{*\star},\bzeta)-y_i\}\{\rho^*(y_i)-\rho(y_i)\}
+[
\{b^{*\star}(\x_i,a^{*\star},\bzeta)-y_i\}\{{\rho^*}''(y_i)-\rho''(y_i)\}
-2\{{\rho^*}'(y_i)-\rho'(y_i)\}]\n\\
&&\times\frac{h^2}{2}\int t^2K(t)dt+O(h^4)
\ese
by Conditions \ref{con:2rho} and \ref{con:5kernel}.
Hence, by Condition \ref{con:6bandwidth}, we have
\bse
\sqrt{n_1}(\wh\theta-\theta)
&=&n^{-1/2}\sumi\left[\sqrt{\pi}\phi\eff^{*\star}(\x_i,r_i,r_iy_i)
+\frac{r_i}{\sqrt{\pi}}\{b^{*\star}(\x_i,a^{*\star},\bzeta)-y_i\}\{\rho^*(y_i)-\rho(y_i)\}\right]+o_p(1).
\ese
\qed

\subsection{Proof of Theorem \ref{th:theta2}}\label{sec:singleproof}
First let
\bse
b^*(\x,a,\Ep)&\equiv&
\frac{\Ep\{a(Y)\rho^*(Y)\mid\x\}}{\Ep\{\rho^{*2}(Y)+\pi/(1-\pi)\rho^*(Y)\mid\x\}},
\ese
and for any function $g(\cdot,\Ep)$,
its $k$th Gateaux derivative with respect to $\Ep$ at $\mu_1$ in the direction $\mu_2$ as
\bse
\frac{\partial^k g(\cdot,\mu_1)}{\partial{\Ep}^k}(\mu_2)&\equiv&
\frac{\partial^k g(\cdot,\mu_1+h\mu_2)}{\partial h^k}\bigg|_{h=0}.
\ese
We have
\bse
\frac{\partial b^*(\x,a,\Ep)}{\partial\Ep}(\wh\E_p-\Ep)
&=&\frac{(\wh\E_p-\Ep)\{a(Y)\rho^*(Y)\mid\x\}\Ep\{\rho^{*2}(Y)+\pi/(1-\pi)\rho^*(Y)\mid\x\}}
{[\Ep\{\rho^{*2}(Y)+\pi/(1-\pi)\rho^*(Y)\mid\x\}]^2}\\
&&-\frac{(\wh\E_p-\Ep)\{\rho^{*2}(Y)+\pi/(1-\pi)\rho^*(Y)\mid\x\}\Ep\{a(Y)\rho^*(Y)\mid\x\}}
{[\Ep\{\rho^{*2}(Y)+\pi/(1-\pi)\rho^*(Y)\mid\x\}]^2}\\
&=&b^*(\x,a,\Ep)\left[\frac{(\wh\E_p-\Ep)\{a(Y)\rho^*(Y)\mid\x\}}{\Ep\{a(Y)\rho^*(Y)\mid\x\}}
-\frac{(\wh\E_p-\Ep)\{\rho^{*2}(Y)+\pi/(1-\pi)\rho^*(Y)\mid\x\}}
{\Ep\{\rho^{*2}(Y)+\pi/(1-\pi)\rho^*(Y)\mid\x\}}\right]\\
&=&b^*(\x,a,\Ep)o_p(n_1^{-1/4}),\\
\frac{\partial^2 b^*(\x,a,\mu)}{\partial{\Ep}^2}(\wh\E_p-\Ep)
&=&\frac{-2(\wh\E_p-\Ep)\{\rho^{*2}(Y)+\pi/(1-\pi)\rho^*(Y)\mid\x\}}
{[\Ep\{\rho^{*2}(Y)+\pi/(1-\pi)\rho^*(Y)\mid\x\}]^3}\\
&&\times\left[(\wh\E_p-\Ep)\{a(Y)\rho^*(Y)\mid\x\}\mu\{\rho^{*2}(Y)+\pi/(1-\pi)\rho^*(Y)\mid\x\}\right.\\
&&\left.-(\wh\E_p-\Ep)\{\rho^{*2}(Y)+\pi/(1-\pi)\rho^*(Y)\mid\x\}\mu\{a(Y)\rho^*(Y)\mid\x\}\right]\\
&=&o_p(n_1^{-1/2})
\ese
for any bounded $a(y)$ and $\mu(\cdot\mid\x)$
since $|(\wh\E_p-\Ep)(\cdot\mid\x)|=o_p(n_1^{-1/4})$ by the assumption,
and these hold uniformly with respect to $\x$ by Condition \ref{con:4compact}.
Then by the Taylor expansion and mean value theorem, for
any bounded $a(y)$ and some $\alpha\in(0,1)$,
\be\label{eq:b}
b^*(\x,a,\wh\E_p)
&=&b^*(\x,a,\Ep)
+\frac{\partial b^*(\x,a,\Ep)}{\partial{\Ep}}(\wh\E_p-\Ep)
+\frac{1}{2}\frac{\partial^2b^*\{\x,a,\Ep+\alpha(\wh\E_p-\Ep)\}}{\partial{\Ep}^2}(\wh\E_p-\Ep)\n\\
&=&b^*(\x,a,\Ep)+\frac{\partial b^*(\x,a,\Ep)}{\partial{\Ep}}(\wh\E_p-\Ep)+o_p(n_1^{-1/2}).
\ee
Noting that $a^*(y)$ is bounded under Condition \ref{con:3bounded}, we further get
\bse
\wh\calL^*(a^*)(y)
&=&n_1^{-1}\sumi r_iK_h(y-y_i)b^*(\x_i,a^*,\wh\E_p)\n\\
&=&n_1^{-1}\sumi r_iK_h(y-y_i)
\left\{b^*(\x_i,a^*,\Ep)
+\frac{\partial b^*(\x_i,a^*,\Ep)}{\partial{\Ep}}(\wh\E_p-\Ep)+o_p(n_1^{-1/2})\right\}\n\\
&=&n_1^{-1}\sumi\calL_{i,h}^*(a^*)(y)+o_p(n_1^{-1/2})
\ese
uniformly in $y$ by Condition \ref{con:4compact}, where we define
\bse
\calL_{i,h}^*(a)(y) \equiv r_i K_h(y-y_i)\frac{\Ep\{a(Y)\rho^*(Y)\mid\x_i\}}
{\Ep\{\rho^{*2}(Y)\mid\x_i\}+\pi/(1-\pi)\Ep\{\rho^*(Y)\mid\x_i\}}.
\ese
The last equality above is because $E\{b^*(\X,a^*,\Ep)\mid y\}=y$ from
the definition of $a^*$, hence
\bse
&&n_1^{-1}\sumi r_iK_h(y-y_i)\frac{\partial b^*(\x_i,a^*,\Ep)}{\partial{\Ep}}(\wh\E_p-\Ep)\\
&=&n_1^{-1}\sumi r_iK_h(y-y_i)\frac{\partial b^*(\x_i,a^*,\Ep)}{\partial{\Ep}}(\wh\E_p-\Ep)
-\frac{\partial[\py(y)E\{b^*(\X,a^*,\Ep)\mid y\}]}{\partial{\Ep}}(\wh\E_p-\Ep)\\
&=&n_1^{-1/4}\left[n_1^{-1}\sumi r_iK_h(y-y_i)n_1^{1/4}\frac{\partial b^*(\x_i,a^*,\Ep)}{\partial{\Ep}}(\wh\E_p-\Ep)
-\py(y)E\left\{n_1^{1/4}\frac{\partial b^*(\X,a^*,\Ep)}{\partial{\Ep}}(\wh\E_p-\Ep)\mid y\right\}\right]\\
&=&n_1^{-1/4}O_p\left\{(n_1h)^{-1/2}\log{n_1}+h^2\right\}\\
&=&o_p(n_1^{-1/2})
\ese
under Conditions \ref{con:4compact}-\ref{con:3bounded}.
In addition, for any bounded function $a(y)$,
\bse
\|(\wh\calL^*-\calL^*)(a)\|_\infty
=O_p\left\{(n_1h)^{-1/2}\log{n_1}+h^2\right\}+o_p(n_1^{-1/4})
=o_p(n_1^{-1/4})
\ese
under Conditions \ref{con:4compact}-\ref{con:3bounded},
and the assumption $|\wh\E_p\{a(Y)\mid\x\}-\Ep\{a(Y)\mid\x\}|=o_p(n_1^{-1/4})$.
Similarly, $\|\wh v-v\|_\infty=o_p(n_1^{-1/4})$.
Then using that $\calL^{*-1}$ is a bounded linear operator by Lemma \ref{lem:l},
$\wh a^*(y)$ can be expressed as
\bse
\wh a^*(y)
&=&\{\calL^*+(\wh\calL^*-\calL^*)\}^{-1}(\wh v)(y)\\
&=&\{\calL^{*-1}-\calL^{*-1}(\wh\calL^*-\calL^*)\calL^{*-1}\}\{v+(\wh v-v)\}(y)+o_p(n_1^{-1/2})\\
&=&a^*(y)+\calL^{*-1}(\wh v-v)(y)-\{\calL^{*-1}(\wh\calL^*-\calL^*)\}(a^*)(y)+o_p(n_1^{-1/2})\\
&=&a^*(y)+\calL^{*-1}\{\wh v-\wh\calL^*(a^*)\}(y)+o_p(n_1^{-1/2})\\
&=&a^*(y)+n_1^{-1}\sumi g_{i,h}(y)+o_p(n_1^{-1/2})
\ese
uniformly in $y$ where $g_{i,h}(y)\equiv\calL^{*-1}\{v_{i,h}-\calL_{i,h}^*(a^*)\}(y)$.
Hence, noting that $b^*(\x,a,\Ep)$ is linear with respect to $a$ and
using \eqref{eq:b}, we get
\bse
b^*(\x,\wh a^*,\wh\E_p)
&=&b^*(\x,a^*,\Ep)+b^*(\x,\wh a^*-a^*,\Ep)
+\frac{\partial b^*(\x,\wh a^*,\Ep)}{\partial{\Ep}}(\wh\E_p-\Ep)+o_p(n_1^{-1/2})\\
&=&b^*(\x,a^*,\Ep)+n_1^{-1}\sumi b^*(\x,g_{i,h},\Ep)
+\frac{\partial b^*(\x,a^*,\Ep)}{\partial{\Ep}}(\wh\E_p-\Ep)+o_p(n_1^{-1/2})
\ese
uniformly in $\x$ by Condition \ref{con:4compact}.
The last equality holds because for any bounded $a(y)$,
\bse
|b^*(\x,a,\Ep)|\leq
\|a\|_\infty\frac{\Ep\{\rho^*(Y)\mid\x\}}{\Ep\{\rho^{*2}(Y)+\pi/(1-\pi)\rho^*(Y)\mid\x\}}
=O(\|a\|_\infty),
\ese
and
\bse
\frac{\partial b^*(\x,\wh a^*,\Ep)}{\partial{\Ep}}(\wh\E_p-\Ep)
-\frac{\partial b^*(\x,a^*,\Ep)}{\partial{\Ep}}(\wh\E_p-\Ep)
&=&\frac{\partial b^*(\x,\wh a^*-a^*,\Ep)}{\partial{\Ep}}(\wh\E_p-\Ep)\\
&=&b^*(\x,\wh a^*-a^*,\Ep)o_p(n_1^{-1/4})\\
&=&\left[O_p(h^2)+O_p\{(n_1h)^{-1/2}\log (n_1)\}\right]o_p(n_1^{-1/4})\\
&=&o_p(n_1^{-1/2}),
\ese
where the third equality is due to $\|E(G_{i,h})\|_\infty=O(h^2)$ by \eqref{eq:g},
and the last equality holds because $h^2=o_p(n_1^{-1/2})$
and $n_1\{\log(n_1)\}^{-4}h^2\to\infty$
by Condition \ref{con:6bandwidth}.
Now, by the definition of $\wt\theta$,
\be\label{eq:theta2}
&&\sqrt{n_1}(\wt\theta-\theta)\n\\
&=&\sqrt{n_1}n^{-1}\sumi\left[\frac{r_i}{\pi}\rho^*(y_i)\{y_i-b^*(\x_i,\wh a^*,\wh\E_p)\}
+\frac{1-r_i}{1-\pi}\{b^*(\x_i,\wh a^*,\wh\E_p)-\theta\}\right]\n\\
&=&\sqrt{n_1}n^{-1}\sumi\phi\eff^*(\x_i,r_i,r_iy_i)
+\sqrt{n_1}n^{-1}\sumi\left\{\frac{r_i}{\pi}\rho^*(y_i)-\frac{1-r_i}{1-\pi}\right\}
\left\{b^*(\x_i,a^*,\Ep)-b^*(\x_i,\wh a^*,\wh\E_p)\right\}\n\\
&=&\sqrt{n_1}n^{-1}\sumi\phi\eff^*(\x_i,r_i,r_iy_i)-T_1-T_2+o_p(1),
\ee
where
\bse
T_1&\equiv&n_1^{-1/2}n^{-1}\sumi\sumj
\left\{\frac{r_i}{\pi}\rho^*(y_i)-\frac{1-r_i}{1-\pi}\right\}b^*(\x_i,g_{j,h},\Ep),\\
T_2&\equiv&\sqrt{n_1}n^{-1}\sumi\left\{\frac{r_i}{\pi}\rho^*(y_i)-\frac{1-r_i}{1-\pi}\right\}
\frac{\partial b^*(\x_i,a^*,\Ep)}{\partial{\Ep}}(\wh\E_p-\Ep).
\ese
Note that $E\{b^*(\x_i,G_{j,h},\Ep)\mid\x_i\}=b^*\{\x_i,E(G_{j,h}),\Ep\}=O(h^2)$ by \eqref{eq:g}.
Then using the property of the U-statistic and Condition \ref{con:6bandwidth},
we can express $T_1$ as
\bse
T_1
&=&n_1^{-1/2}\sumi
\left\{\frac{r_i}{\pi}\rho^*(y_i)-\frac{1-r_i}{1-\pi}\right\}
E\{b^*(\x_i,G_{j,h},\Ep)\mid\x_i,r_i,r_iy_i\}\\
&&+n_1^{-1/2}\sumj
E\left[\left\{\frac{R_i}{\pi}\rho^*(Y_i)-\frac{1-R_i}{1-\pi}\right\}
b^*(\X_i,g_{j,h},\Ep)\mid\x_j,r_j,r_jy_j\right]\\
&&-n^{1/2}n_1^{-1/2}
E\left[\left\{\frac{R_i}{\pi}\rho^*(Y_i)-\frac{1-R_i}{1-\pi}\right\}
b^*(\X_i,G_{j,h},\Ep)\right]+O_p(n_1^{-1/2})\\
&=&n_1^{-1/2}\sumj
E\left[\left\{\frac{R_i}{\pi}\rho^*(Y_i)-\frac{1-R_i}{1-\pi}\right\}
b^*(\X_i,g_{j,h},\Ep)\mid\x_j,r_j,r_jy_j\right]+O_p(n n_1^{-1/2}h^2+n_1^{-1/2})\\
&=&n_1^{-1/2}\sumj
\int\left[\rho^*(y)E\{b^*(\X,g_{j,h},\Ep)\mid y\}\py(y)
-E\{b^*(\X,g_{j,h},\Ep)\mid y\}\qy(y)\right]dy+o_p(1)\\
&=&n_1^{-1/2}\sumj\int\calL^*(g_{j,h})(y)
\left\{\rho^*(y)-\rho(y)\right\}dy+o_p(1)\\
&=&n_1^{-1/2}\sumj\int\{v_{j,h}(y)-\calL_{j,h}^*(a^*)(y)\}
\left\{\rho^*(y)-\rho(y)\right\}dy+o_p(1)\\
&=&n_1^{-1/2}\sumj r_j\int K_h(y-y_j)\{y-b^*(\x_j,a^*,\Ep)\}
\left\{\rho^*(y)-\rho(y)\right\}dy+o_p(1),
\ese
where
\bse
&&\int K_h(y-y_j)\{y-b^*(\x_j,a^*,\Ep)\}\left\{\rho^*(y)-\rho(y)\right\}dy\\
&=&\int K(t)\{y_j+ht-b^*(\x_j,a^*,\Ep)\}\{\rho^*(y_j+ht)-\rho(y_j+ht)\}dt\\
&=&\{y_j-b^*(\x_j,a^*,\Ep)\}\{\rho^*(y_j)-\rho(y_j)\}\\
&&+\left[\{y_j-b^*(\x_j,a^*,\Ep)\}\{{\rho^*}''(y_j)-\rho''(y_j)\}
+2\{{\rho^*}'(y_j)-\rho'(y_j)\}\right]\frac{h^2}{2}\int t^2K(t)dt+O(h^4)\\
&=&\{y_j-b^*(\x_j,a^*,\Ep)\}\{\rho^*(y_j)-\rho(y_j)\}+o(n_1^{-1/2})
\ese
under Conditions \ref{con:2rho}, \ref{con:5kernel}, and \ref{con:6bandwidth}.
On the other hand, using $E\{b^*(\X,a^*,\Ep)\mid y\}=y$,
$T_2$ can be written as
\bse
T_2
&=&\sqrt{n_1}n^{-1}\sumi\left\{\frac{r_i}{\pi}\rho^*(y_i)-\frac{1-r_i}{1-\pi}\right\}
\frac{\partial b^*(\x_i,a^*,\Ep)}{\partial{\Ep}}(\wh\E_p-\Ep)\\
&&-\sqrt{n_1}\frac{\partial}{\partial\Ep}E\left[\left\{\frac{R}{\pi}\rho^*(Y)-\frac{1-R}{1-\pi}\right\}
b^*(\X,a^*,\Ep)\right](\wh\E_p-\Ep)\\
&=&n_1^{1/4}\left(n^{-1}\sumi\left\{\frac{r_i}{\pi}\rho^*(y_i)-\frac{1-r_i}{1-\pi}\right\}
n_1^{1/4}\frac{\partial b^*(\x_i,a^*,\Ep)}{\partial{\Ep}}(\wh\E_p-\Ep)\right.\\
&&\left.-E\left[\left\{\frac{R}{\pi}\rho^*(Y)-\frac{1-R}{1-\pi}\right\}
n_1^{1/4}\frac{\partial b^*(\X,a^*,\Ep)}{\partial\Ep}(\wh\E_p-\Ep)\right]\right)\\
&=&n_1^{1/4}O_p(n^{-1/2})\\
&=&o_p(1).
\ese
Therefore, \eqref{eq:theta2} leads to
\bse
\sqrt{n_1}(\wt\theta-\theta)
&=&n^{-1/2}\sumi\left[\sqrt{\pi}\phi\eff^*(\x_i,r_i,r_iy_i)
+\frac{r_i}{\sqrt{\pi}}\{b^*(\x_i,a^*,\Ep)-y_i\}\{\rho^*(y_i)-\rho(y_i)\}\right]+o_p(1).
\ese
\qed

\subsection{Proposed Doubly Flexible Estimation for $\bt$ such that $\Eq\{\U(\X,Y,\bt)\}=\0$}\label{sec:generalU}

\subsubsection{General Approach to Estimating $\bt$}

Let $\dim(\U)=\dim(\bt)$,
$\Eq\{\partial\U(\X,Y,\bt)/\partial\bt\trans\}$ be invertible,
and $\A\equiv[\Eq\{\partial\U(\X,Y,\bt)/\partial\bt\trans\}]^{-1}$.
We first establish $\calF$, the family of all influence functions for
estimating $\bt$. In Section~\ref{sec:influence2}, we show that
\bse
\calF
\equiv\left[\frac{r}{\pi}[\rho(y)\A\{\U(\x,y,\bt)-\b(\x)\}+\c]
+\frac{1-r}{1-\pi}\{\A\b(\x)-\c\}:
\E\{\b(\X)\mid y\}=\E\{\U(\X,y,\bt)\mid y\},\forall \c\right].
\ese

We now derive the efficient influence function $\bphi\eff(\x,r,ry)$
which corresponds to the semiparametric efficiency bound and also
provides guidance on constructing flexible estimators for $\bt$. The
derivation is provided in Section \ref{sec:eff2proof}.

\begin{pro}\label{pro:eff2}
The efficient influence function $\bphi\eff(\x,r,ry)$ for $\bt$ is
\bse
\bphi\eff(\x,r,ry)
&=&\frac{r}{\pi}\rho(y)\A\left[\U(\x,y,\bt)
-\frac{\Ep\{\U(\x,Y,\bt)\rho^2(Y)+\a(Y)\rho(Y)\mid\x\}}
{\Ep\{\rho^2(Y)+\pi/(1-\pi)\rho(Y)\mid\x\}}\right]\\
&&+\frac{1-r}{1-\pi}\frac{\A\Ep\{\U(\x,Y,\bt)\rho^2(Y)+\a(Y)\rho(Y)\mid\x\}}
{\Ep\{\rho^2(Y)+\pi/(1-\pi)\rho(Y)\mid\x\}},
\ese
where $\a(y)$ satisfies
\bse
\E\left[\frac{\Ep\{\U(\X,Y,\bt)\rho^2(Y)+\a(Y)\rho(Y)\mid\X\}}
{\Ep\{\rho^2(Y)+\pi/(1-\pi)\rho(Y)\mid\X\}}\mid y\right]
=\E\{\U(\X,y,\bt)\mid y\}.
\ese
\end{pro}

The efficient estimator 
can be obtained
by solving the estimating equation $\sumi\bphi\eff(\x_i,r_i,r_iy_i)=\0$,
where the efficient influence function $\bphi\eff(\x,r,ry)$ is
given in Proposition \ref{pro:eff2}.
However, constructing $\bphi\eff(\x,r,ry)$ requires
two possibly unknown functions $\pyx(y,\x)$ and $\rho(y)$.
We will consider replacing the two functions
by working models $\pyx^\star(y,\x)$ and $\rho^*(y)$.

\subsubsection{Proposed Estimator $\wh\bt$: Doubly Flexible in $\rho^*(y)$ and $\pyx^\star(y,\x)$}
Under the adopted working models $\pyx^\star(y,\x)$ and $\rho^*(y)$,
we obtain an estimator $\wh\bt_t$
by solving the estimating equation $\sumi\bphi\eff^{*\star}(\x_i,r_i,r_iy_i)=\0$.
Given that $\A$ is invertible, this is equivalent to solving
\be\label{eq:thetahat2}
\sumi\left[\frac{r_i}{\pi}\rho^*(y_i)\{\U(\x_i,y_i,\bt)-\b^{*\star}(\x_i,\bt)\}
+\frac{1-r_i}{1-\pi}\b^{*\star}(\x_i,\bt)\right]=\0,
\ee
where
\bse
\b^{*\star}(\x,\bt)\equiv
\frac{\Ep^\star\{\U(\x,Y,\bt)\rho^{*2}(Y)+\a^{*\star}(Y,\bt)\rho^*(Y)\mid\x\}}
{\Ep^\star\{\rho^{*2}(Y)+\pi/(1-\pi)\rho^*(Y)\mid\x\}},
\ese
and $\a^{*\star}(y,\bt)$ is a solution for
\be\label{eq:a*2}
\E\left[\frac{\Ep^\star\{\U(\X,Y,\bt)\rho^{*2}(Y)+\a^{*\star}(Y,\bt)\rho^*(Y)\mid\X\}}
{\Ep^\star\{\rho^{*2}(Y)+\pi/(1-\pi)\rho^*(Y)\mid\X\}}\mid y\right]
=\E\{\U(\X,y,\bt)\mid y\}.
\ee
Our analysis shows that $\wh\bt_t$ is still consistent for $\bt$ under
suitable conditions.
We state this result as a proposition below, and provide its proof in
Section \ref{sec:consistency2proof}.

\begin{pro}\label{pro:consistency2}
Assume
$\E\{\partial\bphi\eff^{*\star}(\X,R,RY,\bt)/\partial\bt\trans\}$ is
invertible,
$\bt\in\bT$,  where $\bT$ is compact. Assume
$\E\{\sup_{\bt\in\boldsymbol{\Theta}}\|\bphi\eff^{*\star}(\X,R,RY,\bt)\|_2\}<\infty$.
Then $\wh\bt_t$ is consistent for $\bt$.
\end{pro}

Practically, we may choose to adopt a parametric working model
$\pyx^\star(y,\x,\bzeta)$
instead of a fixed function $\pyx^\star(y,\x)$.
In such construction,
a natural way will be to find an estimator for $\bzeta$, say $\wh\bzeta$,
using samples drawn from the population $\cal P$,
and use the estimated model $\pyx^\star(y,\x,\wh\bzeta)$
to form the estimating equation \eqref{eq:thetahat2}.
We will analyze this strategy in the sequel.

In addition, constructing \eqref{eq:thetahat2} requires $g(\x,y)$
since $\a^{*\star}(y,\bt)$ is obtained by solving \eqref{eq:a*2}.
To overcome this issue,
we incorporate nonparametric estimation. 
Note that directly estimating $g(\x,y)$
is subject to the curse of dimensionality since $\x$ is possibly high-dimensional.
Hence, instead, we estimate the both sides of \eqref{eq:a*2}
using nonparametric regression.
In summary, we approximate \eqref{eq:a*2}
using a kernel-based estimator of $\E(\cdot\mid y)$ as
\be\label{eq:a*hat2}
&&\int\a^{*\star}(t,\bt)\rho^*(t)
\sumi\frac{\pyx^\star(t,\x_i)r_iK_h(y-y_i)}
{\Ep^\star\{\rho^{*2}(Y)+\pi/(1-\pi)\rho^*(Y)\mid\x_i\}}dt\n\\
&=&\sumi\left[\U(\x_i,y,\bt)-\frac{\Ep^\star\{\U(\x_i,Y,\bt)\rho^{*2}(Y)\mid\x_i\}}
{\Ep^\star\{\rho^{*2}(Y)+\pi/(1-\pi)\rho^*(Y)\mid\x_i\}}\right]r_iK_h(y-y_i),
\ee
and solve this integral equation with respect to $\a^{*\star}(\cdot,\bt)$
to obtain $\wh\a^{*\star}(\cdot,\bt)$.

We summarize the estimation procedure in Algorithm \ref{algdouble2}.

\begin{algorithm}[H]
    \caption{Proposed Estimator $\wh\bt$: Doubly Flexible in $\rho^*(y)$ and $\pyx^\star(y,\x)$}\label{algdouble2}
 	\begin{algorithmic}
        \STATE \textbf{Input}: data from population $\mp$: $(y_i, \x_i, r_i=1)$, $i=1,\ldots,n_1$, data from population $\mq$: $(\x_j, r_j=0)$, $j=n_1+1,\ldots,n$, and value $\pi=n_1/n$.
            \STATE \textbf{do}
            \STATE (a) adopt a working model for $\rho(y)$, denoted as $\rho^*(y)$;
            \STATE (b) adopt a working model for $\pyx(y,\x)$, denoted as $\pyx^\star(y,\x)$ or $\pyx^\star(y,\x,\wh\bzeta)$;
            \STATE (c) compute
    $ w_i=[\E_p^\star\{\rho^{*2}(Y)+\pi/(1-\pi)\rho^*(Y)\mid\x_i\}]^{-1}$ for $i=1,\dots,n$;
            \STATE (d) obtain $\wh\a^{*\star}(\cdot,\bt)$ by solving the integral equation (\ref{eq:a*hat2});
            \STATE (e) compute
    $\wh\b^{*\star}(\x_i,\bt)
    =w_i\Ep^\star\{\U(\x_i,Y,\bt)\rho^{*2}(Y)+\wh\a^{*\star}(Y,\bt)\rho^*(Y)\mid\x_i\}$ for $i=1,\dots,n$;
            \STATE (f) obtain $\wh\bt$ by solving the estimating equation
    \bse
    \sumi\left[\frac{r_i}{\pi}\rho^*(y_i)\{\U(\x_i,y_i,\bt)-\wh\b^{*\star}(\x_i,\bt)\}
    +\frac{1-r_i}{1-\pi}\wh\b^{*\star}(\x_i,\bt)\right]=\0.
    \ese
 		\STATE \textbf{Output}: $\wh\bt$.
    \end{algorithmic}
\end{algorithm}

We now study the theoretical properties of $\wh\bt$. The main
technical challenge arises from the gap between the solutions for
\eqref{eq:a*2} and \eqref{eq:a*hat2}. For quantifying the gap,
we introduce some notations. Let
\bse
\Ep^\star\{\a(\x,Y,\bt)\mid\x,\wh\bzeta\}
&\equiv&\int \a(\x,y,\bt)\pyx^\star(y,\x,\wh\bzeta)dy,\\
u^{*\star}(t,y)&\equiv&\py(y)\int\frac{\rho^*(t)\pyx^\star(t,\x,\bzeta)}
{\Ep^\star\{\rho^{*2}(Y)\mid\x\}+\pi/(1-\pi)\Ep^\star\{\rho^*(Y)\mid\x\}}g(\x,y)d\x,\\
\calL^{*\star}(\a)(y,\bt)&\equiv&\py(y)\E\left[\frac{\Ep^\star\{\a(Y,\bt)\rho^*(Y)\mid\X\}}
{\Ep^\star\{\rho^{*2}(Y)\mid\X\}+\pi/(1-\pi)\Ep^\star\{\rho^*(Y)\mid\X\}}\mid y\right]
=\int \a(t,\bt)u^{*\star}(t,y)dt,\\
\calL_{i,h}^{*\star}(\a)(y,\bt)&\equiv&r_i K_h(y-y_i)\frac{\Ep^\star\{\a(Y,\bt)\rho^*(Y)\mid\x_i\}}
{\Ep^\star\{\rho^{*2}(Y)\mid\x_i\}+\pi/(1-\pi)\Ep^\star\{\rho^*(Y)\mid\x_i\}},\\
\wh\calL^{*\star}(\a)(y,\bt)&\equiv&
n_1^{-1}\sumi r_i K_h(y-y_i)\frac{\Ep^\star\{\a(Y,\bt)\rho^*(Y)\mid\x_i,\wh\bzeta\}}
{\Ep^\star\{\rho^{*2}(Y)\mid\x_i,\wh\bzeta\}+\pi/(1-\pi)\Ep^\star\{\rho^*(Y)\mid\x_i,\wh\bzeta\}},\\
\v^{*\star}(y,\bt)&\equiv&\py(y)\E\left[\U(\X,y,\bt)
-\frac{\Ep^\star\{\U(\X,Y,\bt)\rho^{*2}(Y)\mid\X\}}
{\Ep^\star\{\rho^{*2}(Y)+\pi/(1-\pi)\rho^*(Y)\mid\X\}}\mid y\right],\\
\v_{i,h}^{*\star}(y,\bt)&\equiv&r_iK_h(y-y_i)\left[\U(\x_i,y,\bt)
-\frac{\Ep^\star\{\U(\x_i,Y,\bt)\rho^{*2}(Y)\mid\x_i\}}
{\Ep^\star\{\rho^{*2}(Y)+\pi/(1-\pi)\rho^*(Y)\mid\x_i\}}\right],\\
\wh\v^{*\star}(y,\bt)&\equiv&n_1^{-1}\sumi r_iK_h(y-y_i)\left[\U(\x_i,y,\bt)
-\frac{\Ep^\star\{\U(\x_i,Y,\bt)\rho^{*2}(Y)\mid\x_i,\wh\bzeta\}}
{\Ep^\star\{\rho^{*2}(Y)+\pi/(1-\pi)\rho^*(Y)\mid\x_i,\wh\bzeta\}}\right].
\ese
By the definitions of the solutions for \eqref{eq:a*2} and \eqref{eq:a*hat2},
$\a^{*\star}(y,\bt)$ and $\wh\a^{*\star}(y,\bt)$ satisfy
$\calL^{*\star}(\a^{*\star})(y,\bt)=\v^{*\star}(y,\bt)$ and
$\wh\calL^{*\star}(\wh\a^{*\star})(y,\bt)=\wh\v^{*\star}(y,\bt)$.

We list a set of regularity conditions.
\begin{enumerate}[label=(B\arabic*),ref=(B\arabic*),start=1]
    \item\label{con:0existence}
    $\A^{*\star}\equiv[\Eq^{*\star}\{\partial\U(\X,Y,\bt)/\partial\bt\trans\}]^{-1}$ and $\E\{\partial\bphi\eff^{*\star}(\X,R,RY,\bt)/\partial\bt\trans\}$
    are invertible,
    $\bt\in\boldsymbol{\Theta}$ where $\boldsymbol{\Theta}$ is compact,
    and $\E\{\sup_{\bt\in\boldsymbol{\Theta}}\|\bphi\eff^{*\star}(\X,R,RY,\bt)\|_2\}<\infty$.
    $\U(\x,y,\bt)$ is twice differentiable with respect to $y$
    and its derivative is bounded.
    \item\label{con:1complete2}
    $\pyx^\star(y,\x)$ is complete.
    \item\label{con:2rho2}
    $\rho^*(y)>\delta$ for all $y$ on the support of $\py(y)$,
    where $\delta$ is a constant such that $\delta>0$.
    $\rho^*(y)$ is twice differentiable and its derivative is bounded.
    \item\label{con:3bounded2}
    The function $u^{*\star}(t,y)$ is bounded
    and has bounded derivatives with respect to $t$ and $y$ on its support.
    $\a^{*\star}(y,\bt)$ in \eqref{eq:a*2} is bounded.
    \item\label{con:4compact2}
    The support sets of $g(\x,y),\py(y),\rho^*(y)$ are compact.
    \item\label{con:5kernel2}
    The kernel function $K(\cdot)\ge0$ is
    symmetric, bounded, and twice differentiable with bounded first derivative.
    It has support on $(-1,1)$ and satisfies $\int_{-1}^1K(t)dt=1$.
    \item\label{con:6bandwidth2}
    The bandwidth $h$ satisfies $n_1(\log n_1)^{-4}h^2\to\infty$ and
    $n^2n_1^{-1}h^4\to 0$.
\end{enumerate}

\begin{lemma}\label{lem:l2}
For $\a(y,\bt)=\{a_1(y,\bt),\dots,a_d(y,\bt)\}\trans$,
let $\|\a\|_\infty\equiv\max_{k=1\dots d}\|a_k\|_\infty$.
Under the regularity conditions \ref{con:0existence}-\ref{con:4compact2},
the linear operator $\calL^{*\star}: L^\infty(R^d)\to L^\infty(R^d)$ is invertible.
In addition,
there exist constants $c_1, c_2$ such that $0<c_1,c_2<\infty$ and
for all $\a(y)\in L^\infty(R^d)$,
\begin{enumerate}[label=(\roman*)]
    \item $c_1\|\a\|_\infty\leq\|\calL^{*\star}(\a)\|_\infty\leq c_2\|\a\|_\infty$,
    \item $\|{\cal L}^{*\star-1}(\a)\|_\infty\leq c_1^{-1}\|\a\|_\infty$.
\end{enumerate}
\end{lemma}

The proof of Lemma \ref{lem:l2} is in Section \ref{sec:l2proof}.
Below, we present the asymptotic normality of $\wh\bt$, and provide
its proof in Section \ref{sec:thetanewproof}.

\begin{Th}\label{th:thetanew}
Assume $\wh\bzeta$ satisfies $\|\wh\bzeta-\bzeta\|_2=O_p(n_1^{-1/2})$
and $\Ep^\star\{\|\S_\bzeta^\star(Y,\x,\bzeta)\|_2\mid\x\}$ is bounded,
where $\S_\bzeta^\star(y,\x,\bzeta)\equiv\partial\log\pyx^\star(y,\x,\bzeta)/\partial\bzeta$.
For any choice of $\pyx^\star(y,\x,\bzeta)$ and $\rho^*(y)$,
under Conditions \ref{con:0existence}-\ref{con:6bandwidth2},
\bse
\sqrt{n_1}(\wh\bt-\bt)\to N\left\{\0,\A\A^{*\star-1}\bSigma(\A\A^{*\star-1})\trans\right\}
\ese
in distribution as $n_1\to\infty$, where
\be
\bSigma&\equiv&\var\left[\sqrt{\pi}\bphi\eff^{*\star}(\X,R,RY,\bt)
+\frac{R}{\sqrt{\pi}}\A^{*\star}
\left\{\b^{*\star}(\X,\bt)-\U(\X,Y,\bt)\right\}\{\rho^*(Y)-\rho(Y)\}\right].\label{eq:bSigma}
\ee
\end{Th}

\begin{Rem}\label{rem:thetanewpar}
In Theorem \ref{th:thetanew},
$\|\wh\bzeta-\bzeta\|_2=O_p(n_1^{-1/2})$ is the only assumption
imposed on $\wh\bzeta$.
That is, the result in Theorem \ref{th:thetanew} holds
as long as $\wh\bzeta$ is $\sqrt{n_1}$-consistent for $\bzeta$,
regardless of the asymptotic variance of $\sqrt{n_1}(\wh\bzeta-\bzeta)$.
This is easily achievable
by constructing a standard MLE or moment based estimator for $\bzeta$
in $\pyx^\star(y,\x,\bzeta)$,
based on the $n_1$ observations from population $\cal P$.
\end{Rem}

In addition to the above remark,
we can see from Theorem \ref{th:thetanew} that
$\wh\bt$ is indeed the efficient estimator for $\bt$
when the working models $\pyx^\star(y,\x,\bzeta)$ and $\rho^*(y)$ are
correctly specified,
and $\wh\bzeta$ satisfies $\|\wh\bzeta-\bzeta\|_2=O_p(n_1^{-1/2})$.
We formally state this result as Corollary \ref{th:thetaneweff}.

\begin{Cor}\label{th:thetaneweff}
Assume $\wh\bzeta$ satisfies $\|\wh\bzeta-\bzeta\|_2=O_p(n_1^{-1/2})$
and $\Ep^\star\{\|\S_\bzeta^\star(Y,\x,\bzeta)\|_2\mid\x\}$ is bounded.
If $\pyx^\star(y,\x,\bzeta)=\pyx(y,\x)$ and $\rho^*(y)=\rho(y)$,
under Conditions \ref{con:0existence}-\ref{con:6bandwidth2},
\bse
\sqrt{n_1}(\wh\bt\eff-\bt)\to N\left[0,\var\left\{\sqrt{\pi}\bphi\eff(\X,R,RY,\bt)\right\}\right]
\ese
in distribution as $n_1\to\infty$.
\end{Cor}

\subsubsection{Alternative Estimator $\wt\bt$: Singly Flexible in $\rho^*(y)$}

We now consider replacing $\Ep(\cdot\mid\x)$ in the efficient
influence function in Proposition~\ref{pro:eff2} by an arbitrary
estimator $\wh\E_p(\cdot\mid\x)$
with convergence rate faster than $n_1^{-1/4}$.
We first describe the estimation procedure in Algorithm \ref{algsingle2}.

\begin{algorithm}[H]
    \caption{Alternative Estimator $\wt\bt$: Singly Flexible in $\rho^*(y)$}\label{algsingle2}
 	\begin{algorithmic}
        \STATE \textbf{Input}: data from population $\mp$: $(y_i, \x_i, r_i=1)$, $i=1,\ldots,n_1$, data from population $\mq$: $(\x_j, r_j=0)$, $j=n_1+1,\ldots,n$, and value $\pi=n_1/n$.
            \STATE \textbf{do}
            \STATE (a) adopt a working model for $\rho(y)$, denoted as $\rho^*(y)$;
            \STATE (b) adopt a nonparametric or machine learning algorithm for estimating $\Ep(\cdot\mid\x)$, denoted as $\wh\E_p(\cdot\mid\x)$;
            \STATE (c) compute
    $ \wh w_i=[\wh\E_p\{\rho^{*2}(Y)+\pi/(1-\pi)\rho^*(Y)\mid\x_i\}]^{-1}$ for $i=1,\dots,n$;
            \STATE (d) obtain $\wh\a^*(\cdot,\bt)$ by solving the integral equation (\ref{eq:a*hat2}) with $\Ep$ replaced by $\wh\E_p$;
            \STATE (e) compute
    $\wh\b^*(\x_i,\bt)
    =\wh w_i\wh\E_p\{\U(\x_i,Y,\bt)\rho^{*2}(Y)+\wh\a^*(Y,\bt)\rho^*(Y)\mid\x_i\}$ for $i=1,\dots,n$;
            \STATE (f) obtain $\wh\bt$ by solving the estimating equation
    \bse
    \sumi\left[\frac{r_i}{\pi}\rho^*(y_i)\{\U(\x_i,y_i,\bt)-\wh\b^*(\x_i,\bt)\}
    +\frac{1-r_i}{1-\pi}\wh\b^*(\x_i,\bt)\right]=\0.
    \ese
 		\STATE \textbf{Output}: $\wt\bt$.
    \end{algorithmic}
\end{algorithm}

For establishing the theoretical properties of $\wt\bt$, we define
\bse
u^*(t,y)&\equiv&\py(y)\int\frac{\rho^*(t)\pyx(t,\x)}
{\Ep\{\rho^{*2}(Y)\mid\x\}+\pi/(1-\pi)\Ep\{\rho^*(Y)\mid\x\}}g(\x,y)d\x,\\
\calL^*(\a)(y,\bt)&\equiv&\py(y)\E\left[\frac{\Ep\{\a(Y,\bt)\rho^*(Y)\mid\X\}}
{\Ep\{\rho^{*2}(Y)\mid\X\}+\pi/(1-\pi)\Ep\{\rho^*(Y)\mid\X\}}\mid y\right]
=\int \a(t,\bt)u^*(t,y)dt,\\
\calL_{i,h}^*(\a)(y,\bt)&\equiv&r_i K_h(y-y_i)\frac{\Ep\{\a(Y,\bt)\rho^*(Y)\mid\x_i\}}
{\Ep\{\rho^{*2}(Y)\mid\x_i\}+\pi/(1-\pi)\Ep\{\rho^*(Y)\mid\x_i\}},\\
\wh\calL^*(\a)(y,\bt)&\equiv&
n_1^{-1}\sumi r_i K_h(y-y_i)\frac{\wh\E_p\{\a(Y,\bt)\rho^*(Y)\mid\x_i\}}
{\wh\E_p\{\rho^{*2}(Y)\mid\x_i\}+\pi/(1-\pi)\wh\E_p\{\rho^*(Y)\mid\x_i\}},\\
\v^*(y,\bt)&\equiv&\py(y)\E\left[\U(\X,y,\bt)
-\frac{\Ep\{\U(\X,Y,\bt)\rho^{*2}(Y)\mid\X\}}
{\Ep\{\rho^{*2}(Y)+\pi/(1-\pi)\rho^*(Y)\mid\X\}}\mid y\right],\\
\v_{i,h}^*(y,\bt)&\equiv&r_iK_h(y-y_i)\left[\U(\x_i,y,\bt)
-\frac{\Ep\{\U(\x_i,Y,\bt)\rho^{*2}(Y)\mid\x_i\}}
{\Ep\{\rho^{*2}(Y)+\pi/(1-\pi)\rho^*(Y)\mid\x_i\}}\right],\\
\wh\v^*(y,\bt)&\equiv&n_1^{-1}\sumi r_iK_h(y-y_i)\left[\U(\x_i,y,\bt)
-\frac{\wh\E_p\{\U(\x_i,Y,\bt)\rho^{*2}(Y)\mid\x_i\}}
{\wh\E_p\{\rho^{*2}(Y)+\pi/(1-\pi)\rho^*(Y)\mid\x_i\}}\right].
\ese
Then the solutions for \eqref{eq:a*2} and \eqref{eq:a*hat2},
$\a^*(y,\bt)$ and $\wh\a^*(y,\bt)$ satisfy
$\calL^*(\a^*)(y,\bt)=\v^*(y,\bt)$ and $\wh\calL^*(\wh\a^*)(y,\bt)=\wh\v^*(y,\bt)$.

\begin{Th}\label{th:thetanew2}
Assume $\wh\E_p$ satisfies $\|\wh\E_p\{\a(Y)\mid\x\}-\Ep\{\a(Y)\mid\x\}\|_\infty=o_p(n_1^{-1/4})$
for any bounded function $\a(y)$.
For any choice of $\rho^*(y)$,
under Conditions \ref{con:0existence}-\ref{con:6bandwidth2},
\bse
\sqrt{n_1}(\wt\bt-\bt)\to N\left\{\0,\A\A^{*-1}\bSigma(\A\A^{*-1})\trans\right\}
\ese
in distribution as $n_1\to\infty$, where
\bse
\bSigma&\equiv&\var\left[\sqrt{\pi}\bphi\eff^*(\X,R,RY,\bt)
+\frac{R}{\sqrt{\pi}}\A^*
\{\b^*(\X,\bt)-\U(\X,Y,\bt)\}\{\rho^*(Y)-\rho(Y)\}\right].
\ese
\end{Th}

The proof of Theorem \ref{th:thetanew2} is in Section \ref{sec:thetanew2proof}.
Theorem \ref{th:thetanew2} further confirms that
when the shifting model $\rho^*(y)$ is correctly posited,
the estimator $\wt\bt$ achieves the efficiency bound.
We point out this result as Corollary \ref{th:thetanew2eff}.

\begin{Cor}\label{th:thetanew2eff}
Assume $\wh\E_p$ satisfies $\|\wh\E_p\{\a(Y)\mid\x\}-\Ep\{\a(Y)\mid\x\}\|_\infty=o_p(n_1^{-1/4})$
for any bounded function $\a(y)$.
If $\rho^*(y)=\rho(y)$,
under Conditions \ref{con:0existence}-\ref{con:6bandwidth2},
\bse
\sqrt{n_1}(\wt\bt\eff-\bt)\to N\left[0,\var\{\sqrt{\pi}\bphi\eff(\X,R,RY)\}\right]
\ese
in distribution as $n_1\to\infty$.
\end{Cor}

\subsection{Proofs of Section~\ref{sec:generalU}}\label{sec:proofgeneralU}

\subsubsection{Derivation of $\calF$}\label{sec:influence2}

We first state and prove a result regarding the tangent space.

\begin{pro}\label{pro:tangentspace2}
The tangent space of (\ref{eq:likelihood}) is
$\calT\equiv\calT_\ba\oplus(\calT_\bb+\calT_\bg)$, where
\bse
\calT_\ba&=&\left[r\a_1(y):\Ep\{\a_1(Y)\}=\0\right],\\
\calT_\bb&=&\left[r\a_2(\x,y)+(1-r)\Eq\{\a_2(\x,Y)\mid\x\}:\E\{\a_2(\X,y)\mid y\}=\0\right],\\
\calT_\bg&=&\left[(1-r)\Eq\{\a_3(Y)\mid\x\}:\Eq\{\a_3(Y)\}=\0\right].
\ese
\end{pro}

\begin{proof}
Consider a parametric submodel of \eqref{eq:likelihood},
\be\label{eq:submodel2}
f_{\X,R,RY}(\x,r,ry,\bd)
&=&\pi^r(1-\pi)^{1-r}\{g(\x,y,\bb)\py(y,\ba)\}^{r}
\left\{\int g(\x,y,\bb)\qy(y,\bg)dy\right\}^{1-r},
\ee
where $\bd=(\ba\trans,\bb\trans,\bg\trans)\trans$.
We can derive that the score function associated with an arbitrary $\bd$ is
$\S_\bd\equiv(\S_\ba\trans,\S_\bb\trans,\S_\bg\trans)\trans$,
where
\bse
\S_\ba(\x,r,ry)
&\equiv&r\a_\ba(y),\\
\S_\bb(\x,r,ry)
&\equiv&r\a_\bb(\x,y)+(1-r)\Eq\{\a_\bb(\x,Y)\mid\x\},\\
\S_\bg(\x,r,ry)
&\equiv&(1-r)\Eq\{\a_\bg(Y)\mid\x\},
\ese
$\Ep\{\a_\ba(Y)\}=\0$,
$\Ep\{\a_\bb(\X,y)\mid y\}=\Eq\{\a_\bb(\X,y)\mid y\}=\E\{\a_\bb(\X,y)\mid y\}=\0$,
and $\Eq\{\a_\bg(Y)\}=\0$.
The above derivation directly leads to Proposition \ref{pro:tangentspace2}.
\end{proof}

We now establish $\calF$ in Proposition~\ref{pro:influ2}.

\begin{pro}\label{pro:influ2}
The set of the influence functions for $\bt$
is
\bse
\calF
\equiv\left[\frac{r}{\pi}[\rho(y)\A\{\U(\x,y,\bt)-\b(\x)\}+\c]
+\frac{1-r}{1-\pi}\{\A\b(\x)-\c\}:
\E\{\b(\X)\mid y\}=\E\{\U(\X,y,\bt)\mid y\},\forall \c\right].
\ese
\end{pro}

\begin{proof}
Recall the submodel \eqref{eq:submodel2}, then it can be verified that
\bse
\frac{\partial\bt}{\partial\ba\trans}&=&\0,\\
\frac{\partial\bt}{\partial\bb\trans}&=&\A\Eq\{\U(\X,Y,\bt)\a_\bb\trans(\X,Y)\}
=\A\Ep\{\U(\X,Y,\bt)\rho(Y)\a_\bb\trans(\X,Y)\},\\
\frac{\partial\bt}{\partial\bg\trans}&=&\A\Eq\{\U(\X,Y,\bt)\a_\bg\trans(Y)\},
\ese
where $\E\{\a_\bb(\X,y)\mid y\}=\0$ and $\Eq\{\a_\bg(Y)\}=\0$.
Now let $\bphi(\x,r,ry)$ be
\bse
\bphi(\x,r,ry)\equiv\frac{r}{\pi}\bphi_1(\x,y)+\frac{1-r}{1-\pi}\bphi_2(\x).
\ese
For $\bphi(\x,r,ry)$ to be an influence function,
it must satisfy
\be\label{eq:hilbert2}
\E(\bphi)=\Ep\{\bphi_1(\X,Y)\}+\Eq\{\bphi_2(\X)\}=\0
\ee
and $\E(\bphi\S_\bd\trans)=\partial\bt/\partial\bd\trans$,
where $\bd=(\ba\trans,\bb\trans,\bg\trans)\trans$ and
$\S_\bd=(\S_\ba\trans,\S_\bb\trans,\S_\bg\trans)\trans$ is the score function
of the submodel \eqref{eq:submodel2}.
First, $\E(\bphi\S_\ba\trans)=\partial\bt/\partial\ba\trans$ is equivalent to
\be\label{eq:alpha2}
\E\{\bphi_1(\X,y)\mid y\}=\c
\ee
for some constant vector $\c$.
In addition, since
$\E(\bphi\S_\bb\trans)=\Ep[\{\bphi_1(\X,Y)+\bphi_2(\X)\rho(Y)\}\a_\bb\trans(\X,Y)]$
and $\a_\bb(\x,y)$ satisfies $\E\{\a_\bb(\X,y)\mid y\}=\0$,
$\E(\bphi\S_\bb\trans)=\partial\bt/\partial\bb\trans$ implies
$\bphi_1(\x,y)+\bphi_2(\x)\rho(y)=\A\U(\x,y,\bt)\rho(y)+\a(y)$ for some $\a(y)$.
Then \eqref{eq:alpha2} yields
\bse
\bphi_1(\x,y)=\rho(y)[\A\U(\x,y,\bt)-\bphi_2(\x)
-\E\{\A\U(\X,y,\bt)-\bphi_2(\X)\mid y\}]+\c.
\ese
Also, noting that
$\E(\bphi\S_\bg\trans)=\partial\bt/\partial\bg\trans$ is equivalent to
$\E\{\bphi_2(\X)\mid y\}=\E\{\A\U(\X,y,\bt)\mid y\}+\c^*$
for some constant vector $\c^*$,
we have from \eqref{eq:hilbert2} and \eqref{eq:alpha2} that $\c^*=-\c$.
Therefore, defining $\b(\x)\equiv\A^{-1}\{\bphi_2(\x)+\c\}$,
the summary description of the influence function is
\bse
\bphi(\x,r,ry)
=\frac{r}{\pi}[\rho(y)\A\{\U(\x,y,\bt)-\b(\x)\}+\c]
+\frac{1-r}{1-\pi}\{\A\b(\x)-\c\},
\ese
where $\b(\x)$ satisfies
$\E\{\b(\X)\mid y\}=\E\{\U(\X,y,\bt)\mid y\}$ and $\c$ is a constant vector.
\end{proof}

\subsubsection{Proof of Proposition \ref{pro:eff2}}\label{sec:eff2proof}
Note that
\bse
\frac{\Ep\{a(\x,Y)\rho(Y)\mid\x\}}{\Ep\{\rho(Y)\mid\x\}}
=\frac{\int a(\x,y)\rho(y)g(\x,y)\py(y)dy}{\int \rho(y)g(\x,y)\py(y)dy}
=\frac{\int a(\x,y)g(\x,y)\qy(y)dy}{\int g(\x,y)\qy(y)dy}
=\Eq\{a(\x,Y)\mid\x\},
\ese
then $\bphi\eff(\x,r,ry)$ can be alternatively written as
\bse
&&\bphi\eff(\x,r,ry)\\
&=&\frac{r}{\pi}\rho(y)\A\left[\U(\x,y,\bt)
-\frac{\Eq\{\U(\x,Y,\bt)\rho(Y)+\a(Y)\mid\x\}}
{\Eq\{\rho(Y)+\pi/(1-\pi)\mid\x\}}\right]
+\frac{1-r}{1-\pi}\frac{\A\Eq\{\U(\x,Y,\bt)\rho(Y)+\a(Y)\mid\x\}}
{\Eq\{\rho(Y)+\pi/(1-\pi)\mid\x\}},
\ese
where $\a(y)$ satisfies
\bse
\E\left[\frac{\Eq\{\U(\X,Y,\bt)\rho(Y)+\a(Y)\mid\X\}}
{\Eq\{\rho(Y)+\pi/(1-\pi)\mid\X\}}\mid y\right]
=\E\{\U(\X,y,\bt)\mid y\}.
\ese
First, it is immediate that $\bphi\eff(\x,r,ry)$ is an influence function for $\bt$,
i.e., belongs to $\calF$ given in Proposition \ref{pro:influ2}
from letting
\bse
\b(\x)&\equiv&\frac{\Eq\{\U(\x,Y,\bt)\rho(Y)+\a(Y)\mid\x\}}
{\Eq\{\rho(Y)+\pi/(1-\pi)\mid\x\}},\\
\c&\equiv&\0.
\ese
Next, we show that $\bphi\eff(\x,r,ry)$ is
in the tangent space $\calT$ given in Proposition \ref{pro:tangentspace2}.
We decompose $\bphi\eff(\x,r,ry)$ into
\bse
\bphi\eff(\x,r,ry)
=r\{\a_1(y)+\a_2(\x,y)\}+(1-r)[\Eq\{\a_2(\x,Y)\mid\x\}+\Eq\{\a_3(Y)\mid\x\}],
\ese
where
\bse
\a_1(y)&\equiv&\0,\\
\a_2(\x,y)&\equiv&\frac{1}{\pi}\rho(y)\A\{\U(\x,y,\bt)-\b(\x)\},\\
\a_3(y)&\equiv&\frac{1}{\pi}\A\a(y).
\ese
Then it is easy to see that $r\a_1(y)\in\calT_\ba$, and
$r\a_2(\x,y)+(1-r)\Eq\{\a_2(\x,Y)\mid\x\}\in\calT_\bb$
because $\E\{\b(\X)\mid y\}=\E\{\U(\X,y,\bt)\mid y\}$.
Further, $(1-r)\Eq\{\a_3(Y)\mid\x\}\in\calT_\bg$ since
\bse
\Eq\{\a_3(Y)\}
&=&\frac{1}{\pi}\A\Eq[\Eq\{\a(Y)\mid\X\}]\\
&=&\frac{1}{\pi}\A\Eq\{\b(\X)\Eq\{\rho(Y)\mid\X\}]
+\frac{1}{1-\pi}\A\Eq\{\b(\X)\}
-\frac{1}{\pi}\A\Eq[\Eq\{\U(\X,Y,\bt)\rho(Y)\mid\X\}]\\
&=&\frac{1}{\pi}\A\Eq[\Eq\{\b(\X)-\U(\X,Y,\bt)\mid Y\}\rho(Y)]\\
&=&\0,
\ese
where the third equality holds
since $\Eq\{\b(\X)\}=\Eq\{\U(\X,Y,\bt)\}=\0$ by the definition of $\bt$.
Hence $\bphi\eff(\x,r,ry)$ belongs to $\calT$.
\qed
\subsubsection{Proof of Proposition \ref{pro:consistency2}}\label{sec:consistency2proof}
The invertibility of
$\E\{\partial\bphi\eff^{*\star}(\X,R,RY,\bt)/\partial\bt\trans\}$
implies that
$\E\{\bphi\eff^{*\star}(\X,R,RY,\bt)\}=\0$
has a unique solution in the neighborhood of $\bt$,
and the existence of $\partial\bphi\eff^{*\star}(\x,r,ry,\bt)/\partial\bt\trans$
automatically implies that $\bphi\eff^{*\star}(\x,r,ry,\bt)$ is
continuous with respect to $\bt$.
Then following Theorem 2.6 of \cite{newey1994large},
it suffices to show that
$\E\{\bphi\eff^{*\star}(\X,R,RY,\bt)\}=\0$.
It is immediate that $\E\{\b^{*\star}(\X,\bt)\mid y\}=\E\{\U(\X,y,\bt)\mid y\}$
from the definition of $\b^{*\star}(\x,\bt)$.
Hence,
\bse
\E\{\bphi\eff^{*\star}(\X,R,RY,\bt)\}
&=&\A^{*\star}\E\left[\frac{R}{\pi}\rho^*(Y)\{\U(\X,Y,\bt)-\b^{*\star}(\X,\bt)\}
+\frac{1-R}{1-\pi}\b^{*\star}(\X,\bt)\right]\\
&=&\A^{*\star}\Ep\left[\rho^*(Y)\E\{\U(\X,Y,\bt)-\b^{*\star}(\X,\bt)\mid Y\}\right]
+\A^{*\star}\Eq\left[\E\left\{\b^{*\star}(\X,\bt)\mid Y\right\}\right]\\
&=&\A^{*\star}\Eq\left[\E\left\{\U(\X,Y,\bt)\mid Y\right\}\right]\\
&=&\0
\ese
by the definition of $\bt$.
\qed

\subsubsection{Proof of Lemma \ref{lem:l2}}\label{sec:l2proof}
From the definition of $\calL^{*\star}(\a)$
and Conditions \ref{con:3bounded2}-\ref{con:4compact2},
it is immediate that there exists a constant $0<c_2<\infty$
such that $\|\calL^{*\star}(\a)\|_\infty\leq c_2\|\a\|_\infty$.
Now we show there is a constant $0<c_1<\infty$
such that $c_1\|\a\|_\infty\leq\|\calL^{*\star}(\a)\|_\infty$.
Note that if $\calL^{*\star}$ is invertible,
by the bounded inverse theorem
we have $\|\calL^{*\star-1}(\v)\|_\infty\leq c_1^{-1}\|\v\|_\infty$
for some constant $0<c_1<\infty$, i.e.,
$c_1\|\a\|_\infty\leq\|\calL^{*\star}(\a)\|_\infty$.
Hence it suffices to show that $\calL^{*\star}$ is invertible.
We prove this by contradiction.
Suppose there are $\a_1(y,\bt)$ and $\a_2(y,\bt)$ such that
$\calL^{*\star}(\a_1)(y,\bt)=\calL^{*\star}(\a_2)(y,\bt)=\v^{*\star}(y,\bt)$ and $\a_1\neq \a_2$.
Then by Conditions \ref{con:0existence}-\ref{con:2rho2}, we have
\bse
\A^{*\star}\Ep^\star\{\a_1(Y,\bt)\rho^*(Y)\mid\x\}
\neq\A^{*\star}\Ep^\star\{\a_2(Y,\bt)\rho^*(Y)\mid\x\}.
\ese
Now, the efficient score calculated under the posited models is
\bse
\bphi\eff^{*\star}(\x,r,ry,\bt)
&=&\frac{r}{\pi}\rho^*(y)\A^{*\star}\left[\U(\x,y,\bt)
-\frac{\Ep^\star\{\U(\x,Y,\bt)\rho^{*2}(Y)+\a(Y,\bt)\rho^*(Y)\mid\x\}}
{\Ep^\star\{\rho^{*2}(Y)+\pi/(1-\pi)\rho^*(Y)\mid\x\}}\right]\\
&&+\frac{1-r}{1-\pi}
\frac{\A^{*\star}\Ep^\star\{\U(\x,Y,\bt)\rho^{*2}(Y)+\a(Y,\bt)\rho^*(Y)\mid\x\}}
{\Ep^\star\{\rho^{*2}(Y)+\pi/(1-\pi)\rho^*(Y)\mid\x\}},
\ese
where $\a(y,\bt)$ satisfies $\calL^{*\star}(\a)(y,\bt)=\v^{*\star}(y,\bt)$.
Then letting $\a=\a_1$ and $\a=\a_2$ gives two distinct efficient scores,
which contradicts the uniqueness of the efficient score.
Therefore, there is a unique solution $\a^{*\star}(y,\bt)$
for $\calL^{*\star}(\a^{*\star})(y,\bt)=\v^{*\star}(y,\bt)$,
hence $\calL^{*\star}$ is invertible.
\qed

\subsubsection{Proof of Theorem \ref{th:thetanew}}\label{sec:thetanewproof}
We define
\bse
\b^{*\star}(\x,\a,\bzeta,\bt)
\equiv\frac{\Ep^\star\{\U(\x,Y,\bt)\rho^{*2}(Y)+\a(Y,\bt)\rho^*(Y)\mid\x,\bzeta\}}
{\Ep^\star\{\rho^{*2}(Y)+\pi/(1-\pi)\rho^*(Y)\mid\x,\bzeta\}},
\ese
then under Conditions \ref{con:0existence}, \ref{con:3bounded2}-\ref{con:6bandwidth2}
and $\|\wh\bzeta-\bzeta\|_2=O_p(n_1^{-1/2})$,
\bse
&&\wh\calL^{*\star}(\a^{*\star})(y,\bt)-\wh\v^{*\star}(y,\bt)\\
&=&n_1^{-1}\sumi r_iK_h(y-y_i)\{\b^{*\star}(\x_i,\a^{*\star},\wh\bzeta,\bt)-\U(\x_i,y,\bt)\}\n\\
&=&n_1^{-1}\sumi r_iK_h(y-y_i)\left\{\b^{*\star}(\x_i,\a^{*\star},\bzeta,\bt)
+\frac{\partial\b^{*\star}(\x_i,\a^{*\star},\bzeta,\bt)}{\partial\bzeta\trans}
(\wh\bzeta-\bzeta)+o_p(n_1^{-1/2})-\U(\x_i,y,\bt)\right\}\n\\
&=&n_1^{-1}\sumi\left\{\calL^{*\star}_{i,h}(\a^{*\star})(y,\bt)-\v^{*\star}_{i,h}(y,\bt)\right\}
+\left[\py(y)\frac{\partial\E\{\b^{*\star}(\X,\a^{*\star},\bzeta,\bt)\mid y\}}
{\partial\bzeta\trans}+o_p(1)\right](\wh\bzeta-\bzeta)+o_p(n_1^{-1/2})\n\\
&=&n_1^{-1}\sumi\left\{\calL^{*\star}_{i,h}(\a^{*\star})(y,\bt)-\v^{*\star}_{i,h}(y,\bt)\right\}
+o_p(n_1^{-1/2})
\ese
uniformly in $(y,\bt)$
since $\E\{\b^{*\star}(\X,\a^{*\star},\bzeta,\bt)\mid y\}=\E\{\U(\X,y,\bt)\mid y\}$.
In addition, for any bounded function $\a(y,\bt)$,
$\|(\wh\calL^{*\star}-\calL^{*\star})(\a)\|_\infty
=O_p\left\{(n_1h)^{-1/2}\log n_1+h^2\right\}=o_p(n_1^{-1/4})$
under Conditions \ref{con:0existence}, \ref{con:3bounded2}-\ref{con:6bandwidth2}
and the assumption $\|\wh\bzeta-\bzeta\|_2=O_p(n_1^{-1/2})$.
Similarly, $\|\wh\v^{*\star}-\v^{*\star}\|_\infty=o_p(n_1^{-1/4})$.
Then using that $\calL^{*\star-1}$ is a bounded linear operator by Lemma \ref{lem:l2},
$\wh\a^{*\star}(y,\bt)$ can be expressed as
\be\label{eq:a2}
\wh\a^{*\star}(y,\bt)
&=&\{\calL^{*\star}+(\wh\calL^{*\star}-\calL^{*\star})\}^{-1}(\wh\v^{*\star})(y,\bt)\n\\
&=&\{\calL^{*\star-1}-\calL^{*\star-1}(\wh\calL^{*\star}-\calL^{*\star})\calL^{*\star-1}\}
\{\v^{*\star}+(\wh\v^{*\star}-\v^{*\star})\}(y,\bt)+o_p(n_1^{-1/2})\n\\
&=&\a^{*\star}(y,\bt)+\calL^{*\star-1}(\wh\v^{*\star}-\v^{*\star})(y,\bt)
-\{\calL^{*\star-1}(\wh\calL^{*\star}-\calL^{*\star})\}(\a^{*\star})(y,\bt)+o_p(n_1^{-1/2})\n\\
&=&\a^{*\star}(y,\bt)+\calL^{*\star-1}\{\wh\v^{*\star}-\wh\calL^{*\star}(\a^{*\star})\}(y,\bt)
+o_p(n_1^{-1/2})\n\\
&=&\a^{*\star}(y,\bt)+n_1^{-1}\sumi\g_{i,h}(y,\bt)+o_p(n_1^{-1/2})
\ee
uniformly in $y$,
where $\g_{i,h}(y,\bt)\equiv
\calL^{*\star-1}\{\v^{*\star}_{i,h}-\calL_{i,h}^{*\star}(\a^{*\star})\}(y,\bt)$.
Note that
$\|E(\V_{i,h}^{*\star})-\v^{*\star}\|_\infty=O(h^2)$
by Conditions \ref{con:4compact2}-\ref{con:5kernel2},
then since $\calL^{*\star-1}$ is a bounded linear operator by Lemma \ref{lem:l2},
\bse
\|\E\{\calL^{*\star-1}(\V_{i,h}^{*\star})\}-\calL^{*\star-1}(\v^{*\star})\|_\infty
=\|\calL^{*\star-1}\{\E(\V_{i,h}^{*\star})-\v^{*\star}\}\|_\infty
=O(h^2).
\ese
Similarly, $\|\E\{(\calL^{*\star-1}\calL_{i,h}^{*\star})(\a^{*\star})\}-\a^{*\star}\|_\infty=O(h^2)$.
Using $\calL^{*\star-1}(\v^{*\star})(y,\bt)=\a^{*\star}(y,\bt)$, we further get
\be\label{eq:g3}
\|\E(\G_{i,h})\|_\infty
&=&\|\E\{\calL^{*\star-1}(\V_{i,h}^{*\star})\}
-\E\{(\calL^{*\star-1}\calL_{i,h}^{*\star})(\a^{*\star})\}
-\{\calL^{*\star-1}(\v^{*\star})-\a^{*\star}\}\|_\infty\n\\
&\leq&\|\E\{\calL^{*\star-1}(\V_{i,h}^{*\star})\}-\calL^{*\star-1}(\v)\|_\infty
+\|\E\{(\calL^{*\star-1}\calL_{i,h})(\a^{*\star})\}-\a^{*\star}\|_\infty\n\\
&=&O(h^2),
\ee
and  for $i\neq j$,
\be\label{eq:g2}
\left\|\E\left[\frac{\Ep^\star\{\G_{j,h}(Y,\bt)\rho^*(Y)\mid\x_i\}}
{\Ep^\star\{\rho^{*2}(Y)+\pi/(1-\pi)\rho^*(Y)\mid\x_i\}}\mid\x_i\right]\right\|_\infty
&=&\left\|\frac{\Ep^\star\left\{\E(\G_{j,h})(Y,\bt)\rho^*(Y)\mid\x_i\right\}}
{\Ep^\star\{\rho^{*2}(Y)+\pi/(1-\pi)\rho^*(Y)\mid\x_i\}}\right\|_\infty\n\\
&\le&\|\E(\G_{j,h})\|_\infty\left|\frac{\Ep^\star\{\rho^*(Y)\mid\x_i\}}
{\Ep^\star\{\rho^{*2}(Y)+\pi/(1-\pi)\rho^*(Y)\mid\x_i\}}\right|\n\\
&=&O(h^2).
\ee
Also, we have
\bse
\left\|\frac{\Ep^\star\{\a(Y,\bt)\rho^*(Y)\mid\x_i\}}
{\Ep^\star\{\rho^{*2}(Y)+\pi/(1-\pi)\rho^*(Y)\mid\x_i\}}\right\|_\infty
\leq
\|\a\|_\infty
\left|\frac{\Ep^\star\{\rho^*(Y)\mid\x_i\}}
{\Ep^\star\{\rho^{*2}(Y)+\pi/(1-\pi)\rho^*(Y)\mid\x_i\}}\right|
=O(\|\a\|_\infty),
\ese
and from \eqref{eq:a2}
\bse
&&\left\|\frac{\partial\b^{*\star}(\x_i,\wh\a^{*\star},\bzeta,\bt)}
{\partial\bzeta\trans}(\wh\bzeta-\bzeta)
-\frac{\partial\b^{*\star}(\x_i,\a^{*\star},\bzeta,\bt)}
{\partial\bzeta\trans}(\wh\bzeta-\bzeta)\right\|_\infty\\
&=&\left\|\frac{\partial}{\partial\bzeta\trans}
\left[\frac{\Ep^\star\{(\wh\a^{*\star}-\a^{*\star})(Y,\bt)\rho^*(Y)\mid\x_i,\bzeta\}}
{\Ep^\star\{\rho^{*2}(Y)+\pi/(1-\pi)\rho^*(Y)\mid\x_i,\bzeta\}}\right](\wh\bzeta-\bzeta)\right\|_\infty\\
&=&\left\|\frac{\Ep^\star\{(\wh\a^{*\star}-\a^{*\star})(Y,\bt)\rho^*(Y)
\S_\bzeta^{\star\rm T}(Y,\x_i,\bzeta)\mid\x_i\}}
{\Ep^\star\{\rho^{*2}(Y)+\pi/(1-\pi)\rho^*(Y)\mid\x_i\}}(\wh\bzeta-\bzeta)\right.\\
&&\left.-\Ep^\star\{(\wh\a^{*\star}-\a^{*\star})(Y,\bt)\rho^*(Y)\mid\x_i\}
\frac{\Ep[\{\rho^{*2}(Y)+\pi/(1-\pi)\rho^*(Y)\}\S_\bzeta^{\star\rm T}(Y,\x_i,\bzeta)\mid\x_i]}
{[\Ep^\star\{\rho^{*2}(Y)+\pi/(1-\pi)\rho^*(Y)\mid\x_i\}]^2}(\wh\bzeta-\bzeta)\right\|_\infty\\
&\leq&\left\|\wh\a^{*\star}-\a^{*\star}\right\|_\infty
\frac{\Ep^\star\{\rho^*(Y)\|\S_\bzeta^\star(Y,\x_i,\bzeta)\|_2\mid\x_i\}\|\wh\bzeta-\bzeta\|_2}
{\Ep^\star\{\rho^{*2}(Y)+\pi/(1-\pi)\rho^*(Y)\mid\x_i\}}\\
&&+\left\|\wh\a^{*\star}-\a^{*\star}\right\|_\infty\Ep^\star\{\rho^*(Y)\mid\x_i\}
\frac{\Ep^\star[\{\rho^{*2}(Y)+\pi/(1-\pi)\rho^*(Y)\}
\|\S_\bzeta^\star(Y,\x_i,\bzeta)\|_2\mid\x_i]\|\wh\bzeta-\bzeta\|_2}
{[\Ep^\star\{\rho^{*2}(Y)+\pi/(1-\pi)\rho^*(Y)\mid\x_i\}]^2}\\
&=&o_p(n_1^{-1/4})O_p(n_1^{-1/2})
=o_p(n_1^{-1/2}),
\ese
because $\|\wh\bzeta-\bzeta\|_2=O_p(n_1^{-1/2})$,
$\Ep^\star\{\|\S_\bzeta^\star(Y,\x,\bzeta)\|_2\mid\x\}$ is bounded,
and $\|\wh\a^{*\star}-\a^{*\star}\|_\infty
=O_p\left\{(n_1h)^{-1/2}\log n_1+h^2\right\}=o_p(n^{-1/4})$
by \eqref{eq:a2} and Condition \ref{con:6bandwidth2}. Hence, using \eqref{eq:a2} we get
\be\label{eq:b2}
&&\b^{*\star}(\x_i,\wh\a^{*\star},\wh\bzeta,\bt)\\
&=&\b^{*\star}(\x_i,\a^{*\star},\bzeta,\bt)
+\frac{\Ep^\star\{(\wh\a^{*\star}-\a^{*\star})(Y,\bt)\rho^*(Y)\mid\x_i\}}
{\Ep^\star\{\rho^{*2}(Y)+\pi/(1-\pi)\rho^*(Y)\mid\x_i\}}
+\frac{\partial\b^{*\star}(\x_i,\wh\a^{*\star},\bzeta,\bt)}{\partial\bzeta\trans}(\wh\bzeta-\bzeta)
+o_p(n_1^{-1/2})\n\\
&=&\b^{*\star}(\x_i,\a^{*\star},\bzeta,\bt)
+n_1^{-1}\sumj\frac{\Ep^\star\{\g_{j,h}(Y,\bt)\rho^*(Y)\mid\x_i\}}
{\Ep^\star\{\rho^{*2}(Y)+\pi/(1-\pi)\rho^*(Y)\mid\x_i\}}
+\frac{\partial\b^{*\star}(\x_i,\a^{*\star},\bzeta,\bt)}{\partial\bzeta\trans}(\wh\bzeta-\bzeta)
+o_p(n_1^{-1/2})\n
\ee
uniformly in $\x_i$ by Condition \ref{con:4compact2}.

Now we show the consistency of $\wh\bt$
by following Theorem 2.1 of \cite{newey1994large}.
We can view the problem of finding the solution for $\E\{\bphi\eff^{*\star}(\X,R,RY,\bt)\}=\0$
as maximizing the objective function $Q_0(\bt)\equiv -\|\E\{\bphi\eff^{*\star}(\X,R,RY,\bt)\}\|_2^2$,
then Theorem 2.1 of \cite{newey1994large} is directly applicable.
It is immediate that $\E\{\b^{*\star}(\X,\a^{*\star},\bzeta,\bt)\mid y\}=\E\{\U(\X,y,\bt)\mid y\}$
from the definition of $\b^{*\star}(\x,\a^{*\star},\bzeta,\bt)$.
Hence,
\bse
&&\E\{\bphi\eff^{*\star}(\X,R,RY,\bt)\}\\
&=&\A^{*\star}\E\left[\frac{R}{\pi}\rho^*(Y)\{\U(\X,Y,\bt)-\b^{*\star}(\X,\a^{*\star},\bzeta,\bt)\}
+\frac{1-R}{1-\pi}\b^{*\star}(\X,\a^{*\star},\bzeta,\bt)\right]\\
&=&\A^{*\star}\Ep\left[\rho^*(Y)\E\{\U(\X,Y,\bt)-\b^{*\star}(\X,\a^{*\star},\bzeta,\bt)\mid Y\}\right]
+\A^{*\star}\Eq\left[\E\left\{\b^{*\star}(\X,\a^{*\star},\bzeta,\bt)\mid Y\right\}\right]\\
&=&\A^{*\star}\Eq\left[\E\left\{\U(\X,Y,\bt)\mid Y\right\}\right]\\
&=&\0,
\ese
where the last step is by the definition of $\bt$.
Also, Condition \ref{con:0existence} implies that
$\bt$ is the unique solution for
$\E\{\bphi\eff^{*\star}(\X,R,RY,\bt)\}=\0$
in the neighborhood of $\bt$,
$\bt\in\boldsymbol{\Theta}$ which is compact,
and $\bphi\eff^{*\star}(\x,r,ry,\bt)$ is continuous with respect to $\bt$.
Therefore, it suffices to show that the estimating equation
converges in probability to $\E\{\bphi\eff^{*\star}(\X,R,RY,\bt)\}$ uniformly in $\bt$.
Using \eqref{eq:b2}, the estimating equation can be expressed as
\be\label{eq:ee}
&&\A^{*\star}n^{-1}\sumi\left[\frac{r_i}{\pi}\rho^*(y_i)
\left\{\U(\x_i,y_i,\bt)-\b^{*\star}(\x_i,\wh\a^{*\star},\wh\bzeta,\bt)\right\}
+\frac{1-r_i}{1-\pi}\b^{*\star}(\x_i,\wh\a^{*\star},\wh\bzeta,\bt)\right]\n\\
&=&n^{-1}\sumi\bphi\eff^{*\star}(\x_i,r_i,r_iy_i,\bt)\n
+\A^{*\star}n^{-1}\sumi\left\{\frac{r_i}{\pi}\rho^*(y_i)-\frac{1-r_i}{1-\pi}\right\}
\left\{\b^{*\star}(\x_i,\a^{*\star},\bzeta,\bt)
-\b^{*\star}(\x_i,\wh\a^{*\star},\wh\bzeta,\bt)\right\}\n\\
&=&n^{-1}\sumi\bphi\eff^{*\star}(\x_i,r_i,r_iy_i,\bt)
-\A^{*\star}\{T_1(\bt)+T_2(\bt)\}+o_p(n_1^{-1/2}),
\ee
where
\bse
T_1(\bt)&\equiv&n^{-1}n_1^{-1}\sumi\sumj
\left\{\frac{r_i}{\pi}\rho^*(y_i)-\frac{1-r_i}{1-\pi}\right\}
\frac{\Ep^\star\{\g_{j,h}(Y,\bt)\rho^*(Y)\mid\x_i\}}
{\Ep^\star\{\rho^{*2}(Y)+\pi/(1-\pi)\rho^*(Y)\mid\x_i\}},\\
T_2(\bt)&\equiv&n^{-1}\sumi
\left\{\frac{r_i}{\pi}\rho^*(y_i)-\frac{1-r_i}{1-\pi}\right\}
\frac{\partial\b^{*\star}(\x_i,\a^{*\star},\bzeta,\bt)}{\partial\bzeta\trans}(\wh\bzeta-\bzeta).
\ese
Using the property of the U-statistic, Condition \ref{con:6bandwidth2}, and \eqref{eq:g2},
we can rewrite $T_1(\bt)$ as
\bse
T_1(\bt)
&=&n_1^{-1}\sumi
\left\{\frac{r_i}{\pi}\rho^*(y_i)-\frac{1-r_i}{1-\pi}\right\}
\E\left[\frac{\Ep^\star\{\G_{j,h}(Y,\bt)\rho^*(Y)\mid\x_i\}}
{\Ep^\star\{\rho^{*2}(Y)+\pi/(1-\pi)\rho^*(Y)\mid\x_i\}}\mid\x_i,r_i,r_iy_i\right]\n\\
&&+n_1^{-1}\sumj
\E\left[\left\{\frac{R_i}{\pi}\rho^*(Y_i)-\frac{1-R_i}{1-\pi}\right\}
\frac{\Ep^\star\{\g_{j,h}(Y,\bt)\rho^*(Y)\mid\X_i\}}
{\Ep^\star\{\rho^{*2}(Y)+\pi/(1-\pi)\rho^*(Y)\mid\X_i\}}\mid\x_j,r_j,r_jy_j\right]\n\\
&&-n^{1/2}n_1^{-1}
\E\left[\left\{\frac{R_i}{\pi}\rho^*(Y_i)-\frac{1-R_i}{1-\pi}\right\}
\frac{\Ep^\star\{\G_{j,h}(Y,\bt)\rho^*(Y)\mid\X_i\}}
{\Ep^\star\{\rho^{*2}(Y)+\pi/(1-\pi)\rho^*(Y)\mid\X_i\}}\right]+O_p(n_1^{-1})\n\\
&=&n_1^{-1}\sumj
\E\left[\left\{\frac{R_i}{\pi}\rho^*(Y_i)-\frac{1-R_i}{1-\pi}\right\}
\frac{\Ep^\star\{\g_{j,h}(Y,\bt)\rho^*(Y)\mid\X_i\}}
{\Ep^\star\{\rho^{*2}(Y)+\pi/(1-\pi)\rho^*(Y)\mid\X_i\}}\mid\x_j,r_j,r_jy_j\right]\n\\
&&+O_p(n_1^{-1}n h^2+n_1^{-1})\n\\
&=&n_1^{-1}\sumj\int\{\rho^*(y)\py(y)-\qy(y)\}
\E\left[\frac{\Ep^\star\{\g_{j,h}(Y,\bt)\rho^*(Y)\mid\X\}}
{\Ep^\star\{\rho^{*2}(Y)+\pi/(1-\pi)\rho^*(Y)\mid\X\}}\mid y\right]dy+o_p(n_1^{-1/2})\n\\
&=&n_1^{-1}\sumj\int\calL^{*\star}(\g_{j,h})(y,\bt)
\left\{\rho^*(y)-\rho(y)\right\}dy+o_p(n_1^{-1/2})\n\\
&=&n_1^{-1}\sumj\int\{\v_{j,h}^{*\star}(y,\bt)-\calL_{j,h}^{*\star}(\a^{*\star})(y,\bt)\}
\left\{\rho^*(y)-\rho(y)\right\}dy+o_p(n_1^{-1/2})\n\\
&=&n_1^{-1}\sumj r_j\int K_h(y-y_j)
\{\U(\x_j,y,\bt)-\b^{*\star}(\x_j,\a^{*\star},\bzeta,\bt)\}
\left\{\rho^*(y)-\rho(y)\right\}dy+o_p(n_1^{-1/2}),
\ese
where
\bse
&&\int K_h(y-y_j)\{\U(\x_j,y,\bt)-\b^{*\star}(\x_j,\a^{*\star},\bzeta,\bt)\}
\left\{\rho^*(y)-\rho(y)\right\}dy\\
&=&\{\U(\x_j,y_j,\bt)-\b^{*\star}(\x_j,\a^{*\star},\bzeta,\bt)\}\{\rho^*(y_j)-\rho(y_j)\}\\
&&+\left[\{\U(\x_j,y_j,\bt)-\b^{*\star}(\x_j,\a^{*\star},\bzeta,\bt)\}
\{{\rho^*}''(y_j)-\rho''(y_j)\}
+2\U'_y(\x_j,y_j,\bt)\{{\rho^*}'(y_j)-\rho'(y_j)\}\right.\\
&&+\left.\U''_{yy}(\x_j,y_j,\bt)\{\rho^*(y_j)-\rho(y_j)\}\right]\frac{h^2}{2}\int
t^2K(t)dt+O(h^4)\\
&=&\{\U(\x_j,y_j,\bt)-\b^{*\star}(\x_j,\a^{*\star},\bzeta,\bt)\}
\{\rho^*(y_j)-\rho(y_j)\}+o(n_1^{-1/2})
\ese
under Conditions \ref{con:0existence}, \ref{con:2rho2}, \ref{con:5kernel2}, and \ref{con:6bandwidth2}.
On the other hand,
using $\E\{\b^{*\star}(\X,\a^{*\star},\bzeta,\bt)\mid y\}=\E\{\U(\X,y,\bt)\mid y\}$,
$T_2(\bt)$ can be written as
\bse
T_2(\bt)
&=&\E\left[\left\{\frac{R}{\pi}\rho^*(Y)-\frac{1-R}{1-\pi}\right\}
\frac{\partial \b^{*\star}(\X,a^{*\star},\bzeta,\bt)}{\partial\bzeta\trans}\right]
(\wh\bzeta-\bzeta)+o_p(n_1^{-1/2})\\
&=&\frac{\partial}{\partial\bzeta\trans}
\E\left[\left\{\frac{R}{\pi}\rho^*(Y)-\frac{1-R}{1-\pi}\right\}\E\{\U(\X,Y,\bt)\mid Y\}\right]
(\wh\bzeta-\bzeta)+o_p(n_1^{-1/2})\\
&=&o_p(n_1^{-1/2}).
\ese
Hence, \eqref{eq:ee} leads to
\be\label{eq:ee2}
&&\A^{*\star}n^{-1}\sumi\left[\frac{r_i}{\pi}\rho^*(y_i)
\left\{\U(\x_i,y_i,\bt)-\b^{*\star}(\x_i,\wh\a^{*\star},\wh\bzeta,\bt)\right\}
+\frac{1-r_i}{1-\pi}\b^{*\star}(\x_i,\wh\a^{*\star},\wh\bzeta,\bt)\right]\\
&=&n^{-1}\sumi\bphi\eff^{*\star}(\x_i,r_i,r_iy_i,\bt)
+\A^{*\star}n_1^{-1}\sumi r_i\{\b^{*\star}(\x_i,\a^{*\star},\bzeta,\bt)-\U(\x_i,y_i,\bt)\}
\{\rho^*(y_i)-\rho(y_i)\}+o_p(n_1^{-1/2})\n\\
&=&\E\left\{\bphi\eff^{*\star}(\X,R,RY,\bt)\right\}
+\A^{*\star}\Ep\left[\{\b^{*\star}(\X,\a^{*\star},\bzeta,\bt)-\U(\X,Y,\bt)\}
\{\rho^*(Y)-\rho(Y)\}\right]+O_p(n_1^{-1/2})\n\\
&=&\E\left\{\bphi\eff^{*\star}(\X,R,RY,\bt)\right\}+O_p(n_1^{-1/2}),\n
\ee
since $\E\{\b^{*\star}(\X,\a^{*\star},\bzeta,\bt)\mid y\}=\E\{\U(\X,y,\bt)\mid y\}$.
This implies that the estimating equation
converges in probability to $\E\{\bphi\eff^{*\star}(\X,R,RY,\bt)\}$
uniformly in $\bt$ by Condition \ref{con:0existence}.
Hence, $\wh\bt$ is consistent for $\bt$.

Finally, we derive the asymptotic distribution of $\wh\bt$.
By the definition of $\wh\bt$ and \eqref{eq:ee2},
\bse
\0
&=&\sqrt{n_1}\A^{*\star}n^{-1}\sumi\left[\frac{r_i}{\pi}\rho^*(y_i)
\left\{\U(\x_i,y_i,\wh\bt)-\b^{*\star}(\x_i,\wh\a^{*\star},\wh\bzeta,\wh\bt)\right\}
+\frac{1-r_i}{1-\pi}\b^{*\star}(\x_i,\wh\a^{*\star},\wh\bzeta,\wh\bt)\right]\\
&=&n^{-1/2}\sumi\left[\sqrt{\pi}\bphi\eff^{*\star}(\x_i,r_i,r_iy_i,\bt)
+\frac{r_i}{\sqrt{\pi}}\A^{*\star}
\{\b^{*\star}(\x_i,\a^{*\star},\bzeta,\bt)-\U(\x_i,y_i,\bt)\}
\{\rho^*(y_i)-\rho(y_i)\}\right]\\
&&+\wh\B\sqrt{n_1}(\wh\bt-\bt)+o_p(1),
\ese
where
\bse
\wh\B
&=&\E\left\{\frac{\partial\bphi\eff^{*\star}(\X,R,RY,\bt)}{\partial\bt\trans}\right\}
+\A^{*\star}\Ep\left[\frac{\partial\{\b^{*\star}(\X,\a^{*\star},\bzeta,\bt)-\U(\X,Y,\bt)\}}
{\partial\bt\trans}\{\rho^*(Y)-\rho(Y)\}\right]+o_p(1)\\
&=&\E\left\{\frac{\partial\bphi\eff^{*\star}(\X,R,RY,\bt)}{\partial\bt\trans}\right\}
+\A^{*\star}\Ep\left[\frac{\partial}{\partial\bt\trans}
E\{\b^{*\star}(\X,\a^{*\star},\bzeta,\bt)-\U(\X,Y,\bt)\mid Y\}
\{\rho^*(Y)-\rho(Y)\}\right]+o_p(1)\\
&=&\E\left\{\frac{\partial\bphi\eff^{*\star}(\X,R,RY,\bt)}{\partial\bt\trans}\right\}+o_p(1)
\ese
because $\wh\bt$ is consistent for $\bt$
and $\E\{\b^{*\star}(\X,\a^{*\star},\bzeta,\bt)\mid y\}=\E\{\U(\X,y,\bt)\mid y\}$.
In addition,
\bse
\E\left\{\frac{\partial\bphi\eff^{*\star}(\X,R,RY,\bt)}{\partial\bt\trans}\right\}
&=&\A^{*\star}\frac{\partial}{\partial\bt\trans}
\E\left[\frac{R}{\pi}\rho^*(Y)\{\U(\X,Y,\bt)-\b^{*\star}(\X,\bt)\}
+\frac{1-R}{1-\pi}\b^{*\star}(\X,\bt)\right]\\
&=&\A^{*\star}\frac{\partial}{\partial\bt\trans}
\left(\Ep\left[\rho^*(Y)\E\{\U(\X,Y,\bt)-\b^{*\star}(\X,\bt)\mid Y\}\right]
+\Eq\left[\E\left\{\b^{*\star}(\X,\bt)\mid Y\right\}\right]\right)\\
&=&\A^{*\star}\frac{\partial}{\partial\bt\trans}\Eq\left[\E\{\U(\X,Y,\bt)\mid Y\}\right]\\
&=&\A^{*\star}\A^{-1}.
\ese
Therefore, $\sqrt{n_1}(\wh\bt-\bt)$ converges in distribution to
$N\left\{\0,\A\A^{*\star-1}\bSigma(\A\A^{*\star-1})\trans\right\}$
as $n_1\to\infty$, where $\bSigma$ is given in (\ref{eq:bSigma}).
\qed

\subsubsection{Proof of Theorem \ref{th:thetanew2}}\label{sec:thetanew2proof}
We define
\bse
\b^*(\x,\a,\Ep,\bt)
\equiv\frac{\Ep\{\U(\x,Y,\bt)\rho^{*2}(Y)+\a(Y,\bt)\rho^*(Y)\mid\x\}}
{\Ep\{\rho^{*2}(Y)+\pi/(1-\pi)\rho^*(Y)\mid\x\}},
\ese
and for any function $g(\cdot,\mu)$, define
its $k$th Gateaux derivative with respect to $\mu$ at $\mu_1$ in the direction $\mu_2$ as
\bse
\frac{\partial^k g(\cdot,\mu_1)}{\partial\mu^k}(\mu_2)\equiv
\frac{\partial^k g(\cdot,\mu)}{\partial\mu^k}(\mu_2)\Big|_{\mu=\mu_1}\equiv
\frac{\partial^k g(\cdot,\mu_1+h\mu_2)}{\partial h^k}\bigg|_{h=0}.
\ese
Then we have
\bse
&&\frac{\partial\b^*(\x,\a,\Ep,\bt)}{\partial\Ep}(\wh\E_p-\Ep)\\
&=&\frac{(\wh\E_p-\Ep)\{\U(\x,Y,\bt)\rho^{*2}(Y)+\a(Y,\bt)\rho^*(Y)\mid\x\}}
{\Ep\{\rho^{*2}(Y)+\pi/(1-\pi)\rho^*(Y)\mid\x\}}\\
&&-\Ep\{\U(\x,Y,\bt)\rho^{*2}(Y)+\a(Y,\bt)\rho^*(Y)\mid\x\}
\frac{(\wh\E_p-\Ep)\{\rho^{*2}(Y)+\pi/(1-\pi)\rho^*(Y)\mid\x\}}
{[\Ep\{\rho^{*2}(Y)+\pi/(1-\pi)\rho^*(Y)\mid\x\}]^2}\\
&=&o_p(n_1^{-1/4}),\\\\
&&\frac{\partial^2\b^*(\x,\a,\mu,\bt)}{\partial{\Ep}^2}(\wh\E_p-\Ep)\\
&=&\frac{-2(\wh\E_p-\Ep)\{\rho^{*2}(Y)+\pi/(1-\pi)\rho^*(Y)\mid\x\}}
{[\Ep\{\rho^{*2}(Y)+\pi/(1-\pi)\rho^*(Y)\mid\x\}]^3}\\
&&\times\left[(\wh\E_p-\Ep)\{\U(\x,Y,\bt)\rho^{*2}(Y)+\a(Y,\bt)\rho^*(Y)\mid\x\}
\mu\{\rho^{*2}(Y)+\pi/(1-\pi)\rho^*(Y)\mid\x\}\right.\\
&&\left.-\mu\{\U(\x,Y,\bt)\rho^{*2}(Y)+\a(Y,\bt)\rho^*(Y)\mid\x\}
(\wh\E_p-\Ep)\{\rho^{*2}(Y)+\pi/(1-\pi)\rho^*(Y)\mid\x\}\right]\\
&=&o_p(n_1^{-1/2})
\ese
for any bounded $\a(y)$ and $\mu(\cdot\mid\x)$
since $\|(\wh\E_p-\Ep)(\cdot\mid\x)\|_\infty=o_p(n_1^{-1/4})$ by the
assumption,
and these hold uniformly in $\x$ by Condition \ref{con:4compact2}.
Then by the Taylor expansion and mean value theorem,
for any bounded $\a(y)$ and some $\alpha\in(0,1)$,
\be\label{eq:b3}
\b^*(\x,\a,\wh\E_p,\bt)
&=&\b^*(\x,\a,\Ep,\bt)
+\frac{\partial\b^*(\x,\a,\Ep,\bt)}{\partial{\Ep}}(\wh\E_p-\Ep)
+\frac{1}{2}\frac{\partial^2\b^*\{\x,\a,\Ep+\alpha(\wh\E_p-\Ep),\bt\}}{\partial{\Ep}^2}(\wh\E_p-\Ep)\n\\
&=&\b^*(\x,\a,\Ep,\bt)+\frac{\partial \b^*(\x,\a,\Ep,\bt)}{\partial{\Ep}}(\wh\E_p-\Ep)+o_p(n_1^{-1/2}).
\ee
Noting that $\a^*(y,\bt)$ is bounded under Condition
\ref{con:3bounded2}, we further get
\bse
&&\wh\calL^*(\a^*)(y,\bt)-\wh\v^*(y,\bt)\\
&=&n_1^{-1}\sumi r_iK_h(y-y_i)\{\b^*(\x_i,\a^*,\wh\E_p,\bt)-\U(\x_i,y,\bt)\}\n\\
&=&n_1^{-1}\sumi r_iK_h(y-y_i)\left\{\b^*(\x_i,a^*,\Ep,\bt)
+\frac{\partial\b^*(\x_i,\a^*,\Ep,\bt)}{\partial{\Ep}}(\wh\E_p-\Ep)
+o_p(n_1^{-1/2})-\U(\x_i,y,\bt)\right\}\n\\
&=&n_1^{-1}\sumi\left\{\calL_{i,h}^*(\a^*)(y,\bt)-\v^*_{i,h}(y,\bt)\right\}+o_p(n_1^{-1/2})
\ese
uniformly in $y$ by Condition \ref{con:4compact2}.
The last equality above is
because $\E\{\b^*(\X,\a^*,\Ep,\bt)\mid y\}=\E\{\U(\X,y,\bt)\mid y\}$ from the definition of $\a^*$,
hence
\bse
&&n_1^{-1}\sumi r_iK_h(y-y_i)\frac{\partial\b^*(\x_i,\a^*,\Ep,\bt)}{\partial{\Ep}}(\wh\E_p-\Ep)\\
&=&n_1^{-1}\sumi r_iK_h(y-y_i)\frac{\partial\b^*(\x_i,\a^*,\Ep,\bt)}{\partial{\Ep}}(\wh\E_p-\Ep)
-\frac{\partial[\py(y)\E\{\b^*(\X,\a^*,\Ep,\bt)\mid y\}]}{\partial{\Ep}}(\wh\E_p-\Ep)\\
&=&n_1^{-1/4}\left[n_1^{-1}\sumi r_iK_h(y-y_i)n_1^{1/4}
\frac{\partial\b^*(\x_i,\a^*,\Ep,\bt)}{\partial{\Ep}}(\wh\E_p-\Ep)
-\py(y)\E\left\{n_1^{1/4}\frac{\partial\b^*(\X,\a^*,\Ep,\bt)}
{\partial{\Ep}}(\wh\E_p-\Ep)\mid y\right\}\right]\\
&=&n_1^{-1/4}O_p\left\{(n_1h)^{-1/2}\log{n_1}+h^2\right\}\\
&=&o_p(n_1^{-1/2})
\ese
under Conditions \ref{con:3bounded2}-\ref{con:6bandwidth2}.
In addition, for any bounded function $\a(y,\bt)$,
\bse
\|(\wh\calL^*-\calL^*)(\a)\|_\infty
=O_p\left\{(n_1h)^{-1/2}\log n_1+h^2\right\}=o_p(n_1^{-1/4})
\ese
under Conditions \ref{con:0existence}, \ref{con:3bounded2}-\ref{con:6bandwidth2}
and the assumption $\|(\wh\E_p-\Ep)(\cdot\mid\x)\|_\infty=o_p(n_1^{-1/4})$.
Similarly, $\|\wh\v^*-\v^*\|_\infty=o_p(n_1^{-1/4})$.
Then using that $\calL^{*-1}$ is a bounded linear operator by Lemma \ref{lem:l2},
$\wh\a^*(y,\bt)$ can be expressed as
\be\label{eq:a3}
\wh\a^*(y,\bt)
&=&\{\calL^*+(\wh\calL^*-\calL^*)\}^{-1}(\wh\v^*)(y,\bt)\n\\
&=&\{\calL^{*-1}-\calL^{*-1}(\wh\calL^*-\calL^*)\calL^{*-1}\}
\{\v^*+(\wh\v^*-\v^*)\}(y,\bt)+o_p(n_1^{-1/2})\n\\
&=&\a^*(y,\bt)+\calL^{*-1}(\wh\v^*-\v^*)(y,\bt)
-\{\calL^{*-1}(\wh\calL^*-\calL^*)\}(\a^*)(y,\bt)+o_p(n_1^{-1/2})\n\\
&=&\a^*(y,\bt)+\calL^{*-1}\{\wh\v^*-\wh\calL^*(\a^*)\}(y,\bt)
+o_p(n_1^{-1/2})\n\\
&=&\a^*(y,\bt)+n_1^{-1}\sumi\g_{i,h}(y,\bt)+o_p(n_1^{-1/2})
\ee
uniformly in $y$,
where $\g_{i,h}(y,\bt)\equiv\calL^{*-1}\{\v^*_{i,h}-\calL_{i,h}^*(\a^*)\}(y,\bt)$.
We also have for any bounded $\a(y,\bt)$,
\bse
\left\|\frac{\Ep\{\a(Y,\bt)\rho^*(Y)\mid\x_i\}}
{\Ep\{\rho^{*2}(Y)+\pi/(1-\pi)\rho^*(Y)\mid\x_i\}}\right\|_\infty
\leq
\|\a\|_\infty
\left|\frac{\Ep\{\rho^*(Y)\mid\x_i\}}
{\Ep\{\rho^{*2}(Y)+\pi/(1-\pi)\rho^*(Y)\mid\x_i\}}\right|
=O(\|\a\|_\infty),
\ese
and
\bse
&&\left\|\frac{\partial\b^*(\x_i,\wh\a^*,\Ep,\bt)}{\partial\Ep}(\wh\E_p-\Ep)
-\frac{\partial\b^*(\x_i,\a^*,\Ep,\bt)}{\partial\Ep}(\wh\E_p-\Ep)\right\|_\infty\\
&=&\left\|\frac{\partial}{\partial\Ep}
\left[\frac{\Ep\{(\wh\a^*-\a^*)(Y,\bt)\rho^*(Y)\mid\x_i\}}
{\Ep\{\rho^{*2}(Y)+\pi/(1-\pi)\rho^*(Y)\mid\x_i\}}\right](\wh\E_p-\Ep)\right\|_\infty\\
&=&\left\|\frac{(\wh\E_p-\Ep)\{(\wh\a^*-\a^*)(Y,\bt)\rho^*(Y)\mid\x_i\}}
{\Ep\{\rho^{*2}(Y)+\pi/(1-\pi)\rho^*(Y)\mid\x_i\}}\right.\\
&&\left.-\Ep\{(\wh\a^*-\a^*)(Y,\bt)\rho^*(Y)\mid\x_i\}
\frac{(\wh\E_p-\Ep)\{\rho^{*2}(Y)+\pi/(1-\pi)\rho^*(Y)\mid\x_i\}}
{[\Ep\{\rho^{*2}(Y)+\pi/(1-\pi)\rho^*(Y)\mid\x_i\}]^2}\right\|_\infty\\
&\leq&n_1^{-1/4}\frac{\|(\wh\E_p-\Ep)\{n_1^{1/4}(\wh\a^*-\a^*)(Y,\bt)\rho^*(Y)\mid\x_i\}\|_\infty}
{\Ep\{\rho^{*2}(Y)+\pi/(1-\pi)\rho^*(Y)\mid\x_i\}}\\
&&+\left\|\wh\a^*-\a^*\right\|_\infty\Ep\{\rho^*(Y)\mid\x_i\}
\frac{|(\wh\E_p-\Ep)\{\rho^{*2}(Y)+\pi/(1-\pi)\rho^*(Y)\mid\x_i\}|}
{[\Ep\{\rho^{*2}(Y)+\pi/(1-\pi)\rho^*(Y)\mid\x_i\}]^2}\\
&=&n_1^{-1/4}o_p(n_1^{-1/4})+o_p(n_1^{-1/4})o_p(n_1^{-1/4})\\
&=&o_p(n_1^{-1/2})
\ese
because $\|(\wh\E_p-\Ep)(\cdot\mid\x)\|_\infty=o_p(n_1^{-1/4})$
and $\|\wh\a^{*\star}-\a^{*\star}\|_\infty
=O_p\left\{(n_1h)^{-1/2}\log n_1+h^2\right\}=o_p(n^{-1/4})$
by \eqref{eq:g3}, \eqref{eq:a3}, and Condition \ref{con:6bandwidth2}.
Hence using \eqref{eq:b3} and \eqref{eq:a3},
\be\label{eq:b4}
&&\b^*(\x_i,\wh\a^*,\wh\E_p,\bt)\\
&=&\b^*(\x_i,\a^*,\Ep,\bt)
+\frac{\Ep\{(\wh\a^*-\a^*)(Y,\bt)\rho^*(Y)\mid\x_i\}}
{\Ep\{\rho^{*2}(Y)+\pi/(1-\pi)\rho^*(Y)\mid\x_i\}}
+\frac{\partial\b^*(\x_i,\wh\a^*,\Ep,\bt)}{\partial\Ep}(\wh\E_p-\Ep)
+o_p(n_1^{-1/2})\n\\
&=&\b^*(\x_i,\a^*,\Ep,\bt)
+n_1^{-1}\sumj\frac{\Ep\{\g_{j,h}(Y,\bt)\rho^*(Y)\mid\x_i\}}
{\Ep\{\rho^{*2}(Y)+\pi/(1-\pi)\rho^*(Y)\mid\x_i\}}
+\frac{\partial\b^*(\x_i,\a^*,\Ep,\bt)}{\partial\Ep}(\wh\E_p-\Ep)
+o_p(n_1^{-1/2})\n
\ee
uniformly in $\x_i$ by Condition \ref{con:4compact2}.

Now we show the consistency of $\wt\bt$
by following Theorem 2.1 of \cite{newey1994large}.
We can view the problem of finding the solution for $\E\{\bphi\eff^*(\X,R,RY,\bt)\}=\0$
as maximizing the objective function $Q_0(\bt)\equiv -\|\E\{\bphi\eff^*(\X,R,RY,\bt)\}\|_2^2$,
then Theorem 2.1 of \cite{newey1994large} is directly applicable.
It is immediate that $\E\{\b^*(\X,\a^*,\Ep,\bt)\mid y\}=\E\{\U(\X,y,\bt)\mid y\}$
from the definition of $\b^*(\x,\a^*,\Ep,\bt)$.
Hence,
\bse
&&\E\{\bphi\eff^*(\X,R,RY,\bt)\}\\
&=&\A^*\E\left[\frac{R}{\pi}\rho^*(Y)\{\U(\X,Y,\bt)-\b^*(\X,\a^*,\Ep,\bt)\}
+\frac{1-R}{1-\pi}\b^*(\X,\a^*,\Ep,\bt)\right]\\
&=&\A^*\Ep\left[\rho^*(Y)\E\{\U(\X,Y,\bt)-\b^*(\X,\a^*,\Ep,\bt)\mid Y\}\right]
+\A^*\Eq\left[\E\left\{\b^*(\X,\a^*,\Ep,\bt)\mid Y\right\}\right]\\
&=&\A^*\Eq\left[\E\left\{\U(\X,Y,\bt)\mid Y\right\}\right]\\
&=&\0,
\ese
where the last step is by the definition of $\bt$.
Also, Condition \ref{con:0existence} implies that
$\bt$ is the unique solution for $\E\{\bphi\eff^*(\X,R,RY,\bt)\}=\0$
in the neighborhood of $\bt$,
$\bt\in\boldsymbol{\Theta}$ which is compact,
and $\bphi\eff^*(\x,r,ry,\bt)$ is continuous with respect to $\bt$.
Therefore, it suffices to show that the estimating equation
converges in probability to $\E\{\bphi\eff^*(\X,R,RY,\bt)\}$ uniformly in $\bt$.
Using \eqref{eq:b4}, the estimating equation can be expressed as
\be\label{eq:ee3}
&&\A^*n^{-1}\sumi\left[\frac{r_i}{\pi}\rho^*(y_i)
\left\{\U(\x_i,y_i,\bt)-\b^*(\x_i,\wh\a^*,\wh\E_p,\bt)\right\}
+\frac{1-r_i}{1-\pi}\b^*(\x_i,\wh\a^*,\wh\E_p,\bt)\right]\n\\
&=&n^{-1}\sumi\bphi\eff^*(\x_i,r_i,r_iy_i,\bt)
+\A^*n^{-1}\sumi\left\{\frac{r_i}{\pi}\rho^*(y_i)-\frac{1-r_i}{1-\pi}\right\}
\left\{\b^*(\x_i,\a^*,\Ep,\bt)
-\b^*(\x_i,\wh\a^*,\wh\E_p,\bt)\right\}\n\\
&=&n^{-1}\sumi\bphi\eff^*(\x_i,r_i,r_iy_i,\bt)-\A^*\{T_1(\bt)+T_2(\bt)\}+o_p(n_1^{-1/2}),
\ee
where
\bse
T_1(\bt)&\equiv&n^{-1}n_1^{-1}\sumi\sumj
\left\{\frac{r_i}{\pi}\rho^*(y_i)-\frac{1-r_i}{1-\pi}\right\}
\frac{\Ep\{\g_{j,h}(Y,\bt)\rho^*(Y)\mid\x_i\}}
{\Ep\{\rho^{*2}(Y)+\pi/(1-\pi)\rho^*(Y)\mid\x_i\}},\\
T_2(\bt)&\equiv&n^{-1}\sumi
\left\{\frac{r_i}{\pi}\rho^*(y_i)-\frac{1-r_i}{1-\pi}\right\}
\frac{\partial\b^*(\x_i,\a^*,\Ep,\bt)}{\partial\Ep}(\wh\E_p-\Ep).
\ese
Using the property of the U-statistic, Condition \ref{con:6bandwidth2}, and \eqref{eq:g2},
we can rewrite $T_1(\bt)$ as
\bse
T_1(\bt)
&=&n_1^{-1}\sumi
\left\{\frac{r_i}{\pi}\rho^*(y_i)-\frac{1-r_i}{1-\pi}\right\}
\E\left[\frac{\Ep\{\G_{j,h}(Y,\bt)\rho^*(Y)\mid\x_i\}}
{\Ep\{\rho^{*2}(Y)+\pi/(1-\pi)\rho^*(Y)\mid\x_i\}}\mid\x_i,r_i,r_iy_i\right]\n\\
&&+n_1^{-1}\sumj
\E\left[\left\{\frac{R_i}{\pi}\rho^*(Y_i)-\frac{1-R_i}{1-\pi}\right\}
\frac{\Ep\{\g_{j,h}(Y,\bt)\rho^*(Y)\mid\X_i\}}
{\Ep\{\rho^{*2}(Y)+\pi/(1-\pi)\rho^*(Y)\mid\X_i\}}\mid\x_j,r_j,r_jy_j\right]\n\\
&&-n^{1/2}n_1^{-1}
\E\left[\left\{\frac{R_i}{\pi}\rho^*(Y_i)-\frac{1-R_i}{1-\pi}\right\}
\frac{\Ep\{\G_{j,h}(Y,\bt)\rho^*(Y)\mid\X_i\}}
{\Ep\{\rho^{*2}(Y)+\pi/(1-\pi)\rho^*(Y)\mid\X_i\}}\right]+O_p(n_1^{-1})\n\\
&=&n_1^{-1}\sumj
\E\left[\left\{\frac{R_i}{\pi}\rho^*(Y_i)-\frac{1-R_i}{1-\pi}\right\}
\frac{\Ep\{\g_{j,h}(Y,\bt)\rho^*(Y)\mid\X_i\}}
{\Ep\{\rho^{*2}(Y)+\pi/(1-\pi)\rho^*(Y)\mid\X_i\}}\mid\x_j,r_j,r_jy_j\right]\n\\
&&+O_p(n_1^{-1}nh^2 + n_1^{-1})\n\\
&=&n_1^{-1}\sumj\int\{\rho^*(y)\py(y)-\qy(y)\}
\E\left[\frac{\Ep\{\g_{j,h}(Y,\bt)\rho^*(Y)\mid\X\}}
{\Ep\{\rho^{*2}(Y)+\pi/(1-\pi)\rho^*(Y)\mid\X\}}\mid y\right]dy+o_p(n_1^{-1/2})\n\\
&=&n_1^{-1}\sumj\int\calL^*(\g_{j,h})(y,\bt)
\left\{\rho^*(y)-\rho(y)\right\}dy+o_p(n_1^{-1/2})\n\\
&=&n_1^{-1}\sumj\int\{\v_{j,h}^*(y,\bt)-\calL_{j,h}^*(\a^*)(y,\bt)\}
\left\{\rho^*(y)-\rho(y)\right\}dy+o_p(n_1^{-1/2})\n\\
&=&n_1^{-1}\sumj r_j\int K_h(y-y_j)
\{\U(\x_j,y,\bt)-\b^*(\x_j,\a^*,\Ep,\bt)\}
\left\{\rho^*(y)-\rho(y)\right\}dy+o_p(n_1^{-1/2}),
\ese
where
\bse
&&\int K_h(y-y_j)\{\U(\x_j,y,\bt)-\b^*(\x_j,\a^*,\Ep,\bt)\}
\left\{\rho^*(y)-\rho(y)\right\}dy\\
&=&\{\U(\x_j,y_j,\bt)-\b^*(\x_j,\a^*,\Ep,\bt)\}\{\rho^*(y_j)-\rho(y_j)\}\\
&&+\left[\{\U(\x_j,y_j,\bt)-\b^*(\x_j,\a^*,\Ep,\bt)\}
\{{\rho^*}''(y_j)-\rho''(y_j)\}
+2\U'_y(\x_j,y_j,\bt)\{{\rho^*}'(y_j)-\rho'(y_j)\}\right.\\
&&+\left.\U''_{yy}(\x_j,y_j,\bt)\{\rho^*(y_j)-\rho(y_j)\}\right]\frac{h^2}{2}\int t^2K(t)dt+O(h^4)\\
&=&\{\U(\x_j,y_j,\bt)-\b^*(\x_j,\a^*,\Ep,\bt)\}
\{\rho^*(y_j)-\rho(y_j)\}+o(n_1^{-1/2})
\ese
under Conditions \ref{con:0existence}, \ref{con:2rho2},
\ref{con:5kernel2}, and \ref{con:6bandwidth2}.
On the other hand, using $\E\{\b^*(\X,\a^*,\Ep,\bt)\mid y\}=\E\{\U(\X,y,\bt)\mid y\}$,
$T_2(\bt)$ can be written as
\bse
T_2(\bt)
&=&n^{-1}\sumi\left\{\frac{r_i}{\pi}\rho^*(y_i)-\frac{1-r_i}{1-\pi}\right\}
\frac{\partial\b^*(\x_i,\a^*,\Ep,\bt)}{\partial{\Ep}}(\wh\E_p-\Ep)\\
&&-\frac{\partial}{\partial\Ep}\E\left[\left\{\frac{R}{\pi}\rho^*(Y)-\frac{1-R}{1-\pi}\right\}
\E\{\U(\X,Y,\bt)\mid y\}\right](\wh\E_p-\Ep)\\
&=&n_1^{-1/4}\left(n^{-1}\sumi\left\{\frac{r_i}{\pi}\rho^*(y_i)-\frac{1-r_i}{1-\pi}\right\}
n_1^{1/4}\frac{\partial\b^*(\x_i,\a^*,\Ep,\bt)}{\partial{\Ep}}(\wh\E_p-\Ep)\right.\\
&&\left.-\E\left[\left\{\frac{R}{\pi}\rho^*(Y)-\frac{1-R}{1-\pi}\right\}
n_1^{1/4}\frac{\partial\b^*(\X,\a^*,\Ep,\bt)}{\partial\Ep}(\wh\E_p-\Ep)\right]\right)\\
&=&n_1^{-1/4}O_p(n^{-1/2})\\
&=&o_p(n_1^{-1/2}).
\ese
Hence, \eqref{eq:ee3} leads to
\be\label{eq:ee4}
&&\A^*n^{-1}\sumi\left[\frac{r_i}{\pi}\rho^*(y_i)
\left\{\U(\x_i,y_i,\bt)-\b^*(\x_i,\wh\a^*,\wh\E_p,\bt)\right\}
+\frac{1-r_i}{1-\pi}\b^*(\x_i,\wh\a^*,\wh\E_p,\bt)\right]\\
&=&n^{-1}\sumi\bphi\eff^*(\x_i,r_i,r_iy_i,\bt)
+\A^*n_1^{-1}\sumi r_i\{\b^*(\x_i,\a^*,\Ep,\bt)-\U(\x_i,y_i,\bt)\}
\{\rho^*(y_i)-\rho(y_i)\}+o_p(n_1^{-1/2})\n\\
&=&\E\left\{\bphi\eff^*(\X,R,RY,\bt)\right\}
+\A^*\Ep\left[\{\b^*(\X,\a^*,\Ep,\bt)-\U(\X,Y,\bt)\}
\{\rho^*(Y)-\rho(Y)\}\right]+O_p(n_1^{-1/2})\n\\
&=&\E\left\{\bphi\eff^*(\X,R,RY,\bt)\right\}+O_p(n_1^{-1/2}),\n
\ee
since $\E\{\b^*(\X,\a^*,\Ep,\bt)\mid y\}=\E\{\U(\X,y,\bt)\mid y\}$.
This implies that the estimating equation
converges in probability to $\E\{\bphi\eff^*(\X,R,RY,\bt)\}$
uniformly in $\bt$ by Condition \ref{con:0existence}.
Hence, $\wt\bt$ is consistent for $\bt$.

Finally, we derive the asymptotic distribution of $\wt\bt$.
By the definition of $\wt\bt$ and \eqref{eq:ee4},
\bse
\0
&=&\sqrt{n_1}\A^*n^{-1}\sumi\left[\frac{r_i}{\pi}\rho^*(y_i)
\left\{\U(\x_i,y_i,\wt\bt)-\b^*(\x_i,\wh\a^*,\wh\E_p,\wt\bt)\right\}
+\frac{1-r_i}{1-\pi}\b^*(\x_i,\wh\a^*,\wh\E_p,\wt\bt)\right]\\
&=&n^{-1/2}\sumi\left[\sqrt{\pi}\bphi\eff^*(\x_i,r_i,r_iy_i,\bt)
+\frac{r_i}{\sqrt{\pi}}\A^*
\{\b^*(\x_i,\a^*,\Ep,\bt)-\U(\x_i,y_i,\bt)\}
\{\rho^*(y_i)-\rho(y_i)\}\right]\\
&&+\wh\B\sqrt{n_1}(\wt\bt-\bt)+o_p(1),
\ese
where
\bse
\wh\B
&=&\E\left\{\frac{\partial\bphi\eff^*(\X,R,RY,\bt)}{\partial\bt\trans}\right\}
+\A^*\Ep\left[\frac{\partial}{\partial\bt\trans}
\E\{\b^*(\X,\a^*,\Ep,\bt)-\U(\X,Y,\bt)\mid Y\}\{\rho^*(Y)-\rho(Y)\}\right]+o_p(1)\\
&=&\A^*\frac{\partial}{\partial\bt\trans}
\E\left[\frac{R}{\pi}\rho^*(Y)\{\U(\X,Y,\bt)-\b^*(\X,\bt)\}
+\frac{1-R}{1-\pi}\b^*(\X,\bt)\right]+o_p(1)\\
&=&\A^*\frac{\partial}{\partial\bt\trans}\Eq\left\{\U(\X,Y,\bt)\right\}+o_p(1)\\
&=&\A^*\A^{-1}+o_p(1),
\ese
because $\wt\bt$ is consistent for $\bt$
and $\E\{\b^*(\X,\a^*,\Ep,\bt)\mid y\}=\E\{\U(\X,y,\bt)\mid y\}$.
Therefore, $\sqrt{n_1}(\wt\bt-\bt)$ converges in distribution to
$N\left\{\0,\A\A^{*-1}\bSigma(\A\A^{*-1})\trans\right\}$ as $n_1\to\infty$.
\qed

\end{document}